\documentclass[10.5pt,a4paper]{article}
\usepackage{xy,amsmath,amssymb,slashed,url,bm,textgreek,upgreek}
\usepackage[table]{xcolor}
\usepackage{bbm,array,amsfonts,wrapfig,lscape,float,mathtools,multirow,longtable}
\usepackage{graphicx}
\usepackage{epstopdf}
\usepackage{comment}
\usepackage{braket}
\usepackage{tikz-cd}
\usepackage[safe]{tipa}
\usepackage{stmaryrd}

\usepackage{newtxtext,newtxmath}
\usepackage{xparse}

\NewDocumentCommand{\dslash}{s}{%
  \IfBooleanTF{#1}
    {\big/\mkern-7mu\big/}
    {/\mkern-6mu/}%
}

\def\be{\begin{equation}}
\def\ee{\end{equation}}

\def\tilde{\widetilde}

\def\frak{\mathfrak}

\def\R{{\mathbb R}}
\def\C{{\mathbb C}}

\def\[{\bigl [}

\def\]{\bigr ]}

\def\tr{{\mathrm {tr}}}
\def\Z{{\mathbb Z}}

\def\symbg{\textbf{Sym}\left(\Sigma,\textbf{B}{\cal G}\right)}

\def\tilde{\widetilde}

\def\bar{\overline}

\def\suu{\String{U(1)_k}}

\font\teneurm=eurm10 \font\seveneurm=eurm7  \font\fiveeurm=eurm5
\newfam\eurmfam
\textfont\eurmfam=\teneurm \scriptfont\eurmfam=\seveneurm
\scriptscriptfont\eurmfam=\fiveeurm

\font\tencmmib=cmmib10 \skewchar\tencmmib='177
\font\sevencmmib=cmmib7 \skewchar\sevencmmib='177
\font\fivecmmib=cmmib5 \skewchar\fivecmmib='177

\usepackage{amsthm}

\newcommand{\shape}{
  \raisebox{1pt}{\rm\normalfont\textesh}
}

\usepackage{mathalpha}
\usepackage{tensor}

\newcommand{\aut}[1]{\textrm{AUT}\left({#1}\right)}
\newcommand{\out}[1]{\textrm{OUT}\left({#1}\right)}

\newcommand{\String}[1]{\text{String}({#1})}

\theoremstyle{definition}

\newtheorem{definition}{Definition}[section]
\newtheorem*{fact*}{Fact}

\newtheorem*{propo*}{Proposition}
\newtheorem{propo}{Proposition}[section]

\newtheorem{remark}{Remark}[section]

\title{
\fontsize{18pt}{20pt}\selectfont\bfseries
 Higher-form symmetries as higher automorphism bundles
}

\author[a]{Alonso Perez-Lona}

\affiliation[a]{Department of Physics MC 0435, 850 West Campus Drive, Virginia Tech, Blacksburg, VA 24061}

\emailAdd{aperezl@vt.edu}

\abstract{The notion of global higher-form symmetries has received much attention, but leaves room for a more systematic mathematical formulation. In this article, we highlight the concept of higher automorphism bundles from the field of higher categorical differential geometry and higher gauge theory,
and we demonstrate that this neatly reproduces and clarifies many examples and phenomena discussed in the literature. We rigorously construct the higher-form symmetries of pure gauge theory of a general strict Lie $2$-group, featuring center higher-form symmetries. We then apply this explicitly to several physically-relevant examples, such as $U(1)$ bundles, bundle gerbes, and certain string $2$-groups related to $SU(n)$ instanton restriction, and $5d$ supergravity. We elaborate on the nontrivial interplay between global higher-form symmetries, connection data, and symmetry gauging.}
\begin{document}
\setlength{\abovedisplayskip}{2pt}
\setlength{\belowdisplayskip}{2pt}
\tikzset{Rightarrow/.style={double equal sign distance,>={Implies},->},
triple/.style={-,preaction={draw,Rightarrow}},
quadruple/.style={preaction={draw,Rightarrow,shorten >=0pt},shorten >=1pt,-,double,double
distance=0.2pt}}

\maketitle

\flushbottom

\section{Introduction}\label{sec:intro}

The notion of symmetry, essential for Physics as a whole, has received renewed attention due to \cite{Gaiotto:2014kfa}. There, symmetries in Quantum Field Theory were reconceptualized as topological quantum operators, which are quantum operators supported at submanifolds of spacetime in a way that is invariant under deformations of such submanifolds. One motivation comes from Noether currents, which describe what are now known as \textit{invertible 0-form} (smooth) global symmetries. The Noether current, in the case the symmetry generated can be integrated to an action by the group $U(1)$, for example, is a differential $(n-1)$-form $j$ that is closed \textit{on-shell}
\begin{equation}
    d\,j\vert_{\rm on-shell} =0.
\end{equation}
One posits the existence of a quantum operator in \textit{correspondence} with the function
\begin{equation}
    U_{\alpha}\left(S^{(d-1)}\right) \longleftrightarrow \exp\left(\alpha\oint_{S^{(d-1)}} j \right),
\end{equation}
for $\alpha$ the symmetry parameter, supported at a closed $(d-1)$-dimensional submanifold $S^{(d-1)}\subset M$ of spacetime. The quantum operator is called  \textit{topological} in the sense that
\begin{equation}
    U_{\alpha}\left(S^{(d-1)}\right) = U_{\alpha}\left(S'^{(d-1)}\right)
\end{equation}
whenever $S^{(d-1)},S'^{(d-1)}$ are cobordant, as a consequence of $j$ being closed on-shell. 

The generalization proposed in \cite{Gaiotto:2014kfa} features \textit{higher-form} symmetries as topological operators supported on submanifolds of \textit{higher} \textit{co}dimension. In particular, smooth symmetries are thought to give rise to \textit{higher Noether currents} (as rigorously derived in \cite{Sati:2015yda} cf. \cite{schreiber-noether}).

This perspective has led to significant insights in QFT, as can be seen from e.g. \cite{Cordova:2022ruw,Bhardwaj:2023kri,Luo:2023ive} and references therein. However, from the mathematical perspective, there is a significant drawback. The envisioned correspondence between higher quantum operators, on the one hand, and functions/classical observables, on the other hand (and defects, on the third), relies on the path-integral formulation \cite[p.2]{Gaiotto:2014kfa}. As a consequence, it is seldom clear how to construct the actions that should define quantum operators. Further, the operators thought to be ``charged'' under these generalized symmetries are not only supposed to be \textit{nonlocal} but generally non-topological. Unlike the case of QFT with only local operators (e.g. \cite{Haag:1963dh,Henriques:2016ajg,Costello:2023knl}), currently no definition of QFT with extended non-topological operators exists. Some mathematical proposals and frameworks have therefore emerged to ameliorate these issues, including \cite{Gripaios:2022yjy,Gripaios:2023ups,Gwilliam:2025vdu}, which require a complete axiomatization of the field theory of interest, and \cite{Carqueville:2017aoe,Carqueville:2023qrk,Freed:2022qnc}, which focus on the topological, as opposed to smooth, aspect, of symmetries. 

In this paper, we explain how certain \textit{invertible higher-form symmetries}, in great generality, admit a description as \textit{automorphisms} in the context of \textit{smooth higher geometry}, without needing a full characterization of the corresponding \textit{fully quantum} theory. We furthermore discuss the role of connections for these symmetries. This formulation gives a straightforward understanding of \textit{gauging}. After presenting the general proposal in Section~\ref{sec:general}, we construct the explicit symmetry action for principal $2$-group bundles in Section~\ref{sec:2-action}. In subsequent sections, we apply this to concrete examples of pure (higher) gauge theory, discussing the incorporation of connections in several cases. The general outline of these examples consists on first computing the higher-form symmetries described by higher automorphisms, then describing the action on the ``charged'' objects, often involving connections, and finally describing the gauging of these symmetries. In Section~\ref{sec:puregauge}, we discuss how known results of pure $G=U(1)$ gauge theory fit our general framework. In Section~\ref{sec:bfields}, we analyze the simplest case of a smooth $2$-group that is not a Lie group, that of B fields as bundle gerbes. Finally, in Section~\ref{sec:string2groups}, we derive specific higher-form symmetries of pure gauge theory of some String $2$-groups as nontrivial examples of principal $2$-group bundles, explaining the subtleties that can arise when incorporating connections into the discussion. Section~\ref{sec:conclusion} mentions several directions of further work. Appendix~\ref{app:adjbundles} is a compilation of mathematical definitions we use to describe principal bundles of smooth $2$-groups with nontrivial connections. Appendix~\ref{app:2grpds} collects definitions and results on the theory of $2$-groupoids used for the construction presented in Section~\ref{sec:2-action}. Appendix~\ref{app:string2gps} summarizes how String $2$-groups are understood as higher central extensions, and reviews their smooth Drinfeld centers. Appendix~\ref{app:principal-3-bundles} collects the Čech data necessary to define principal bundles of certain smooth $3$-groups (without connection).

We note that the interplay of higher group automorphisms and symmetries of geometric structures has also been highlighted in \cite{Bunk:2023sov}, where the authors elaborate in particular on the extension of spacetime diffeomorphisms by higher group automorphisms, and provide some suggestions as to how these extensions could be related to the higher-form symmetries of concrete examples of physical relevance.

\section{General statement}\label{sec:general}

\paragraph{$\sigma$-models and symmetries.} The fields of a $\sigma$-model, by definition, are smooth maps $\phi:\Sigma\to X$ from an $n$-dimensional spacetime manifold $\Sigma$ to some other smooth space $X$. The symmetries are generated by a subgroup $H\leq \text{Diff}(X)$ of the diffeomorphism group $\text{Diff}(X)$ of the target space. The subgroup $H$ is determined by the dynamics proper to the $\sigma$-model. For example, the most common case is the dynamics dictated by the kinetic action
\begin{equation}
    S = \int_{\Sigma}  \, g_{ij}(\phi(x))\, d\phi^i\wedge\star d\phi^j,
\end{equation}
for $g$ a metric on $X$. The corresponding subgroup is $\text{Isom}(X,g)\leq \text{Diff}(X)$,  the group of isometries of $(X,g)$. 

In this sense, ${\rm Diff}(X)$ gives the most general possible symmetries coming from transformations on $X$. In the context of generalized symmetries of $\sigma$-models, these have sometimes been called \textit{electric symmetries} \cite{Sheckler:2025rlk}. In what follows we will work with ${\rm Diff}(X)$, which afterwards can be restricted to the relevant subgroup determined by the specific dynamics.

The global ${\rm Diff}(X)$ symmetries are more precisely \textit{parameterized} by \textit{flat} smooth maps
\begin{equation}\label{eq:diffglobalsymmetry}
    \text{Map}_{\flat}\left(\Sigma, \text{Diff}(X)\right) = \{ f:\Sigma \to \text{Diff}(X) \vert \, df=0\}.
\end{equation}
The composition law of $\text{Diff}(X)$ induces a composition law on $\text{Map}_{\flat}(\Sigma, \text{Diff}(X))$. In the (frequent) case $\Sigma$ is connected, the smooth collection of such locally-constant maps is simply $\text{Diff}(X)$ itself, equipped with its composition law, which is why one can think of the symmetry as parameterized by the Lie group itself. 

Concretely, this is a symmetry of the $\sigma$-model because it acts on the \textit{smooth space} of fields 
\begin{equation}\label{eq:osigmamodel}
    \text{Fields}(\Sigma):= \text{Map}\left(\Sigma,X\right),
\end{equation}
by composition. In other words, the smooth space of fields becomes a \textit{module} of
\begin{equation}
    \text{Sym}(\Sigma,X):= {\rm Map}_{\flat}(\Sigma,{\rm Diff}(X)).
\end{equation}

This group controls the symmetries of the $\sigma$-model ${\rm Map}(\Sigma,X)$, in the sense that symmetries described by some other Lie group $G$ are defined via a homomorphism
\begin{equation}\label{eq:gaction}
    \alpha: G \to {\rm Diff}(X),
\end{equation}
or, more precisely, a diagram
\begin{equation}
    \begin{tikzcd}
                               &  & G \arrow[dd, "\alpha"] \\
                               &  &                        \\
\Sigma \arrow[rruu] \arrow[rr] &  & \text{Diff}(X)        
\end{tikzcd},\label{eq:preliminary-global-symmetry}
\end{equation}
where the maps $\Sigma \to G$, $\Sigma\to {\rm Diff}(X)$ are flat maps. Since $X$ is a $\text{Diff}(X)$-module, then a homomorphism (\ref{eq:gaction}) makes $X$ a $G$-module. Note, in particular, that the action homomorphism (\ref{eq:gaction}) does \textit{not} need to be injective, meaning its \textit{kernel acts trivially} on $X$, and thus on the field theory. 

Now, given a global symmetry (\ref{eq:preliminary-global-symmetry}), one is often interested in \textit{gauging} it. Nevertheless, except for very restricted cases, smooth manifolds are \textit{insufficient} for describing the fields of the $G$\textit{-gauged} theory. The $G$-gauged $\sigma$-model has fields
\begin{equation}\label{eq:ex-fields}
    {\rm Fields}_G = {\rm Map}(\Sigma, X\dslash G),
\end{equation}
where $X\dslash G$ is the \textit{smooth homotopy quotient}, which can be thought of as $X$ equipped with \textit{equivalences} between points related by the $G$-action. For example, for $G$ a finite group, this homotopy quotient is known as an \textit{orbifold} \cite{satake1956generalization}, a certain generalization of a manifold. From a physical point of view, the fields (\ref{eq:ex-fields}) valued in the smooth homotopy quotient $X\dslash G$ incorporates the so-called \textit{twisted sectors} in 2d QFT (cf. \cite{Dixon:1985jw,Pantev:2005wj,Hellerman:2006zs}). In particular, the original fields (\ref{eq:osigmamodel}) become \textit{equivalent} to their image under the $G$-action in the $G$-gauged theory (\ref{eq:ex-fields}).

\paragraph{$\sigma$-models of smooth higher stacks.} An insight of \cite{Schreiber:2013pra} is that a wide variety of physically-relevant field theories can be realized as $\sigma$-models, provided the context is sufficiently general. The smooth maps of (\ref{eq:osigmamodel}), for example, take place in the category of (infinite-dimensional) smooth manifolds. In the present article, by contrast, the ``context'' we will work in is the \textit{cohesive $\infty$-topos}
\begin{equation}
    \textbf{H}:= {\rm Sh}\left(\text{CartSp},\infty{\rm Grpd} \right)
\end{equation}
of sheaves on the site of cartesian spaces valued in $\infty$-groupoids. This is understood as the ``context'' of \textit{smooth} $\infty$-groupoids, or \textit{smooth higher stacks}, for simplicity. Among the advantages of \textbf{H} for our purposes are the existence of:
\begin{enumerate}
    \item \textit{smooth higher groups} $\cal G$ (and their actions),
    \item \textit{moduli stacks of principal $\cal G$ bundles} $\textbf{B}{\cal G}$, and, with some subtleties, \textit{moduli stacks of principal $\cal G$ bundles with connections} $\textbf{B}_{\nabla}{\cal G}$, where maps $\Sigma\to \textbf{B}_{\nabla}{\cal G}$ correspond to a \textit{particular}\footnote{As opposed to an equivalence class thereof.} connection on a principal $\cal G$ higher bundle over $\Sigma$,
    \item and \textit{quotient stacks} $\cal X\dslash G$ which always exist without imposing additional conditions on the $\cal G$ action.
\end{enumerate}

By ``field theory'' we are not claiming a \textit{fully quantum} field theory, but at most a \textit{pre}quantum theory yet to undergo quantization. A completely general quantization process is yet to be constructed, but see \cite{nuiten2013cohomological,bongers2014geometric,Schreiber:2014xva}. As mentioned in Section~\ref{sec:intro}, the fact that a fully general quantization process is still missing serves as a motivation for the present paper to rigorously discuss global higher-form symmetries without referring to quantum aspects of field theories.

Generalizing (\ref{eq:osigmamodel}) is conceptually straightforward. Still assuming $\Sigma$ is an $n$-dimensional smooth manifold, we now consider the $\sigma$-model with fields
\begin{equation}
    \textbf{Fields}:= \textbf{H}(\Sigma,\mathcal{X})
\end{equation}
for ${\cal X}$ a higher smooth stack, an object of \textbf{H}. Note that \textit{the collection} \textbf{Fields} \textit{of fields is a  higher smooth stack itself,} and therefore a proper definition of symmetry should respect this structure.

Similarly, the symmetries are controlled by the smooth higher stack
\begin{equation}
{\rm Aut}\left({\cal X}\right)\subset \textbf{H}\left({\cal X,X}\right)    ,
\end{equation}
the \textit{automorphism $\infty$-group} of $\cal X$. 

From the point of view of the field theory on $\Sigma$, we propose that the \textbf{global higher-form symmetries} are parameterized by the \textit{higher smooth stack of }\textit{flat maps}
\begin{equation}\label{eq:flatsymmetryparameters}
    \textbf{Sym}(\Sigma,{\cal X}):=\textbf{H}_{\flat}(\Sigma,{\rm Aut}({\cal X})) = \textbf{H}\left(\shape_n \left(\Sigma\right), \text{Aut}({\cal X})\right) = \textbf{H}\left( \Sigma, \flat {\rm Aut}\left( {\cal X} \right) \right),
\end{equation}
where $\shape_n\left( \Sigma \right)$ is the $n$-groupoid obtained from the smooth path $n$-groupoid $\mathcal{P}_n\Sigma$ of $\Sigma$ by dividing out full homotopy of $n$-disks $\Delta^n\to\Sigma$, relative boundary (see e.g. \cite[Proposition 1.2.122]{Schreiber:2013pra} cf. \cite{SW11,Schreiber:2008kcv}), and $\flat {\rm Aut}\left({\cal X}\right)$ is the image of the smooth stack ${\rm Aut}\left({\cal X}\right)$ under the flat modality (cf. \cite[Observation 5.2.27]{Schreiber:2013pra}). 

The stack of global higher-form symmetries comes equipped with a canonical action provided by the \textit{evaluation map}
\begin{equation}
    \textbf{ev}: {\rm Aut}({\cal X})\times {\cal X} \to {\cal X}
\end{equation}
  \begin{equation}\label{eq:evaluation-diagram-general}
    \begin{tikzcd}[
      column sep={between origins,80pt}
    ]
      \mathbf{H}\big(
        \Sigma
        ,\,
        \flat \mathrm{Aut}({\cal X})
      \big)
      \times
      \mathbf{H}\big(
        \Sigma
        ,\,
        {\cal X}
      \big)
      \ar[
        rr,
        "{
          (\mbox{-})\cdot(\mbox{-})
        }"
      ]
      \ar[
        dr,
        "{ \sim }"{sloped}
      ]
      &[-30pt]&
      \mathbf{H}\big(
        \Sigma
        ,\,
        {\cal X}
      \big)
      \mathrlap{\,.}
      \\
      &
      \mathbf{H}\big(
        \Sigma        ,\,
       \left(\flat{\rm Aut}({\cal X})\right)
        \times
        {\cal X}
      \big)
      \ar[
        ur,
        "{
          \mathbf{H}(\Sigma,\textbf{ev})
        }"{sloped}
      ]
    \end{tikzcd}
  \end{equation}

Technically, the maps (\ref{eq:flatsymmetryparameters}) describe \textit{flat principal smooth $\infty$-groupoid bundles}. Indeed, as highlighted in \cite[Remark 5.1.149]{Schreiber:2013pra}, for $\alpha: * \to {\rm Aut}\left( {\cal X} \right)$, $s:* \to \Sigma$, and $f:\Sigma \to {\rm Aut}\left( {\cal X}\right)$ the smooth homotopy fiber is
\begin{equation}
\begin{tikzcd}
\Omega_{\alpha} {\rm Aut}\left( {\cal X} \right) \arrow[rr] \arrow[dd] &  & P \arrow[rr] \arrow[dd] &  & * \arrow[dd, "\alpha"]           \\
                                                                       &  &                         &  &                                  \\
* \arrow[rr, "s"']                                                     &  & \Sigma \arrow[rr, "f"'] &  & {\rm Aut}\left( {\cal X} \right)
\end{tikzcd},
\end{equation}
which shows ${\rm Aut}\left({\cal X}\right)$ acts as a moduli stack for the higher smooth stack defined by the looping $\Omega_z{\rm Aut}\left({\cal X}\right)$ at $z$.

Importantly, flat principal smooth $\infty$-groupoid bundles are equivalent to principal smooth $\infty$-groupoid bundles \textit{with flat connection} \cite{Schreiber:2013pra}.

The information encoded by an object in $\textbf{Sym}(\Sigma,{\cal X})$ is perhaps most easily understood in terms of \textit{Čech data}, of which we will make frequent use in this paper. For this, we first replace $\Sigma$ with a cofibrant object, the smooth higher stack $U(\Sigma)$ provided by a good cover $\{ U_i \}_{i\in{\cal I}}$. This stack can be presented in terms of \textit{Kan complexes} as
\begin{itemize}
    \item objects $(x \in U_i, i\in {\cal I})$,
    \begin{equation}
        \begin{tikzcd}
            x_{i}
        \end{tikzcd}
    \end{equation}
    \item $1$-morphisms $(x \in U_i\cap U_j, i, j): (x,i) \to (x,j)$,
    \begin{equation}
        \begin{tikzcd}
            x_i
            \ar[rr, "{x_{ij}}"] && x_j
        \end{tikzcd}
    \end{equation}
    \item $2$-morphisms $(x\in U_i\cap  U_j \cap U_j, i,j,k):(x,j,k)\circ (x,i,j) \to (x,i,k)$,
    \begin{equation}
        \begin{tikzcd}
                                                 &    & x_k                                                         \\
                                                 & {} &                                                             \\
x_i \arrow[rr, "x_{ij}"'] \arrow[rruu, "x_{ik}"] &    & x_j \arrow[uu, "x_{jk}"'] \arrow[lu, "x_{ijk}", Rightarrow]
\end{tikzcd}
    \end{equation}
    \item $3$-morphisms $(x,i,j,k,l)$:
    \begin{equation}
        \begin{tikzcd}
x_j \arrow[rr, "x_{jk}"] \arrow[rd, "x_{ijk}" description, Rightarrow]            &    & x_k \arrow[dd, "x_{kl}"] \arrow[ldd, "x_{ikl}" description, Rightarrow] &                          &    & x_j \arrow[rr, "x_{jk}"] \arrow[rrdd, "x_{jl}" description] \arrow[rdd, "x_{ijl}" description, Rightarrow] &    & x_k \arrow[dd, "x_{kl}"] \arrow[ld, "x_{jkl}" description, Rightarrow] \\
                                                                                  & {} &                                                                         & {} \arrow[r, "x_{ijkl}",triple] & {} &                                                                                                            & {} &                                                                        \\
x_i \arrow[uu, "x_{ij}"] \arrow[rr, "x_{il}"'] \arrow[rruu, "x_{ik}" description] & {} & x_l                                                                     &                          &    & x_i \arrow[uu, "x_{ij}"] \arrow[rr, "x_{il}"']                                                             & {} & x_l                                                                   
\end{tikzcd},
    \end{equation}
    \item $\cdots$
\end{itemize}

Then, an object
\begin{equation}
    {\cal F}: U(\Sigma) \to {\rm Aut}({\cal X}),
\end{equation}
in $\textbf{Sym}\left(\Sigma, {\cal X}\right)$ is a functor of smooth $\infty$-groupoids, a flat assignment :
\begin{itemize}
    \item to each patch $U_i$ a (weakly) invertible smooth functor $F_i: {\cal X} \to {\cal X}$ 
    \begin{equation}
        \begin{tikzcd}
            F_{i}
        \end{tikzcd}
    \end{equation}
    \item to each intersection $U_{ij}:=U_i \cap U_j$ a smooth natural isomorphism $\eta_{ij}:F_i \to F_j$
    \begin{equation}
        \begin{tikzcd}
            F_i
            \ar[rr, "{\eta_{ij}}"] && F_j
        \end{tikzcd}
    \end{equation}
    \item to each triple intersection $U_{ijk}:=U_i\cap U_j \cap U_k$ a smooth modification \cite{JohnsonYau2021} $m_{ijk}:\eta_{jk}\circ \eta_{ij} \to \eta_{ik}$
    \begin{equation}
        \begin{tikzcd}
                                                 &    & F_k                                                         \\
                                                 & {} &                                                             \\
F_i \arrow[rr, "\eta_{ij}"'] \arrow[rruu, "\eta_{ik}"] &    & F_j \arrow[uu, "\eta_{jk}"'] \arrow[lu, "m_{ijk}", Rightarrow]
\end{tikzcd}
    \end{equation}
    \item to each quadruple intersection $U_{ijkl}:=U_i\cap U_j \cap U_k \cap U_l$ a smooth perturbation \cite{gurski2006algebraic} $\sigma_{ijkl}:\eta_{ikl}\eta_{ijk} \Rrightarrow \eta_{ijl}\eta_{jkl}$
    \begin{equation}
        \begin{tikzcd}
F_j \arrow[rr, "\eta_{jk}"] \arrow[rd, "m_{ijk}" description, Rightarrow]            &    & F_k \arrow[dd, "\eta_{kl}"] \arrow[ldd, "m_{ikl}" description, Rightarrow] &                          &    & F_j \arrow[rr, "\eta_{jk}"] \arrow[rrdd, "\eta_{jl}" description] \arrow[rdd, "m_{ijl}" description, Rightarrow] &    & F_k \arrow[dd, "\eta_{kl}"] \arrow[ld, "m_{jkl}" description, Rightarrow] \\
                                                                                  & {} &                                                                         & {} \arrow[r, "\sigma_{ijkl}",triple] & {} &                                                                                                            & {} &                                                                        \\
F_i \arrow[uu, "\eta_{ij}"] \arrow[rr, "\eta_{il}"'] \arrow[rruu, "\eta_{ik}" description] & {} & F_l                                                                     &                          &    & F_i \arrow[uu, "\eta_{ij}"] \arrow[rr, "\eta_{il}"']                                                             & {} & F_l                                                                   
\end{tikzcd},
    \end{equation}
    \item and higher smooth \textit{transfors} \cite{crans2003localizations,crans1999tensor} for higher intersections assigned in the same consistent manner.
\end{itemize}

With this proposed generalization at hand, we can define an \textit{action} of some \textit{smooth higher group} $\cal G$, a group object in \textbf{H}, on the field theory via a homomorphism
\begin{equation}
    \alpha: {\cal G} \to {\rm Aut}({\cal X}),
\end{equation}
and these global higher-form symmetries are parameterized by
\begin{equation}
    \begin{tikzcd}
                               &    & \cal G \arrow[dd, "\alpha"] \arrow[ldd, "\eta", Rightarrow] \\
                               &    &                                                             \\
\Sigma \arrow[rruu] \arrow[rr] & {} & {\rm Aut}({\cal X})                                        
\end{tikzcd},
\end{equation}
where now the diagram commutes \textit{up to a smooth homotopy}, and by (\ref{eq:flatsymmetryparameters}) the morphisms from $\Sigma$ are flat.

This perspective allows us to describe the \textit{gauging} of these symmetries, as follows. The fields of the $\cal G$-gauged theory are
\begin{equation}
    \textbf{Fields}_{\cal G}:= \textbf{H}\left(\Sigma, {\cal X\dslash G}\right),
\end{equation}
where $\cal X\dslash G$ is the \textit{quotient smooth higher stack} fitting in the sequence
\begin{equation}\label{eq:gauging-fibration}
    {\cal X} \to  {\cal X\dslash G}\to \textbf{B} {\cal G}
\end{equation}
for $\textbf{B} {\cal G}$ the \textit{delooping} of the higher smooth group $\cal G$. Unlike the case of smooth manifolds, the quotient $\cal X\dslash G$ \textit{always exists} as a higher smooth stack, independently of the properties of the $\cal G$-action on $\cal X$.

An immediate consequence of this description of gauging is that the original stack of fields maps into (but does \textit{not} equal)
\begin{equation}
    \textbf{Fields} \to \textbf{Fields}_{\cal G}
\end{equation}
the stack of fields of the $\cal G$-gauged theory, in such a way that any field becomes \textit{equivalent} to any of its images under the $\cal G$-action, as befits a quotient up to isomorphism.

An example of this kind was explored in \cite{Perez-Lona:2023llv}, where ${\cal X}=X$ is a smooth manifold, and $\cal G$ a finite $2$-group.

\paragraph{Higher gauge theory.} Of special importance are the higher smooth stacks ${\cal X}=\textbf{B}{\cal H}$ corresponding to the delooping $\textbf{B}\cal H$ of some $n$-truncated smooth $\infty$-group  $\cal H$, equivalently a smooth $n$-group. The higher smooth stack $\textbf{B}\cal H$ is the moduli stack of principal $\cal H$ bundles. In this case, there is a \textit{truncation}\footnote{Since ``truncation'' commonly refers to the operation $\tau_{\leq n}: \infty{\rm Grpd}\to n{\rm Grpd}\to \infty{\rm Grpd}$, Eq.~(\ref{eq:aut-trunaction-general}) is technically an equivalence of objects in \textbf{H} between the $n$-truncation $\tau_{\leq n}{\rm Aut}\left( \textbf{B}{\cal H}\right)$, and $\aut{\cal H}$, an object in \textbf{H} that is already $n$-truncated.}
\begin{equation}\label{eq:aut-trunaction-general}
    \tau_{\leq n}{\rm Aut}({\textbf{B}{\cal H}}) = \aut{\cal H}
\end{equation}
of ${\rm Aut}({\textbf{B}{\cal H}})$, the smooth \textit{automorphism} $(n+1)$\textit{-group} $\aut{\cal H}$ of $\cal H$. This was introduced in \cite{nikolaus2015principal} to provide a cohomological classification of $\cal H$-gerbes. In more detail, a smooth morphism
\begin{equation}
    f: \Sigma \to \textbf{B}\aut{\cal H}
\end{equation}
gives rise to a \textit{fiber} bundle over $\Sigma$ whose fiber is the delooping $\textbf{B}{\cal H}$. This is a generalization of Giraud's nonabelian cohomology as a classification of gerbes \textit{with a band} \cite{giraud}.

The automorphism smooth $(n+1)$-group $\aut{\cal H}$ fits in the exact sequence of smooth $(n+1)$-groups
\begin{equation}\label{eq:higheraut}
    \textbf{B}{\cal Z(H)} \to \aut{\cal H}\to \out{\cal H},
\end{equation}
Here, $\out{\cal H}$ is the smooth $n$-group of outer automorphisms (regarded as an $(n+1)$-group with trivial information at the top degree), and, in Giraud's terms, provides the band of the $\cal H$-gerbe. Further, $ \textbf{B}{\cal Z(H)}$ is the delooping of the \textit{center} of $\cal H$. This is ``abelian'' in the sense of smooth $(n+1)$-groups, and is also often referred to as a \textit{braided} $(n+1)$-group \cite{Schreiber:2013pra}. The $n=1$ case this is the familiar automorphism $2$-group. The $n=2$ case is explored in \cite{SR07}. 

While for general smooth $n$-groups, the automorphism $(n+1)$-group and its $n$-group of outer automorphisms become more complicated, one should note the inclusion of smooth $(n+1)$-groups $\textbf{B}{\cal Z}({\cal H})\to \text{AUT}({\cal H})$, which tells us that among the global higher symmetries there are those parameterized by  principal ${\cal Z}({\cal H})$ (higher) bundles with a \textit{flat} connection over $\Sigma$. As we will see exemplified, the canonical action of these center higher bundles provided by the evaluation map (\ref{eq:evaluation-diagram-general}) corresponds to the \textbf{tensoring} of flat principal ${\cal Z}({\cal H})$ bundles with principal ${\cal H}$ bundles:
\begin{equation}\label{eq:tensorbundles}
    \left( \textit{flat} \, \text{principal} \, {\cal Z}({\cal H}) \, \text{bundles} \right) \otimes \left( \text{principal} \, {\cal H} \, \text{bundles}  \right)  \to \left( \text{principal} \, {\cal H} \, \text{bundles}  \right),
\end{equation}
which is realized by smooth higher \textit{functors} $\textbf{Fields}\to\textbf{Fields}$.
That these center higher bundles admit a product crucially depends on the fact that ${\cal Z}({\cal H})$ is, by definition, a braided smooth $n$-group, so that $\textbf{B}{\cal Z}({\cal H})$ is itself a group object in \textbf{H}. The well-definedness of the action, on the other hand, depends on ${\cal Z}({\cal H})$ being central in $\cal H$. We refer to this subsymmetry simply as the \textbf{center higher-form symmetry.}

Sometimes, the center higher-form symmetry of higher principal bundles can be \textit{promoted} to the \textit{connective} picture in an obvious way, namely, a symmetry of higher principal bundles \textit{with connection}. In such a case, the action (\ref{eq:tensorbundles}) becomes a tensoring action of principal ${\cal Z}({\cal H})$ bundles with flat connection with principal $\cal H$ bundles \textit{with connection}, not necessarily flat,
\begin{gather}\label{eq:tensorbundles-wconn}
    \left( \, \text{principal} \, {\cal Z}({\cal H}) \, \text{bundles w/ \textit{flat} conn.} \right) \otimes \left( \text{principal} \, {\cal H} \, \text{bundles w/ conn.}  \right) \\ \to \left( \text{principal} \, {\cal H} \, \text{bundles w/ conn.}  \right) ,\nonumber
\end{gather}
which may be called the \textit{connective} center higher-form symmetry.

This is, for example, what is called the $U(1)$ $1$-form symmetry of $U(1)$ gauge theory (e.g. \cite{Gaiotto:2014kfa, Schafer-Nameki:2023jdn}, cf. Section~\ref{sec:puregauge}). However, as concretely exemplified in Section~\ref{sec:string2groups}, one needs to be more careful when doing gauge theory of higher groups, the main reason being that \textit{nontrivial} connections of truly nonabelian higher bundles require additional choices of structure, called \textit{adjustments} \cite{FSS12,Schreiber:2013pra,Rist:2022hci,Gagliardo:2025oio}, which are not necessarily respected by the action (\ref{eq:tensorbundles}) \textit{strictly}. Therefore, in general, one ought to distinguish between symmetries of (higher) principal bundles, and of principal bundles \textit{with connection}. It is important to note, however, that the acting stack (\ref{eq:flatsymmetryparameters}) is the same since, as mentioned above, flat higher principal bundles are equivalently higher principal bundles with flat connection. The difference stems solely from the introduction of generally non-flat connections on the acted bundles.

In Field theory, $\sigma$-models with stacky target spaces $\cal X$ are by now frequent. A common scenario, alluded to above, is ${\cal X}=\textbf{B}_{\nabla}{\cal H}$ the moduli stack of principal $\cal H$ bundles with connection, for $\cal H$ some smooth ($\infty$-)group. These are the natural target spaces for fields that satisfy certain Bianchi identities/quantization conditions. These field theories are thus better understood as \textit{pre}quantum field theories (e.g. \cite{Sati:2023mta}). 

A streamlined approach we will take in the present note to the symmetries these field theories with fields $\Sigma \to \textbf{B}_{\nabla}{\cal H}$ is as follows:
\begin{enumerate}
    \item First, we will compute the automorphisms of the object $\textbf{B}{\cal H}$.  These are the symmetries of principal $\cal H$ bundles, without connection.
    \item Then, we discuss a possible refinement to connective symmetries, which, by twisting the connections appropriately, are automorphisms of $\textbf{B}_{\nabla}{\cal H}$. This, in particular, indicates what are the ``charged'' objects under such transformations.

    \item Afterwards, for a given symmetry ${\cal G}\to \aut{\cal H}$, we compute the weak quotient $(\textbf{B}{\cal H})\dslash {\cal G}\cong \textbf{B}\cal K$ for $\cal K$ some other smooth $\infty-$group (in the examples we consider, $\cal K$ can be quotient or a split extension of $\cal H$).
 
    \item Finally, we compute the fields of the $\cal G$-gauged theory as $\Sigma\to \textbf{B}_{\nabla}{\cal K}$.
\end{enumerate}
The reader should note from the simplest example, that of ${\cal H}=1$, that the resulting fields of the $\cal G$-gauged theory are smooth maps $\Sigma\to \textbf{B}_{\nabla}{\cal G}$, principal $\cal G$-bundles with connection. Therefore, this process describes the so-called \textit{dynamical gaugings} \cite{Wang:2017loc,Hidaka:2021kkf,Karasik:2022kkq,Antinucci:2024zjp,Gagliano:2024off}. The other kind of gaugings, ``topological gaugings'', are also possible in this formalism by appropriately discretizing the otherwise smooth higher groups, connective data, and actions, but our focus is on smooth structures, and connective data.

In the following section, we rigorously construct the action for the case $\cal G$ a $2$-group, which is the basis for the examples explored in the rest of the paper.

\section{Explicit construction for principal $2$-bundles}\label{sec:2-action}

The goal of this section is to realize the general formulation of higher-form symmetries presented in Section~\ref{sec:general} for the special case of principal $\cal G$ bundles, for $\cal G$ a smooth strict $2$-group. While the propositions, in principle, follow from the abstract properties of \textbf{H}, to compute examples we need concrete formulae, which we derive here. 

According to (\ref{eq:flatsymmetryparameters}), the fields are the higher smooth stack
\begin{equation}
    \textbf{Fields} = \textbf{H}\left(\Sigma, \textbf{B}{\cal G}\right),
\end{equation}
and the global higher-form symmetries are given by
\begin{equation}\label{eq:2gpsymmetrystack}
    \textbf{Sym}\left(\Sigma, \textbf{B}{\cal G}\right) = \textbf{H}_{\flat}\left(\Sigma, {\rm Aut}(\textbf{B}{\cal G})\right),
\end{equation}
which has a composition law under which \textbf{Fields} is a module.

To realize these statements, we will model this setting in terms of maps of $2$-groupoids. For this, we replace $\Sigma$ by $U_2(\Sigma)$ the $2$-groupoid defined by a good cover $\{U_i\}_{i \in {\cal I}}$, described in general in Section~\ref{sec:general}. 

We model $\textbf{B}{\cal G}$ as the one-object \textit{strict} $2$-groupoid whose unique hom-category is $\cal G$. We will be using the \textbf{FR} convention \cite{nLabstrict2gp}. Thus, the diagram
\begin{equation}
      \begin{tikzcd}
    \bullet
    \ar[
      rr,
      bend left=40,
      "{g_1}"{description, name=s}
    ]
    \ar[
      rr,
      bend right=40,
      "{g_2}"{description, name=t}
    ]
    \ar[
      from=s, to=t,
      Rightarrow,
      "{h}"
    ]
    &&
    \bullet
  \end{tikzcd}
\end{equation}
corresponds to the equation
\begin{equation}
    g_2 = \delta(h)g_1,
\end{equation}
and compositions are computed as
\begin{equation}
  \begin{tabular}{ccc}
     \begin{tikzcd}
    \bullet
    \ar[
      rr,
      bend left=40,
      "{g_1}"{description, name=s}
    ]
    \ar[
      rr,
      bend right=40,
      "{g_2}"{description, name=t}
    ]
    \ar[
      from=s, to=t,
      Rightarrow,
      "{ h_1}"
    ]
    &&
    \bullet
    \ar[
      rr,
      bend left=40,
      "{g_1'}"{description, name=s}
    ]
    \ar[
      rr,
      bend right=40,
      "{g_2'}"{description, name=t}
    ]
    \ar[
      from=s, to=t,
      Rightarrow,
      "{ h_2}"
    ]
    &&
   \bullet
  \end{tikzcd}    &  $=$ & \begin{tikzcd}
      \bullet
    \ar[
      rrr,
      bend left=40,
      "{g_1 g_1'}"{description, name=s}
    ]
    \ar[
      rrr,
      bend right=40,
      "{g_2g_2'}"{description, name=t}
    ]
    \ar[
      from=s, to=t,
      Rightarrow,
      "{h_1\left( g_1 \triangleright h_2\right)}"
    ]
    &&&
 \bullet
  \end{tikzcd} 
  \end{tabular}.
\end{equation}

We will also model ${\rm Aut}(\textbf{B}{\cal G})$ by the sub-$2$-groupoid of invertible \textit{strict} $2$-endofunctors
\begin{equation}
    {\rm Aut}(\textbf{B}{\cal G}) \subset {\rm Hom}\left( \textbf{B}{\cal G}, \textbf{B}{\cal G}\right) \subset {\rm Lax2Fun}\left(\textbf{B}{\cal G}, \textbf{B}{\cal G}\right).
\end{equation}
Notice that what we are calling ${\rm Aut}(\textbf{B}{\cal G})$ here is not the full $2$-groupoid of invertible \textit{lax} endofunctors of the one-object $2$-groupoid $\textbf{B}{\cal G}$, but only of \textit{strict} $2$-functors ${\rm Hom}\left( \textbf{B}{\cal G}, \textbf{B}{\cal G}\right)$. Yet, because we are mainly interested in how the symmetries act from the point of view of the field theory on $\Sigma$, which are captured by maps (\ref{eq:2gpsymmetrystack}) from $\Sigma$ to ${\rm Aut}(\textbf{B}{\cal G})$, rather than by ${\rm Aut}({\textbf{B}{\cal G}})$ on its own, the end result is equivalent, even for smooth $2$-functors. This is due to the fact that $U_2(\Sigma)$ is a cofibrant replacement in the local projective model structure on simplicial sheaves \textbf{H} \cite[Example 4.3.42]{SS25-Bund}, \cite{FSS12}, \cite[Proposition 6.3.15]{Schreiber:2013pra}, so that up to equivalence strict $2$-functors capture the most general smooth $2$-functors out of $U_2(\Sigma)$. That being said, the computation of the full ${\rm Aut}({\textbf{B}{\cal G}})$ $2$-groupoid based on lax invertible endofunctors can also be physically relevant for different purposes. In \cite{Waldorf:2022lzs}, for example, certain kinds of smooth lax $2$-functors that are not necessarily strict are considered in the context of T-folds in T-duality.

With this in mind, we first compute the automorphism $2$-groupoid  ${\rm Aut}({\textbf{B}{\cal G}})$:

\begin{definition}[Automorphism $2$-groupoid of a crossed module]\label{def:aut2gp}

Let $\textbf{B}{\cal G}$ be the one-object strict $2$-groupoid defined by the delooping of a $2$-group $\cal G$ presented as a crossed module ${\cal G} = \left( G, H,\delta,\triangleright \right)$.

Its automorphism $2$-groupoid ${\rm Aut}({\textbf{B}{\cal G}})$ of strict invertible $2$-endofunctors, natural isomorphisms, and modifications is
\begin{equation}
  \label{UnderlyingGroupoidOfAutomorphism2GroupXMod}
 {\rm Aut}(\mathbf{B}{\cal G})
  =
  \left\{
  \begin{tikzcd}
    \left(\phi,\psi \right)
    \ar[
      rr,
      bend left=40,
      "{(\eta_{\bullet},\beta_{\eta})}"{description, name=s}
    ]
    \ar[
      rr,
      bend right=40,
      "{(\eta'_{\bullet},\beta_{\eta'})}"{description, name=t}
    ]
    \ar[
      from=s, to=t,
      Rightarrow,
      "{ m_{\bullet} }"
    ]
    &&
   \left(\phi',\psi' \right)
  \end{tikzcd}
  \middle\vert\;
  \begin{aligned}
    \left(\phi,\psi \right), \left(\phi',\psi' \right) & \ (\ref{XModFunctor1}), (\ref{XModFunctor2})
    \\
    \left(\eta_{\bullet},\beta_{\eta}\right), \left(\eta'_{\bullet},\beta_{\eta'}\right) & (\ref{XModTrans1})-(\ref{XModTrans3})
    \\
    m_{\bullet} &  \ (\ref{XModModif})
  \end{aligned}
  \right\}
 ,
\end{equation}
where 
\begin{itemize}
    \item the objects are the invertible homomorphisms \cite{bams/1183513797} of the crossed module $(G,H,\delta,\triangleright)$, described by pairs
\begin{equation}\label{eq:bgendofunctor}
    F=(\phi \in {\rm Aut}(G), \psi\in {\rm Aut}(H)),
\end{equation}
such that
\begin{equation}\label{XModFunctor1}
    \phi \circ \delta = \delta \circ \psi,
\end{equation}
\begin{equation}\label{XModFunctor2}
    \psi\left( g \triangleright h\right) = \phi(g) \triangleright \psi(h),
\end{equation}
\item the $1$-morphisms are natural isomorphisms
\begin{equation}
    \eta= \left(\eta_{\bullet},\beta_{\eta} \right): \left(\phi,\psi\right)\to \left(\phi',\psi'\right),
\end{equation}
which consists of a choice of group element $\eta_{\bullet}\in G$, and a function $\beta_{\eta}: G \to H$ such that
\begin{equation}\label{XModTrans1}
  \phi(f)\eta_{\bullet} =\delta\left(\beta_{\eta}(f)\right) \eta_{\bullet}\phi'(f),
\end{equation}
\begin{equation}\label{XModTrans2}
    \beta_{\eta}(fg) = \left(\phi(f)\triangleright\beta_{\eta}(g)\right)\beta_{\eta}(f),
\end{equation}
\begin{equation}\label{XModTrans3}
   \psi(h)\beta_{\eta}(f) = \beta_{\eta}(\delta(h)f) \left(\eta_{\bullet}\triangleright \psi'(h)\right),
\end{equation}
\item the $2$-morphisms are modifications
\begin{equation}
    m=m_{\bullet}: \left(\eta_{\bullet},\beta_{\eta}\right) \to \left(\eta'_{\bullet},\beta_{\eta'}\right)
\end{equation}
which is determined by a group element $m_{\bullet}\in H$ such that
\begin{equation}
    \eta_{\bullet}' = \delta(m_{\bullet})\eta_{\bullet},
\end{equation}
\begin{equation}\label{XModModif}
   \left(\phi(f)\triangleright m_{\bullet} \right) \beta_{\eta}(f) = \beta_{\eta'}(f) m_{\bullet}.
\end{equation}
\end{itemize}
\end{definition}
\begin{remark}[Delooping of Drinfeld center $\cal Z(G)$ of $\cal G$]\label{rmk:drinfeld-center}
Of particular importance is the one-object sub-$2$-groupoid of ${\rm Aut}(\textbf{B}{\cal G})$ at the identity $2$-functor 
${\rm Id}_{\textbf{B}{\cal G}}$. This $2$-groupoid consists of
\begin{itemize}
    \item A single object 
    \begin{equation}{\rm Id}_{\textbf{B}{\cal G}}=\left({\rm id}_G,{\rm id}_H\right),\end{equation}
    \item $1$-morphisms
    \begin{equation}\label{eq:zeta-drinfeld-object}
        \zeta = \left(\zeta_{\bullet},\beta_{\zeta}\right): \left({\rm id}_G,{\rm id}_H\right) \to \left({\rm id}_G,{\rm id}_H\right),
    \end{equation}
    consisting of a choice of group element $\zeta_{\bullet}\in G$ and a function $\beta_{\zeta}: G \to H$ such that
    \begin{equation}\label{ZXModTrans1}
  f\zeta_{\bullet} =\delta\left(\beta_{\zeta}(f)\right) \eta_{\bullet}f,
\end{equation}
\begin{equation}\label{ZXModTrans2}
    \beta_{\zeta}(fg) = \left(f\triangleright\beta_{\zeta}(g)\right)\beta_{\zeta}(f),
\end{equation}
\begin{equation}\label{ZXModTrans3}
   h\beta_{\zeta}(f) = \beta_{\zeta}(\delta(h)f) \left(\zeta_{\bullet}\triangleright h\right),
\end{equation}
\item the $2$-morphisms are modifications
\begin{equation}\label{eq:drinfeldmorphism}
    m=m_{\bullet}: \left(\zeta_{\bullet},\beta_{\zeta}\right) \to \left(\zeta'_{\bullet},\beta_{\zeta'}\right)
\end{equation}
which is determined by a group element $m_{\bullet}\in H$ such that
\begin{equation}
    \zeta_{\bullet}' = \delta(m_{\bullet})\zeta_{\bullet},
\end{equation}
\begin{equation}\label{ZXModModif}
   \left(f\triangleright m_{\bullet} \right) \beta_{\zeta}(f) = \beta_{\zeta'}(f) m_{\bullet}.
\end{equation}
\end{itemize}
It is straightforward to see that the $1$- and $2$-morphisms, along with their identities, coincide with the objects and morphisms of the \textit{Drinfeld center} $\cal Z(G)$ of $\cal G$, when the latter is regarded as a monoidal category \cite{etingof2015tensor}. Indeed, a pair $(\zeta_{\bullet},\beta_{\zeta})$ as in (\ref{eq:zeta-drinfeld-object}) describes an object $\zeta_{\bullet}$ of $\cal G$ along with a \textit{half-braiding} $\beta_{\zeta}$, and the $2$-morphisms (\ref{eq:drinfeldmorphism}) are the morphisms of the monoidal category $\cal G$ compatible with the half-braidings.

Therefore, the \textit{delooping} of the Drinfeld center $\cal Z(G)$ of $\cal G$ is a sub-$2$-groupoid of ${\rm Aut}(\textbf{B}{\cal G})$
\begin{equation}
    \textbf{B}{\cal Z(G)} \to {\rm Aut}(\textbf{B}{\cal G}),
\end{equation}
which is the inclusion of the center referred to in Equation~(\ref{eq:higheraut}).

\end{remark}

The key objects for the present discussion are, of course, the $2$-groupoids $\textbf{Fields}$ and $\textbf{Sym}(\Sigma,\textbf{B}{\cal G})$, which we now define.

\begin{definition}[$2$-groupoid of principal $\cal G$ bundles (cf. Appendix~\ref{app:adjbundles}, e.g. \cite{Rist:2022hci})]

The $2$-groupoid of principal $\cal G$ bundles over $\Sigma$ with respect to a good cover $\{U_i\}_{i\in {\cal I}}$ is the mapping $2$-groupoid of strict $2$-functors
\begin{equation}
  \label{UnderlyingGroupoidOfGBundles}
\textbf{Fields}= {\rm Hom}(U_2(\Sigma), \textbf{B}{\cal G})
  =
  \left\{
  \begin{tikzcd}
    \left(g_{ij},h_{ijk}\right)
    \ar[
      rr,
      bend left=40,
      "{(a_i,b_{ij})}"{description, name=s}
    ]
    \ar[
      rr,
      bend right=40,
      "{(\tilde{a}_i,\tilde{b}_{ij})}"{description, name=t}
    ]
    \ar[
      from=s, to=t,
      Rightarrow,
      "{m_i}"
    ]
    &&
   \left(g'_{ij},h'_{ijk}\right)
  \end{tikzcd}
  \middle\vert\;
  \begin{aligned}
    (g_{ij}, h_{ijk}), (g'_{ij}, h'_{ijk}) & \ (\ref{eq:gftn})-(\ref{eq:hcocycle})
    \\
   (a_i, b_{ij}), (\tilde{a}_i, \tilde{b}_{ij}) & (\ref{eq:aftn})-(\ref{eq:bcocycle})
    \\
    m_{i} &  \ (\ref{eq:mftn})-(\ref{eq:mcocycle})
  \end{aligned}
  \right\}
 ,
\end{equation}
described as follows
\begin{itemize}
    \item The objects $(g_{ij},h_{ijk})$ are functions at the intersections
    \begin{eqnarray}
        g_{ij} &:& U_{ij} \to G,\label{eq:gftn}
        \\
        h_{ijk} &:& U_{ijk} \to H,
    \end{eqnarray}
    satisfying the identities
    \begin{eqnarray}
        g_{ik} &=& \delta(h_{ijk}) g_{ij} g_{jk},
        \\
        h_{ikl}h_{ijk} &=& h_{ijl}\left( g_{ij} \triangleright h_{jkl}\right) \label{eq:hcocycle},
    \end{eqnarray}
    \item the $1$-morphisms 
    \begin{equation}
        \left(a_i,b_{ij}\right): \left(g_{ij}, h_{ijk}\right) \xrightarrow{\sim}  \left(g'_{ij}, h'_{ijk}\right)
    \end{equation}
    are functions
    \begin{eqnarray}
        a_{i} &:& U_{i} \to G,\label{eq:aftn}
        \\
        b_{ij} &:& U_{ij} \to H,
    \end{eqnarray}
    satisfying
    \begin{eqnarray}
    a_i g_{ij}' &=& \delta(b_{ij}) g_{ij} a_j,
    \\
        \left(a_i\triangleright h'_{ijk}\right) b_{ij} b_{jk} &=& b_{ik} h_{ijk},
         \label{eq:bcocycle}
    \end{eqnarray}
    \item the $2$-morphisms
    \begin{equation}
        m_i: \left( a_i, b_{ij} \right) \xrightarrow{\sim} \left( \tilde{a}_i, \tilde{b}_{ij} \right),
    \end{equation}
    is a function
    \begin{eqnarray}
        m_i &:& U_i \to H \label{eq:mftn},
    \end{eqnarray}
    satisfying
    \begin{eqnarray}
        \tilde{a}_i &=& \delta(m_i) a_i,
        \\
        \tilde{b}_{ij} \left( g_{ij} \triangleright m_j \right) &=& m_i b_{ij} \label{eq:mcocycle}.
    \end{eqnarray}
\end{itemize}
\end{definition}
\begin{definition}[Symmetry $2$-groupoid $\textbf{Sym}\left(\Sigma,\textbf{B}{\cal G}\right)$]\label{def:sym2grpd}

Given a good cover $\{U_i\}_{i \in {\cal I}}$ of $\Sigma$, and ${\rm Aut}(\textbf{B}{\cal G})$ the automorphism $2$-groupoid from Definition~\ref{def:aut2gp}, the $2$-groupoid $\symbg$ of global higher-form symmetries is
\begin{equation}
  \label{UnderlyingGroupoidOfAutomorphismValuedMaps}
 \textbf{Sym}\left(\Sigma,\textbf{B}{\cal G} \right)
  =
  \left\{
  \begin{tikzcd}
    \Phi_{ijk}
    \ar[
      rr,
      bend left=40,
      "{(\nu_i,\mu_{ij})}"{description, name=s}
    ]
    \ar[
      rr,
      bend right=40,
      "{\left(\tilde{\nu}_i,\tilde{\mu}_{ij}\right)}"{description, name=t}
    ]
    \ar[
      from=s, to=t,
      Rightarrow,
      "{ M_i}"
    ]
    &&
   \Phi'_{ijk}
  \end{tikzcd}
  \middle\vert\;
  \begin{aligned}
   \Phi_{ijk}, \Phi'_{ijk} & \ (\ref{XModFunctor1}), (\ref{XModFunctor2})
    \\
    (\nu_i,\mu_{ij}), \left(\tilde{\nu}_i,\tilde{\mu}_{ij}\right) & (\ref{XModTrans1})-(\ref{XModTrans3})
    \\
    M_i &  \ (\ref{XModModif}),
  \end{aligned}
  \right\},
\end{equation}
described as follows
\begin{itemize}
    \item the objects
    \begin{equation}
        \Phi_{ijk} = \left((\phi_i,\psi_i), (\eta_{\bullet,ij},\beta_{\eta_{ij}}),m_{\bullet,ijk} \right),
    \end{equation}
    consist of
    \begin{enumerate}
        \item $(\phi_i,\psi_i): \textbf{B}{\cal G}\to \textbf{B}{\cal G}$ a strict invertible $2$-endofunctor of $\textbf{B}{\cal G}$ (Eq's.~(\ref{eq:bgendofunctor})-(\ref{XModFunctor2})) for each patch $U_i$,
        \item $(\eta_{\bullet,ij},\beta_{\eta_{ij}})$ a natural isomorphism (Eq's.~(\ref{XModTrans1})-(\ref{XModTrans3}))
        \begin{equation}
            (\eta_{\bullet,ij},\beta_{\eta_{ij}}): (\phi_i,\psi_i) \xrightarrow{\sim} (\phi_j,\psi_j),
        \end{equation}
        at each intersection $U_{ij}$
        \item $m_{\bullet,ijk}$ a modification (Eq.~(\ref{XModModif}))
        \begin{equation}
            m_{\bullet,ijk}: (\eta_{\bullet,jk},\beta_{\eta_{jk}})\circ (\eta_{\bullet,ij},\beta_{\eta_{ij}}) \xrightarrow{\sim} (\eta_{\bullet,ik},\beta_{\eta_{ik}}),
        \end{equation}
        at each triple intersection $U_{ijk}$.
    \end{enumerate}
    Together, they satisfy the cocycle condition
    \begin{equation}
\begin{tikzcd}
{(\phi_j,\psi_j)} \arrow[rrrr, "{(\eta_{\bullet,jk},\beta_{\eta_{jk}})}" description] \arrow[rrdd, "{m_{\bullet,ijk}}" description, Rightarrow]                                                                        &  &    &  & {(\phi_k,\psi_k)} \arrow[dddd, "{(\eta_{\bullet,kl},\beta_{\eta_{kl}})}" description] \arrow[lldddd, "{m_{\bullet,ikl}}" description, Rightarrow] &   & {(\phi_j,\psi_j)} \arrow[rrrr, "{(\eta_{\bullet,jk},\beta_{\eta_{jk}})}" description] \arrow[rrrrdddd, "{(\eta_{\bullet,jl},\beta_{\eta_{jl}})}" description] \arrow[rrdddd, "{m_{\bullet,ijl}}" description, Rightarrow] &  &    &  & {(\phi_k,\psi_k)} \arrow[dddd, "{(\eta_{\bullet,kl},\beta_{\eta_{kl}})}" description] \arrow[lldd, "{m_{\bullet,jkl}}" description, Rightarrow] \\
                                                                                                                                                                                                                       &  &    &  &                                                                                                                                                   &   &                                                                                                                                                                                                                           &  &    &  &                                                                                                                                                 \\
                                                                                                                                                                                                                       &  & {} &  &                                                                                                                                                   & = &                                                                                                                                                                                                                           &  & {} &  &                                                                                                                                                 \\
                                                                                                                                                                                                                       &  &    &  &                                                                                                                                                   &   &                                                                                                                                                                                                                           &  &    &  &                                                                                                                                                 \\
{(\phi_i,\psi_i)} \arrow[uuuu, "{(\eta_{\bullet,ij},\beta_{\eta_{ij}})}" description] \arrow[rrrr, "{(\eta_{\bullet,il},\beta_{\eta_{il}})}"'] \arrow[rrrruuuu, "{(\eta_{\bullet,ik},\beta_{\eta_{ik}})}" description] &  & {} &  & {(\phi_l,\psi_l)}                                                                                                                                 &   & {(\phi_i,\psi_i)} \arrow[uuuu, "{(\eta_{\bullet,ij},\beta_{\eta_{ij}})}" description] \arrow[rrrr, "{(\eta_{\bullet,il},\beta_{\eta_{il}})}"']                                                                            &  & {} &  & {(\phi_l,\psi_l)}                                                                                                                              
\end{tikzcd},
    \end{equation}
\item the $1$-morphisms
\begin{equation}
    \left(\nu_i, \mu_{ij}\right) = \left( (\nu_{\bullet,i}, \beta_{\nu_i}), \mu_{\bullet,ij}\right): \Phi_{ijk} \xrightarrow{\sim} \Phi'_{ijk},
\end{equation}
consist of
\begin{enumerate}
    \item a natural isomorphism
    \begin{equation}
        (\nu_{\bullet,i}, \beta_{\nu_i}): \left(\phi_i,\psi_i\right) \xrightarrow{\sim} \left(\phi'_i,\psi'_i\right),
    \end{equation}
    defined at each patch $U_i$,
    \item a modification
    \begin{equation}
        \mu_{\bullet,ij}:(\nu_{\bullet,j}, \beta_{\nu_j}) \circ (\eta_{\bullet,ij}, \beta_{\eta_{ij}}) \xrightarrow{\sim} (\eta'_{\bullet,ij}, \beta_{\eta'_{ij}}) \circ (\nu_{\bullet,i}, \beta_{\nu_i}),
    \end{equation}
    at each intersection $U_{ij}$.
\end{enumerate}
Together, they satisfy the relation

\begin{equation}
\begin{tikzcd}
                                                                                                                                      &    & {(\phi'_k,\psi'_k)}                                                                                                                              &                                                                                                                                                                                                                                         \\
                                                                                                                                      & {} &                                                                                                                                                  &                                                                                                                                                                                                                                         \\
{(\phi'_i,\psi'_i)} \arrow[rr, "{(\eta'_{\bullet,ij},\beta_{\eta_{ij}'})}"] \arrow[rruu, "{(\eta'_{\bullet,ik},\beta_{\eta_{ik}'})}"] &    & {(\phi'_j,\psi'_j)} \arrow[uu, "{(\eta'_{\bullet,jk},\beta_{\eta_{jk}'})}" description] \arrow[lu, "{m_{\bullet,ijk}'}" description, Rightarrow] & {(\phi_k,\psi_k)} \arrow[luu, "{(\nu_{\bullet,k},\beta_{\nu_k})}"']                                                                                                                                                                     \\
                                                                                                                                      &    &                                                                                                                                                  &                                                                                                                                                                                                                                         \\
{(\phi_i,\psi_i)} \arrow[uu, "{(\nu_{\bullet,i},\beta_{\nu_i})}"] \arrow[rrr, "{(\eta_{\bullet,ij},\beta_{\eta_{ij}})}"']             &    &                                                                                                                                                  & {(\phi_j,\psi_j)} \arrow[luu, "{(\nu_{\bullet,j},\beta_{\nu_j})}" description] \arrow[llluu, "\mu_{ij}", Rightarrow] \arrow[uu, "{(\eta_{\bullet,jk},\beta_{\eta_{jk}})}"'] \arrow[luuuu, "{\mu_{\bullet,jk}}" description, Rightarrow]
\end{tikzcd} =
\begin{tikzcd}
                                                                                                                                                                                 & {(\phi'_k,\psi'_k)} &                                                                                                                                  \\
{(\phi'_i,\psi'_i)} \arrow[ru, "{(\eta'_{\bullet,ik},\beta_{\eta_{ik}'})}"]                                                                                                      &                     & {(\phi_k,\psi_k)} \arrow[lu, "{(\nu_{\bullet,k},\beta_{\nu_k})}"'] \arrow[ll, "{\mu_{\bullet,ik}}" description, Rightarrow]      \\
                                                                                                                                                                                 & {}                  &                                                                                                                                  \\
{(\phi_i,\psi_i)} \arrow[uu, "{(\nu_{\bullet,i},\beta_{\nu_i})}"] \arrow[rr, "{(\eta_{\bullet,ij},\beta_{\eta_{ij}})}"'] \arrow[rruu, "{(\eta_{\bullet,jk},\beta_{\eta_{jk}})}"] &                     & {(\phi_j,\psi_j)} \arrow[uu, "{(\eta_{\bullet,jk},\beta_{\eta_{jk}})}"'] \arrow[lu, "{m_{\bullet,ijk}}" description, Rightarrow]
\end{tikzcd},
\end{equation}
\item the $2$-morphisms
\begin{equation}
    M_i=M_{\bullet,i}: \left( (\nu_{\bullet,i}, \beta_{\nu_i}), {\mu}_{\bullet,ij}\right) \xrightarrow{\sim} \left( (\tilde{\nu}_{\bullet,i}, \beta_{\tilde{\nu}_i}), \tilde{\mu}_{\bullet,ij}\right)
\end{equation}
are modifications
\begin{equation}
    M_{\bullet,i}: (\nu_{\bullet,i}, \beta_{\nu_i}) \xrightarrow{\sim} (\tilde{\nu}_{\bullet,i}, \beta_{\tilde{\nu}_i})
\end{equation}
defined at each patch $U_i$, satisfying the identity
\begin{equation}
      \begin{tikzcd}
    (\phi_i,\psi_i)
    \ar[
      dd,
      bend left=50,
      "{ \nu_i }"{description, pos=.4, name=s}
    ]
    \ar[
      dd,
      bend right=50,
      "{ \tilde{\nu}_i}"{description, pos=.6, name=t}
    ]
    \ar[
      from=s, to=t,
      Rightarrow,
      "{ M_i }"{pos=1}
    ]
    \ar[
      rr,
      "{ (\eta_{\bullet,ij},\beta_{\eta_{ij}}) }"{description}
    ]
    &&
    (\phi_j,\psi_j)
    \ar[
      dd,
      bend left=50,
      "{\nu_j}"{description}
    ]
    \ar[
      ddll,
      Rightarrow,
      shorten=15pt,
      shift right=5pt,
      bend left=20,
      "{ \mu_{ij} }"{description}
    ]
    \\
    \\
    (\phi'_i,\psi'_i)
    \ar[
      rr,
      "{ (\eta'_{\bullet,ij},\beta_{\eta'_{ij}}) }"{description}
    ]
    &&
    (\phi'_j,\psi'_j)
  \end{tikzcd} 
  =
 \begin{tikzcd}
    (\phi_i,\psi_i)
    \ar[
      dd,
      bend right=50,
      "{  \tilde{\nu}_i }"{description, pos=.6, name=t}
    ]
    \ar[
      rr,
      "{ (\eta_{\bullet,ij},\beta_{\eta_{ij}}) }"{description}
    ]
    &&
    (\phi_j,\psi_j)
    \ar[
      dd,
      bend left=50,
      "{ \nu_j }"{description, pos=.4, name=s}
    ]
    \ar[
      dd,
      bend right=50,
      "{ \tilde{\nu}_j }"{description, pos=.6, name=t}
    ]
    \ar[
      from=s, to=t,
      Rightarrow,
      "{ M_j }"{pos=1}
    ]
    \ar[
      ddll,
      Rightarrow,
      shorten=15pt,
      shift left=5pt,
      bend right=20,
      "{ \tilde{\mu}_{ij} }"{description}
    ]
    \\
    \\
    (\phi'_i,\psi'_i)
    \ar[
      rr,
      "{ (\eta'_{\bullet,ij},\beta_{\eta'_{ij}}) }"{description}
    ]
    &&
    (\phi'_j,\psi'_j)
  \end{tikzcd}.
\end{equation}
\end{itemize}
\end{definition}
\begin{remark}[Flat center principal $\cal Z(G)$ bundles]
In light of the inclusion
\begin{equation}
    \textbf{B}{\cal Z(G)} \to {\rm Aut}\left(\textbf{B}{\cal G}\right)
\end{equation}
described in Remark~\ref{rmk:drinfeld-center}, we observe that amongst the symmetries in $\textbf{Sym}\left(\Sigma,\textbf{B}{\cal G}\right)$ we have as a sub-$2$-groupoid the $2$-groupoid of \textit{flat} principal $\cal Z(G)$ bundles over $\Sigma$. These describe the \textit{center higher-form symmetry}, as we will see explicitly below.
\end{remark}

To exhibit an action of $\textbf{Sym}(\Sigma, \textbf{B}{\cal G})$ on \textbf{Fields}, we first need to show the $2$-groupoid $\textbf{Sym}(\Sigma, \textbf{B}{\cal G})$ has a composition law.
\begin{propo}[$\textbf{Sym}(\Sigma, \textbf{B}{\cal G})$ is a Gray monoid]\label{propo:gray-monoid}

Let $\textbf{Sym}(\Sigma, \textbf{B}{\cal G})$ be the mapping $2$-groupoid from Definition~\ref{def:sym2grpd}. It is a Gray monoid under which all objects are invertible.  
\end{propo}
\begin{proof}
   The composition functor 
   \begin{equation}
       \otimes: \textbf{Sym}(\Sigma, \textbf{B}{\cal G}) \otimes_{\rm G} \textbf{Sym}(\Sigma, \textbf{B}{\cal G}) \to \textbf{Sym}(\Sigma, \textbf{B}{\cal G})
   \end{equation}
   is constructed from that which the $2$-category ${\rm Hom}(\textbf{B}{\cal G}, \textbf{B}{\cal G})$ has by virtue of being a hom-$2$-category \cite[Proposition 6.3.1]{gurski2006algebraic}, \cite[Section 5]{gordon1995coherence}. We specify the images of all objects and morphisms in $\textbf{Sym}(\Sigma, \textbf{B}{\cal G}) \otimes_{\rm G} \textbf{Sym}(\Sigma, \textbf{B}{\cal G})$; the verification that these assignments satisfy the required axioms of a monoidal product on $\textbf{Sym}(\Sigma, \textbf{B}{\cal G})$ follows from a direct application of the identities of Definitions \ref{def:aut2gp}, \ref{def:sym2grpd}.
   \begin{itemize}
       \item On objects
       \begin{eqnarray}
           \otimes\left(\Phi_{ijk}, \Phi'_{ijk}\right) &=& \left(\Phi\Phi'\right)_{ijk}
           \\
           \left(\Phi\Phi'\right)_{ijk} &:=& \Bigg(\left(\phi_i\phi'_i,\psi_i\psi'_i\right),\left(\eta_{\bullet,ij}\phi_j\left(\eta'_{\bullet,ij}\right),\left(\beta_{\eta_{ij}}\phi_i'\right)\left(\eta_{\bullet,ij}\triangleright\left(\psi_j\beta_{\eta'_{ij}}\right)\right)\right),
           \\ & &
           \left(\eta_{\bullet,ik}\triangleright \psi_k\left(m'_{ijk}\right)\right)m_{\bullet,ijk}\left(\eta_{\bullet,ij}\triangleright \left(\beta_{\eta_{jk}}\left(\eta
           _{\bullet,ij}\right) \right)^{-1}\right)\Bigg),
       \end{eqnarray}
       \item on generating $1$-morphisms
        \begin{eqnarray}
            \otimes\left(\Phi_{ijk}\otimes \left(\left(\nu'_{\bullet,i},\beta_{\nu'_i}\right),\mu'_{ij} \right)\right) &=& \left(\left(\phi_i\left(\nu'_{\bullet,i}\right), \psi_i\left(\beta_{\nu'_i}\right)\right),\beta_{\eta_{ij}}\left(\nu'_{\bullet,i} \right) \left( \eta_{\bullet,ij}\triangleright \psi_j\left( \mu'_{\bullet,ij}\right)\right)  \right),
       \\  \otimes\left(\left(\left(\nu_{\bullet,i},\beta_{\nu_i}\right),\mu_{ij} \right) \otimes \Phi_{ijk}'\right) &=& \left( \left(\nu_{\bullet,j},\beta_{\nu_j}\phi'_j \right),\mu_{\bullet,ij}\left(\eta_{\bullet,ij}\triangleright\left(\beta_{\nu_j}\left(\eta'_{\bullet,ij}\right) \right)^{-1} \right)\right),
        \end{eqnarray}
        \item on generating $2$-cells
        \begin{eqnarray}
          \otimes \left(\Phi_{ijk} \otimes M'_{\bullet,i} \right) 
          &=&  \psi_i\left(M'_{\bullet,i}\right)
          \\
          \otimes \left( M_{\bullet,i} \otimes \Phi'_{ijk}\right) 
          &=& M_{\bullet,i},
          \end{eqnarray}
\begin{equation}\label{eq:composition-sym-sigma}
          \otimes \left( \Sigma_{\left((\nu_{\bullet,i},\beta_{\nu_i}),\mu_{ij} \right),\left((\nu'_{\bullet,i},\beta_{\nu'_i}),\mu'_{ij} \right)}\right) = \left(\beta_{\nu_i}\left(\nu_{\bullet,i}' \right) \right)^{-1}.
        \end{equation}
   \end{itemize}

The monoidal unit is the map with identity information at all levels
\begin{equation}
    I = \left(({\rm id}_G, {\rm id}_H)_i, (e_G, e_{G\to H})_{ij}, (e_H)_{ijk} \right),
\end{equation}
where $e_{G \to H}:G \to H$ is the constant map at $e_H$.  
\end{proof}

Having shown $\textbf{Sym}\left(\Sigma, \textbf{B}{\cal G}\right)$ is a Gray monoid where all objects are invertible, it makes sense to talk about $\textbf{Sym}\left(\Sigma, \textbf{B}{\cal G}\right)$-modules.

\begin{propo}[\textbf{Fields} is a $\textbf{Sym}\left(\Sigma,\textbf{B}{\cal G}\right)$-module]\label{propo:fields-mod}

Let $\textbf{Fields}={\rm Hom}\left(U_2\left(\Sigma\right),\textbf{B}{\cal G}\right)$ be the $2$-groupoid of principal $\cal G$ bundles over $\Sigma$ (with respect to a good cover $\{U_i\}_{i\in {\cal I}}$), and $\textbf{Sym}\left(\Sigma,\textbf{B}{\cal G}\right)= {\rm Hom}\left(U_2(\Sigma),{\rm Aut}\left(\textbf{B}{\cal G}\right) \right)$ the Gray monoid of Proposition~\ref{propo:gray-monoid}. Then \textbf{Fields} is a $\textbf{Sym}\left(\Sigma,\textbf{B}{\cal G}\right)$-module $2$-groupoid.
\end{propo}

\begin{proof}
  As before, we provide the strict $2$-functor. The verification that it satisfies the required identities follows from a direct application of the definitions.

The action $2$-functor 
\begin{equation}
    \cdot : \textbf{Sym}\left(\Sigma, \textbf{B}{\cal G}\right) \otimes_G \textbf{Fields} \to \textbf{Fields}
\end{equation}
has the following images:
\begin{itemize}
    \item On objects
    \begin{eqnarray}
       \cdot  \left(\Phi_{ijk},\left(g_{ij},h_{ijk} \right) \right)  &=& \left(\Phi g_{ij},\Phi h_{ijk} \right)
       \\
       \left(\Phi g_{ij},\Phi h_{ijk} \right) &:=& \left(\eta_{\bullet, ij}\phi_j\left(g_{ij}\right), \left(\eta_{\bullet,ik}\triangleright \psi_k\left(h_{ijk}\right) \right) m_{\bullet,ijk} \left(\eta_{\bullet,ij} \triangleright \left(\beta_{\eta_{jk}}\left(g_{ij}\right) \right)^{-1} \right) \right),
    \end{eqnarray}
    \item on generating $1$-morphisms
    \begin{eqnarray}
        \cdot\left(\Phi_{ijk}\otimes \left(a_i,b_{ij}\right)\right) &=& \left(\phi_i\left(a_i\right), \beta_{\eta_{ij}}\left(a_i\right)\left(\eta_{\bullet,ij} \triangleright \psi_j\left(b_{ij}\right)\right)\right)
        \\
        \cdot \left(\left(\left(\nu_{\bullet,i},\beta_{\nu_{i}}\right),\mu_{\bullet, ij}\right),\left(g_{ij},h_{ijk}\right) \right) &=& \left(\nu_{\bullet,i} , \mu_{\bullet,ij}\left( \eta_{\bullet,ij} \triangleright \left(\beta_{\nu_j}\left(g_{ij}\right)\right)^{-1} \right)\right),
    \end{eqnarray}
    \item on generating $2$-morphisms
    \begin{eqnarray}
       \cdot\left( \Phi_{ijk} \otimes m_i\right) &=& \psi_i\left(m_i\right),
        \\
        \cdot \left(M_{\bullet,i} \otimes \left( g_{ij}, h_{ijk}\right) \right) &=& M_{\bullet,i},
    \end{eqnarray}
    \begin{equation}\label{eq:action-sym-sigma}
        \cdot \left(\Sigma_{\left(\left(\nu_{\bullet,i},\beta_{\nu_{i}}\right),\mu_{\bullet, ij}\right),\left(a_i,b_{ij}\right)} \right) = \left(\beta_{\nu_i}\left(a_i\right) \right)^{-1}.
    \end{equation}
\end{itemize}

The action strict $2$-functor satisfies the module condition \textit{strictly}, meaning the diagram
\begin{equation}
\begin{tikzcd}
                                                                                                                                                                                                                                                                                                                                                                                         & {\textbf{Sym}\left(\Sigma,\textbf{B}{\cal G}\right) \otimes_G \left(\textbf{Sym}\left(\Sigma,\textbf{B}{\cal G}\right) \otimes _G \textbf{Fields}\right)} \arrow[r, "{\left({\rm id}_{\textbf{Sym}\left(\Sigma,\textbf{B}{\cal G}\right)}\right)\otimes_G \left(\cdot\right)}", shift left] & {\textbf{Sym}\left(\Sigma,\textbf{B}{\cal G}\right) \otimes_G} \arrow[dd, "\cdot"] \\
{\left( \textbf{Sym}\left(\Sigma,\textbf{B}{\cal G}\right) \otimes_G \textbf{Sym}\left(\Sigma,\textbf{B}{\cal G}\right)\right)\otimes_G \textbf{Fields}} \arrow[rd, "\left(\otimes\right) \otimes_G \left({\rm id}_{\textbf{Fields}}\right)"'] \arrow[ru, "{a_{\textbf{Sym}\left(\Sigma,\textbf{B}{\cal G}\right),\textbf{Sym}\left(\Sigma,\textbf{B}{\cal G}\right),\textbf{Fields}}}"] &                                                                                                                                                                                                                                                                                             &                                                                                    \\
                                                                                                                                                                                                                                                                                                                                                                                         & {\textbf{Sym}\left(\Sigma,\textbf{B}{\cal G}\right) \otimes_G \textbf{Fields}} \arrow[r, "\cdot"']                                                                                                                                                                                          & \textbf{Fields}                                                                   
\end{tikzcd}
\end{equation}
commutes strictly, using the associator $2$-functor 
\begin{equation}
    a_{\cal C,D,E}: \left({\cal C}\otimes_G {\cal D} \right)\otimes_G {\cal E} \to {\cal C} \otimes _G \left( {\cal D}\otimes_G {\cal E}\right)
\end{equation}
of bicategories under the Gray product \cite[Theorem 12.2.23]{JohnsonYau2021}.
  
\end{proof}

\begin{remark}[Center higher-form symmetry]
The composition law of Proposition~\ref{propo:gray-monoid} can be seen to adequatedly restrict to the sub-$2$-groupoid of flat principal $\cal Z(G)$ bundles, so that the latter is moreover a monoidal sub-$2$-groupoid of $\textbf{Sym}\left(\Sigma,\textbf{B}{\cal G}\right)$. Combined with Proposition~\ref{propo:fields-mod}, this shows that \textbf{Fields} has an action by flat principal $\cal Z(G)$ bundles. This is the \textit{center higher-form symmetry} tensoring action of Equation~\ref{eq:tensorbundles}:
\begin{equation}
    \left( \textit{flat} \, \text{principal} \, {\cal Z}({\cal G}) \, \text{bundles} \right) \otimes \left( \text{principal} \, {\cal G} \, \text{bundles}  \right)  \to \left( \text{principal} \, {\cal G} \, \text{bundles}  \right).
\end{equation}

\end{remark}

\section{Pure $G$ gauge theory}\label{sec:puregauge}

In this section, we consider pure gauge theories. A pure gauge theory with gauge Lie (1-)group $G$ can be realized as a $\sigma$-model, where the target space is the \textit{smooth moduli stack of principal $G$ bundles with connection} 
\begin{equation}
    \textbf{Fields} : = \textbf{H}(\Sigma, \textbf{B}_{\nabla}G).
\end{equation}

We first describe the symmetries without connections by applying the results of Section~\ref{sec:2-action}, and then proceed to analyze the interplay with connections, symmetries, and gaugings in the case $G=U(1)$.

\subsection{Symmetries of principal $G$ bundles}\label{ssec:aut2gp}

As already recognized in the literature (e.g. \cite{Gaiotto:2014kfa,Bhardwaj:2023kri,Bartsch:2023pzl,Santilli:2024dyz}), $G$ gauge theories have an $\text{Out}(G)$ 0-form symmetry, and a $Z(G)$ 1-form symmetry or, more precisely, a global 2-group symmetry controlled by the \textit{automorphism $2$-group} $\text{AUT}(G)$. Presented as a strict $2$-group, this is the crossed module of Lie groups (see Appendix~\ref{app:adjbundles})
\begin{equation}\label{eq:aut2gp}
    \text{AUT}(G) = \left( G \xrightarrow{\text{ad}} \text{Aut}(G) \right),
\end{equation}
where ${\rm ad}: G \to  {\rm Aut}(G)$ is the adjoint action with kernel $\text{ker}(\text{ad}) = Z(G)$, the 1-form symmetry, and cokernel $\text{coker}(\text{ad})= \text{Out}(G)$, the 0-form symmetry. Indeed, this fits into an extension of smooth 2-groups
\begin{equation}\label{eq:autsequence}
    1\to \textbf{B}Z(G) \to \text{AUT}(G)\to \text{Out}(G) \to 1,
\end{equation}
for $\textbf{B}Z(G)$ the 2-group corresponding to the delooping of the abelian group $Z(G)$.

This smooth automorphism 2-group can be computed as the collection of automorphisms of $G$, not in the category of Lie groups, but in the 2-category of smooth groupoids. In other words, and as highlighted in Section~\ref{sec:general}, we can compute this as (the $2$-truncation of) the automorphisms of $\textbf{B}G$ in \textbf{H}, meaning that as smooth 2-groups
\begin{equation}
    \aut{G}= \tau_{\leq 1}{\rm Aut}(\textbf{B}G),
\end{equation}
(cf. Eq.(\ref{eq:aut-trunaction-general})) where ${\rm Aut}(\textbf{B}G)$ is computed in \textbf{H}.

The higher-form symmetries for $G$ gauge theory, according to (\ref{eq:flatsymmetryparameters}), are parameterized by the stack of flat smooth morphisms
\begin{equation}
     \textbf{Sym}(\Sigma,{\cal X}):=\textbf{H}_{\flat}(\Sigma,\aut{G}),
\end{equation}
where $\text{AUT}(G)$ is regarded as a smooth stack.

In this special case of smooth ($1$-)groups, we can define $\cal G$-symmetries for $\cal G$ a smooth $2$-group by considering smooth $2$-group homomorphisms
\begin{equation}\label{eq:hgaction}
    \alpha: {\cal G}\to \aut{G},
\end{equation}
which can be gauged.

Given that Lie groups are special cases of smooth strict $2$-groups via the crossed module presentation
\begin{equation}
    G = \left(G, 0, \delta=0, \triangleright=0\right),
\end{equation}
we can apply the results of Section~\ref{sec:2-action} to obtain the symmetries.

\paragraph{Automorphism $2$-groupoid.} Definition~\ref{def:aut2gp} applied to this special case gives the $2$-groupoid ${\rm Aut}(\textbf{B}G)$ whose
\begin{itemize}
    \item objects are group automorphisms
    \begin{equation}
        \phi: G\xrightarrow{\sim} G,
    \end{equation}
    \item $1$-morphisms
    \begin{equation}
        \eta_{\bullet}: \phi \to \phi'
    \end{equation}
    are group elements $\eta_{\bullet}\in G$ such that
    \begin{equation}
        \phi(f)\eta_{\bullet} = \eta_{\bullet}\phi'(f),
    \end{equation}
    for all $f\in G$,
    \item $2$-morphisms are only identity $2$-morphisms,
\end{itemize}
which recovers the automorphism $2$-group (\ref{eq:aut2gp}). In particular, the center ${\cal Z}(G)$ in Remark~\ref{rmk:drinfeld-center} corresponds to the usual center $Z(G)$ of $G$, recovering the inclusion of its delooping in the sequence (\ref{eq:autsequence}).

\paragraph{Symmetry action.} The $2$-groupoid of symmetries $\textbf{Sym}\left(\Sigma,\textbf{B}G\right)$ as per Definition~\ref{def:sym2grpd} is the $2$-groupoid whose
\begin{itemize}
    \item objects are pairs $\left(\phi_i, \eta_{\bullet,ij}\right)$, consisting of $G$ isomorphisms $\phi_i\in {\rm Aut}(G)$ at each patch $U_i$, and group elements $\eta_{\bullet,ij}\in G$, together satisfying, for all $f\in G$
    \begin{eqnarray}
        \phi_i(f)\eta_{\bullet,ij} &=& \eta_{\bullet,ij}\phi_j(f),
        \\
        \eta_{\bullet,ik} &=& \eta_{\bullet,ij}\eta_{\bullet,jk},
    \end{eqnarray}
    \item $1$-morphisms
    \begin{equation}
        \nu_{\bullet,i}: \left(\phi_i, \eta_{\bullet,ij}\right) \xrightarrow{\sim } \left(\phi'_i, \eta'_{\bullet,ij}\right),
    \end{equation}
    are group elements $\nu_{\bullet,i}$ such that
    \begin{eqnarray}
        \phi_i(f) \nu_{\bullet,i} &=& \nu_{\bullet,i}\phi'_i(f),
        \\
        \nu_{\bullet,i}\eta'_{\bullet,ij} &=& \eta_{\bullet,ij} \nu_{\bullet,j},
    \end{eqnarray}
    \item $2$-morphisms are only identity $2$-morphisms.
\end{itemize}

Having shown that the composition of symmetries, as well as the action on principal $G$ bundles is functorial, it suffices to state these at the level of objects, for $\left(\phi_i,\eta_{\bullet,ij}\right) \in {\rm ob}\left(\textbf{Sym}\left(\Sigma,\textbf{B}G\right)\right)$ and $g_{ij}$ a principal $G$ bundle:
\begin{equation}
    \otimes \left(\left(\phi_i,\eta_{\bullet,ij}\right), \left(\phi'_i,\eta'_{\bullet,ij}\right) \right) = \left(\phi_i\phi'_i,\eta_{\bullet,ij}\phi_i\left(\eta'_{\bullet,ij}\right) \right),
\end{equation}
\begin{equation}
    \cdot \left(\left(\phi_i,\eta_{\bullet,ij}\right),g_{ij} \right) = \eta_{\bullet,ij}\phi_j\left(g_{ij}\right).
\end{equation}

\subsection{$G=U(1)$}\label{ssec:u1}

Let us now specialize this discussion further to $G=U(1)$, in order to see how connections interact with the symmetries described so far.

\paragraph{Symmetry $2$-group.}  Since $U(1)$ is abelian, it follows that $\text{Aut}(U(1)) = \text{Out}(U(1))$, which in this case is the finite group $ \Z_2$, and $Z(U(1))=U(1)$. The automorphism $2$-group can thus be presented as a crossed module of Lie groups
\begin{equation}\label{aut-u1}
    \text{AUT}(U(1)) = \left(U(1) \xrightarrow{0} \Z_2 \right),
\end{equation}
where $\Z_2$ acts faithfully on $U(1)$. Thus, a \textit{flat} morphism
\begin{equation}
    f: \Sigma \to \text{AUT}(U(1))
\end{equation}
in $\textbf{Sym}\left(\Sigma,\aut{U(1)} \right)$ consists of pairs $(z,P)$ where $z:\Sigma\to \Z_2$ is a globally-defined $\Z_2$ group element representing an automorphism of $U(1)$, and $P$ is a flat principal $U(1)$ bundle over $\Sigma$. This symmetry admits a connective refinement, as a symmetry of principal $U(1)$ bundles with connection. Thus, for $\{U_i\}_{i\in {\cal I}}$ a good Čech cover for $\Sigma$, we can describe these symmetries as a tuple $(z,(g_{ij},A_i))$ for $(g_{ij},A_i)$ a principal $U(1)$ bundle with a flat connection.

In Čech data, the induced composition law on $\textbf{H}_{\flat}(\Sigma, \aut{U(1)})$ is
\begin{equation}
    \left(z,\left(g_{ij},A_i\right)\right)\otimes \left(z',\left(g'_{ij},A'_i\right)\right) = \left(zz', \left(g_{ij} \left(z \left(g'_{ij}\right)\right),A_i+z^*A_i'\right)\right),
\end{equation}
where $z^*: \mathfrak{u}(1)\to \mathfrak{u}(1)$ is the differential of $z:U(1)\to U(1)$. On maps of the form $(1,(g_{ij},A_i))$, this composition corresponds to the usual \textit{tensor} product of principal $U(1)$ bundles with (flat) connection.

The action of this global $U(1)$ $1$-form symmetry is the same composition law, except that the bundle on which it acts is not necessarily flat. More precisely, for a field $\phi:\Sigma \to \textbf{B}_{\nabla}U(1)$, in Čech data $(g_{ij},a_i)$, the action of the global $\aut{U(1)}$ symmetry, parameterized by $\left(z,\left(g'_{ij},A_i\right)\right)$, is
\begin{equation}
    \left(z,\left(g'_{ij},A_i\right)\right)\cdot \left(g_{ij},a_i\right) = \left(g'_{ij}\left(z\left( g_{ij}\right)\right), A_i + z^*a_i \right).
\end{equation}
This action is readily seen to be functorial.

Note that here the center higher-form symmetry is encoded by flat principal $U(1)$ bundles themselves. From the geometric perspective, the center symmetry action is rooted on the observation that $U(1)=S^1$ can be regarded as the total space of a principal $U(1)$ bundle, which acts by left-translation. One has the sequence
\begin{equation}\label{eq:u1s1}
    U(1) \to U(1) \to U(1)\dslash U(1)=:{\rm INN}(U(1))\to \textbf{B}U(1),
\end{equation}
where the smooth homotopy quotient $U(1)\dslash U(1)={\rm INN}(U(1))$ is the \textit{inner automorphism 2-group} of $U(1)$, which can be realized as a smooth crossed module
\begin{equation}
    {\rm INN}(U(1))=\left(U(1)\xrightarrow{\rm id}U(1)\right).
\end{equation}
The relevant sequence for this discussion is simply the stacky version given by the delooping of (\ref{eq:u1s1})
\begin{equation}\label{eq:bu1s1}
    \textbf{B}U(1) \to \textbf{B}U(1) \to \textbf{B}{\rm INN}(U(1)) \to \textbf{B}^2 U(1),
\end{equation}
which exists because $\textbf{B}U(1)$ itself is a group object (a 2-group). This sequence is essential to understand gaugings, as we describe further below.

\paragraph{Action on Wilson loops and holonomies.} Let us explore the symmetry action further. As mentioned in Section~\ref{sec:intro}, the ``quantum operators'' charged under higher-form symmetries are \textit{extended} operators. For the $U(1)$ $1$-form here, these operators are the \textit{Wilson loops}\footnote{Or, more generally, infinitely-long Wilson lines \cite{Gomes:2023ahz} (assumed to vanish at infinity).}. These are operators coming from connection holonomies, which are mathematically well-defined. We can thus elaborate on these particular actions before quantizing. We will still refer to these as Wilson loops, with the understanding that we are working at the (prequantum) level of holonomies, not of fully quantum operators.

Denote a Wilson loop as
\begin{equation}
    W_n(\gamma,a) = {\rm exp}\left(in\oint_{\gamma} a   \right),
\end{equation}
for $a$ a connection on a principal $U(1)$ bundle over $\Sigma$, and $n\in \Z$ an integer\footnote{This just corresponds to computing the holonomy of the $n$th tensor power of $a$ if $n\in\Z_{+}$, or of that of the dual bundle with connection $(g_{ij}^{-1},-a)$ if $n\in\Z_-$.}. While this definition makes it seem that the Wilson loop depends on the particular connection $a$, the fact that $\gamma$ is a loop implies that if $a\sim_{\rm gauge}a'$ are two gauge-equivalent connections, then
\begin{equation}
   {\rm exp}\left(in\oint_{\gamma} a   \right) = {\rm exp}\left(in\oint_{\gamma} a'   \right).
\end{equation}
Thus, the Wilson loop $W_n(\gamma,-)$ for a given loop $\gamma\subset \Sigma$ can be understood as a $U(1)$-valued ``function'' on $\cal A/G$, the space of connections modulo gauge transformations, or, equivalently, as a function on the \textit{stack} of fields $\textbf{H}(\Sigma,\textbf{B}_{\nabla}U(1))$ which \textit{only depends on the connected component of the gauge field} (since the morphisms in this stack are precisely gauge transformations). 

If we regard $\Sigma$ as a smooth $\infty$-groupoid via its smooth fundamental $\infty$-groupoid $\Pi_{\infty}\Sigma$, the novelty here is that the charged objects now depend not on the objects (points in $\Sigma$) of $\Pi_{\infty}\Sigma$, but on the (higher) morphisms therein. This is the crucial nonlocal generalization of charged quantum operators of \cite{Gaiotto:2014kfa}. However, in this case, \textit{the dependence on the stack of fields remains ``local''}, in the sense that it only depends on the objects (and, in fact, on its connected component).

Returning to the action of the $U(1)$ $1$-form symmetry, we can present it as
\begin{equation}\label{eq:wilsonloopaction}
    (A) \cdot W_n({\gamma},a) = W_n(\gamma,a)\,{\rm exp}\left(in\oint_{\gamma}A\right),
\end{equation}
where ${\rm exp}(in\oint_{\gamma}A)\in U(1)$ and, furthermore, if $\gamma, \gamma'$ are the boundary $\partial S=\gamma \coprod \bar{\gamma}'$ of a surface $S$ (for $\bar{\gamma}'$ the orientation-reversal of $\gamma'$), and so
\begin{equation}
    \text{exp}\left(\oint_{\gamma}A - \oint_{\gamma'}A\right) = \text{exp}\left(\int_S dA\right) = 1,
\end{equation}
meaning the parameter only depends on the class of $\gamma$, since $A$ is flat.

On the other hand, the $\Z_2$ $0$-form symmetry acts on the Wilson lines as
\begin{equation}
    z\cdot W_n(\gamma,a) = W_n(\gamma,-a) = \left(W_n(\gamma,a)\right)^{-1}.
\end{equation}

\paragraph{Gauging.} Having described the symmetries with connections, let us now describe the gauging. For simplicity, we treat them separately.

The $\Z_2$ $0$-form symmetry, by definition, acts on the gauge group $U(1)$ via group homomorphisms, and thus gives rise to a group extension, the \textit{semidirect product}
\begin{equation}
    1\to U(1) \to U(1)\rtimes \Z_2 \to \Z_2 \to 1, 
\end{equation}
where $U(1)\rtimes \Z_2 \cong O(2)$. At the level of stacks, this is realized as the sequence
\begin{equation}
  \textbf{B}U(1) \to \textbf{B}O(2) \to \textbf{B}\Z_2,
\end{equation}
which exhibits $\textbf{B}O(2) \cong \textbf{B}(U(1)\rtimes \Z_2) \cong (\textbf{B}U(1)) \dslash  \Z_2$ as the smooth homotopy quotient of a total space $\textbf{B}U(1)$ by a $\Z_2$-action. This is the statement that the quotient stack $(\textbf{B}U(1)) \dslash  \Z_2$ encodes $\Z_2$-equivariant principal $U(1)$ bundles, equivalently $(U(1)\rtimes \Z_2)=O(2)$ principal bundles \cite[Proposition 4.1.19]{SS25-Bund}, \cite[Eq. 3.4]{Schafer-Nameki:2023jdn}. This gaugeable global $0$-form $\Z_2$ symmetry is not unique to pure $U(1)$ gauge theory but also appears in higher $U(1)$ gauge theory, as we will see in Section~\ref{sec:bfields}.

The effect of this on the original $U(1)$ fields, as befits a weak quotient, is that a $U(1)$ gauge field becomes gauge equivalent to its image under the action of $\Z_2$. This is accomplished explicitly by the inclusion
\begin{eqnarray}
    U(1) &\hookrightarrow & O(2),\\
    \exp(i\alpha) &\mapsto & \begin{pmatrix}
        \cos\alpha & \sin\alpha \\
        -\sin\alpha & \cos\alpha
    \end{pmatrix},
\end{eqnarray}
for $\alpha\in\R/\Z$. For Čech data of a principal $U(1)$ bundle with connection this is
\begin{equation}
    (g_{ij}=\exp(i\alpha_{ij}),A_i) \mapsto \left(\begin{pmatrix}
        \cos\alpha_{ij} & \sin\alpha_{ij} \\
        -\sin\alpha_{ij} & \cos\alpha
    \end{pmatrix}, \begin{pmatrix}
        0 & A_i \\
        -A_i & 0
    \end{pmatrix}\right).
\end{equation}
The action of $\Z_2$ on this data is
\begin{equation}
    (g^{-1}_{ij}=\exp(-i\alpha_{ij}),-A_i) \mapsto \left(\begin{pmatrix}
        \cos\alpha_{ij} & -\sin\alpha_{ij} \\
        \sin\alpha_{ij} & \cos\alpha
    \end{pmatrix}, \begin{pmatrix}
        0 & -A_i \\
        A_i & 0
    \end{pmatrix}\right).
\end{equation}
These are gauge-equivalent as $O(2)$ bundles via the gauge transformation
\begin{equation}
    g_i= \begin{pmatrix}
        1 & 0 \\
        0 & -1
    \end{pmatrix}\in \Z_2\hookrightarrow O(2),
\end{equation}
\begin{eqnarray}
    \begin{pmatrix}
        0 & -A_i \\
        A_i & 0
    \end{pmatrix} &=& \begin{pmatrix}
        1 & 0 \\
        0 & -1
    \end{pmatrix}\begin{pmatrix}
        0 & A_i \\
        -A_i & 0
    \end{pmatrix}\begin{pmatrix}
        1 & 0 \\
        0 & -1
    \end{pmatrix}^{-1},
    \\
    \begin{pmatrix}
        \cos\alpha_{ij} & -\sin\alpha_{ij} \\
        \sin\alpha_{ij} & \cos\alpha
    \end{pmatrix} &=& \begin{pmatrix}
        1 & 0 \\
        0 & -1
    \end{pmatrix} \begin{pmatrix}
        \cos\alpha_{ij} & \sin\alpha_{ij} \\
        -\sin\alpha_{ij} & \cos\alpha
    \end{pmatrix} \begin{pmatrix}
        1 & 0 \\
        0 & -1
    \end{pmatrix}^{-1}.
\end{eqnarray}

While we simply argued at the level of objects (bundles), this is functorial, since we know the action itself is functorial.

Thus, the field stack of the $\Z_2$-gauged theory is
\begin{equation}
\textbf{Fields}_{\Z_2}:=\textbf{H}\left(\Sigma,\textbf{B}_{\nabla}O(2)\right).
\end{equation}

Now, let us describe the gauging of the $U(1)$ $1$-form symmetry. Again, we need to consider the weak quotient under some group object. The relevant sequence is (cf. Eq.~(\ref{eq:bu1s1}))
\begin{equation}
    \textbf{B}U(1) \to \textbf{B}U(1) \to \textbf{B}\left(U(1)\dslash U(1)\right)=\textbf{B}{\rm INN}\left(U(1)\right) \to \textbf{B}^2 U(1).
\end{equation}
That is to say, the stack of fields for the $\textbf{B}U(1)$-gauged theory is 
\begin{equation}\label{eq:bu1gaugedfields}
    \textbf{Fields}_{\textbf{B}U(1)} = \textbf{H}\left(\Sigma,\textbf{B}_{\nabla}{\rm INN}U(1))\right),
\end{equation}
the stack of principal ${\rm INN}(U(1))$ bundles with (unadjusted) connection. The smooth $2$-group ${\rm INN}(U(1))$ admits the presentation as a crossed module of Lie groups
\begin{equation}\label{eq:innu1xmod}
    {\rm INN}(U(1)) = \left(U(1)\xrightarrow{\rm id} U(1) \right),
\end{equation}
with trivial action, fitting in the short exact sequence
\begin{equation}
    1\to 1 \to U(1) \xrightarrow{\rm id} U(1) \to 1 \to 1.
\end{equation}

Note that, \textit{topologically}, $\vert \textbf{B}{\rm INN}(U(1)) \vert \cong E(BU(1))\cong *$, for $E(BU(1))$ the universal bundle of the topological group $BU(1)$, so that there are no nontrivial flux/cohomology classes anymore. This is what has been referred to as obtaining a ``trivial theory'' after gauging (e.g. \cite{Pantev:2023dim,Witten:1995gf,Gukov:2006jk}).

To understand better in what sense this theory is trivial, let us describe the fields (\ref{eq:bu1gaugedfields}), namely, connections on principal ${\rm INN}(U(1))$ higher bundles over $\Sigma$, using the crossed module presentation (\ref{eq:innu1xmod}). A principal ${\rm INN}(U(1))$ bundle with (unadjusted) connection consists of $U(1)$-valued cochains $g_{ij}, z_{ijk}$, differential $1$-forms $a_i,c_{ij}$, and differential $2$-forms $b_i$, satisfying the identities
\begin{equation}\label{eq:u1coboundary}
    g_{ik} = g_{ij}g_{jk}z_{ijk},
\end{equation}
\begin{equation}\label{eq:u12cocycle}
     z_{ikl} z_{ijk} =z_{ijl} z_{jkl},
\end{equation}
\begin{equation}
    c_{ik} = c_{ij}+c_{jk} +d \log z_{ijk},
\end{equation}
\begin{equation}
    a_j = a_i + d\log g_{ij} - c_{ij},
\end{equation}
\begin{equation}
    b_j = b_i + dc_{ij},
\end{equation}
along with the flat fake-curvature condition
\begin{equation}
    0 = da_{i} + b_i.
\end{equation}
These connections all have trivial curvature $3$-form since
\begin{equation}
    h_i := db_i = 0,
\end{equation}
from the exactness of $b_i$.

In particular, the identities (\ref{eq:u1coboundary})-(\ref{eq:u12cocycle}) indicate that the underlying bundles are simply bundle gerbes with a \textit{choice} of trivialization. This explains the triviality \textit{at the level of topology.}

As for the connections, a gauge transformation is 
\begin{equation}
    \left(y_{ij},h_i, \lambda_i\right): \left(g_{ij}, z_{ijk}, a_i, c_{ij}, b_{i}\right) \xrightarrow{\sim} \left(g'_{ij}, z'_{ijk}, a'_i, c'_{ij}, b'_{i}\right),
\end{equation}
 given by $U(1)$-valued cochains $y_{ij},h_i$, and differential $1$-forms $\lambda_i$ satisfying
\begin{equation}
    z'_{ijk} y_{ij}y_{jk} = z_{ijk} y_{ik},
\end{equation}
\begin{equation}
    g'_{ij} = h_i^{-1} g_{ij} h_j \exp\left(y_{ij}\right),
\end{equation}
\begin{equation}
    c'_{ij} = c_{ij} + \lambda_j - \lambda_i +d\log y_{ij},
\end{equation}
\begin{equation}
    a'_i = a_i + d\log h_i -  \lambda_i
\end{equation}
\begin{equation}
    b'_i = b_i + d\lambda_i.
\end{equation}

We can regard principal $U(1)$ bundles with connection $(g_{ij},a_i)$ as principal ${\rm INN}(U(1))$ bundles with connection as
\begin{equation}
    (g_{ij}, a_i) \mapsto (g_{ij},1,a_i,0,-da_i).
\end{equation}
The action of a principal $U(1)$-bundle with \textit{flat} connection $(y_{ij},\lambda_i)$ is
\begin{equation}
    (y_{ij},\lambda_i) \cdot (g_{ij}, a_i) = (g_{ij}y_{ij},a_i + \lambda_i).
\end{equation}
As principal ${\rm INN}(U(1))$ bundles with connection, these are related by a gauge transformation
\begin{equation}
    (y_{ij},0,-\lambda_i): (g_{ij},1,a_i,0,-da_i) \xrightarrow{\sim} \left(g_{ij}y_{ij},0, a_i+\lambda_i,0,-da_i\right).
\end{equation}
This is the gauge transformation colloquially referred to as the ``combined transformation'' which causes the $U(1)$ gauge field $a$ to get ``eaten'' by the $2$-form $b$ (cf. \cite{Heidenreich:2020pkc,Aloni:2024jpb, Craig:2024dnl}).
Thus, the principal $U(1)$ bundles with connection become gauge equivalent to their images under the action of $U(1)$-bundles with flat connection in the gauged theory with fields (\ref{eq:bu1gaugedfields}). In other words, the original $U(1)$ bundles with connection are now equivariant under the $U(1)$ $1$-form symmetry.

\paragraph{Non-faithfully acting $\R$ $1$-form symmetry.} 

To finish this section, let us explore a simple example of a non-faithfully acting $1$-form symmetry, namely an $\R$ $1$-form symmetry acting by $U(1)$ left translations via the exponential map
\begin{equation}\label{eq:rsurjectiveu1}
    \Z\hookrightarrow \R \xrightarrow{\exp} U(1),
\end{equation}
indicating there is a trivially-acting $\Z$ $1$-form symmetry.

In more detail, we consider principal $\R$-bundles with flat connection $(x_{ij},A_i)$ acting on principal $U(1)$ bundles with connection $\left(g_{ij},a_i\right)$ as
\begin{equation}
    \left(x_{ij},A_i\right) \cdot \left(g_{ij},a_i\right) = \left(\exp\left(x_{ij} \right)g_{ij},  A_i + a_i\right),
\end{equation}
where in particular if $m_{ij} \in \Z$ then $\exp\left(m_{ij}\right)=1$.

To obtain the fields of the $\R$-gauged theory, we exhibit a sequence
\begin{equation}\label{eq:u1rgauged}
    \textbf{B}\R\to \textbf{B}U(1) \to {\cal X} \to \textbf{B}^2\R.
\end{equation}
This is accomplished by taking ${\cal X}:= \textbf{B}{\cal G}$ where $\cal G$ is the Lie $2$-group represented by the crossed module of Lie groups
\begin{equation}
    {\cal G} = \left(\R \xrightarrow{\exp} U(1) \right),
\end{equation}
with trivial action of $U(1)$ on $\R$, which fits in the exact sequence
\begin{equation}
    1 \to \Z \to \R \xrightarrow{\exp} U(1) \to 1 \to 1.
\end{equation}
Note that, topologically $\vert\textbf{B}{\cal G}\vert\cong B^2\Z$, and the sequence (\ref{eq:u1rgauged}) simply comes from continuing the delooping of the sequence (\ref{eq:rsurjectiveu1}) as
\begin{equation}
    \cdots \to \textbf{B}\R\to \textbf{B}U(1) \to \textbf{B}^2\Z \to \textbf{B}^2\R\to \cdots.
\end{equation}

We want, however, the explicit connection information. In Čech data, a principal $\cal G$ bundle with (unadjusted) connection is defined by $U(1)$-valued cochains $g_{ij}$, $\R$-valued cochains $x_{ijk}$, differential $1$-forms $a_i,c_{ij}$, and differential $2$-forms $b_i$, all together satisfying
\begin{equation}
    g_{ik} = g_{ij}g_{jk}\exp\left(x_{ijk} \right),
\end{equation}
\begin{equation}
     x_{ikl} + x_{ijk} =x_{ijl} + x_{jkl},
\end{equation}
\begin{equation}
    c_{ik} = c_{ij}+c_{jk} +d x_{ijk},
\end{equation}
\begin{equation}
    a_j = a_i + d\log g_{ij} - c_{ij},
\end{equation}
\begin{equation}
    b_j = b_i + dc_{ij},
\end{equation}
along with the flat fake-curvature condition
\begin{equation}
    0 = da_{i} + b_i,
\end{equation}
and zero curvature $3$-form as in the case of ${\rm INN}(U(1))$ connections.

A gauge transformation 
\begin{equation}
    \left(y_{ij},h_i, \lambda_i\right): \left(g_{ij}, x_{ijk}, a_i, c_{ij}, b_{i}\right) \xrightarrow{\sim} \left(g'_{ij}, x'_{ijk}, a'_i, c'_{ij}, b'_{i}\right)
\end{equation}
is given by $\R$-valued cochains $y_{ij}$, $U(1)$-valued cochains $h_{i}$, and differential $1$-forms $\lambda_i$ satisfying
\begin{equation}
    x'_{ijk} + y_{ij}+y_{jk} = x_{ijk}  + y_{ik},
\end{equation}
\begin{equation}
    g'_{ij} = h_i^{-1} g_{ij} h_j \exp\left(y_{ij}\right),
\end{equation}
\begin{equation}
    c'_{ij} = c_{ij} + \lambda_j - \lambda_i +dy_{ij},
\end{equation}
\begin{equation}
    a'_i = a_i + d\log h_i -  \lambda_i
\end{equation}
\begin{equation}
    b'_i = b_i + d\lambda_i.
\end{equation}

We can regard principal $U(1)$ bundles with connection $(g_{ij},a_i)$ as principal $\cal G$ bundles with connection as
\begin{equation}
    \left(g_{ij}, a_i\right) \mapsto \left(g_{ij},0,a_i,0,-da_i\right).
\end{equation}
The action of a principal $\R$-bundle with flat connection $(y_{ij},\lambda_i)$ is
\begin{equation}
    \left(y_{ij},\lambda_i\right) \cdot \left(g_{ij}, a_i\right) = \left(g_{ij}\exp(y_{ij}),a_i + \lambda_i\right).
\end{equation}
But as principal $\cal G$ bundles with connection, this is a gauge transformation
\begin{equation}
    \left(y_{ij},0,-\lambda_i\right): \left(g_{ij},0,a_i,0,-da_i\right) \xrightarrow{\sim} \left(g_{ij}\exp\left(y_{ij}\right),0, a_i+\lambda_i,0,-da_i\right).
\end{equation}
Thus, as before, the principal $U(1)$ bundles with connection become gauge equivalent to their images under the action of $\R$-bundles with flat connection in the gauged theory whose fields are
\begin{equation}
    \textbf{Fields}_{\textbf{B}\R} = \textbf{H}\left(\Sigma,\textbf{B}_{\nabla}{\cal G}\right) = \textbf{H}\left(\Sigma , \textbf{B}_{\nabla}\left(\R\xrightarrow{\exp} U(1) \right)\right),
\end{equation}
which at the level of topology are simply measured by integral cohomology classes ${\rm Map}(\Sigma, B{\cal G}) = {\rm Map}(\Sigma, B^2\Z) = H^2(\Sigma, \Z)$, and thus is not distinguishable from the original theory at this coarse level of analysis.  Ultimately, this is rooted in the fact that the higher group whose action was gauged, $\textbf{B}\R$, is topologically trivial $\vert \textbf{B}\R\vert \cong *$.

\section{B fields}\label{sec:bfields}

We now pass to our first example of a higher group, that of bundle gerbes, which appear in several different contexts in Physics (see e.g. \cite{Bunk:2021quu}). For example, in the absence of sources, B fields are globally modeled as connections on bundle $1$-gerbes (principal $U(1)$ $2$-bundles) \cite{gawkedzki2021topological}. These are described by smooth maps
\begin{equation}\label{eq:field2stack}
    \textbf{Fields}:= \textbf{H}\left(\Sigma,\textbf{B}^2_{\nabla}U(1)\right).
\end{equation}

An object therein is a tuple $\left(h_{ijk}, b_i,a_{ij} \right)$ such that
\begin{eqnarray}
    b_j &=& b_i + da_{ij},
    \\
    a_{ik} &=& a_{ij} + a_{jk} + d\log h_{ijk},
    \\
     h_{ikl} h_{ijk} &=& h_{ijl} h_{jkl}.
\end{eqnarray}

\paragraph{Symmetries.} The higher-form global symmetries for principal $\textbf{B}U(1)$ bundles are parameterized by flat morphisms
\begin{equation}
\textbf{Sym}\left(\Sigma,\textbf{B}^2U(1)\right)=\textbf{H}_{\flat} \left(\Sigma,{\rm Aut}\left(\textbf{B}^2U(1)\right)\right)
\end{equation}
for $\textbf{B}^2U(1)$ the second delooping of the smooth group $U(1)$.

Let us state the relevant objects, following the results of Section~\ref{sec:2-action}. The $2$-groupoid ${\rm Aut}\left(\textbf{B}^2 U(1)\right)$ is, succinctly,
\begin{equation}
  \label{UnderlyingGroupoidOfAutomorphism2GroupBU1}
 {\rm Aut}(\mathbf{B}^2U(1))
  =
  \left\{
  \begin{tikzcd}
    \phi
    \ar[
      rr,
      bend left=40,
      "{1}"{description, name=s}
    ]
    \ar[
      rr,
      bend right=40,
      "{1}"{description, name=t}
    ]
    \ar[
      from=s, to=t,
      Rightarrow,
      "{ m_{\bullet} }"
    ]
    &&
    \phi
  \end{tikzcd}
  \middle\vert\;
  \begin{aligned}
    \phi & \ \in {\rm Aut}\left(U(1)\right)=\Z_2
    \\
    m_{\bullet} &  \in U(1)
  \end{aligned}
  \right\}
  \mathrlap{\,.}
\end{equation}

An object $\Phi_{ijk}, \Phi'_{ijk} \in {\rm ob}\left( \textbf{Sym}\left(\Sigma,\textbf{B}^2U(1)\right)\right)$ in the symmetries is a pair
\begin{equation}
    \Phi_{ijk} = \left(z,h_{ijk} \right)
\end{equation}
where $z\in {\rm Aut}\left( U(1)\right)$ is globally-defined, and $h_{ijk}\in U(1)$ satisfies
\begin{equation}
     h_{ikl} h_{ijk} = h_{ijl} h_{jkl}.
\end{equation}
The center higher-form symmetry corresponds to the choice $z={\rm id}_{U(1)}$, describing a flat bundle gerbe.

In particular, the sequence (\ref{eq:autsequence}) becomes
\begin{equation}
    1\to \textbf{B}{\cal Z}(\textbf{B}U(1))\to \aut{\textbf{B}U(1)} \to \text{OUT}(\textbf{B}U(1)) \to 1.
\end{equation}
\begin{equation}
    \textbf{B}{\cal Z}(\textbf{B}U(1)) = \left(U(1)\xrightarrow{0}0\xrightarrow{0}0\right) = \textbf{B}^2U(1),
    \end{equation}
\begin{equation}\label{eq:2outbu1}
    \text{OUT}\left(\textbf{B}U(1)\right) = \left(0\xrightarrow{0}\Z_2\right) = \aut{U(1)}
\end{equation}
(ignoring the topmost information of $\text{OUT}(\textbf{B}U(1))$, which is trivial).

The composition and action on bundle gerbes, at the level of objects, is
\begin{equation}
    \otimes \left(\left(z, h_{ijk}\right), \left(z', h'_{ijk}\right)\right) = \left(zz',h_{ijk}z\left(h'_{ijk}\right)\right),
\end{equation}
\begin{equation}
    \cdot \left(\left(z, h_{ijk}\right),  h'_{ijk}\right) = h_{ijk}z\left(h'_{ijk}\right).
\end{equation}

As before, to incorporate connections, we treat the $\Z_2$ $0$-form and center symmetries separately, for simplicity.

The action of $z\in \Z_2$ in $\Phi_{ijk}$ is the charge conjugation
\begin{equation}\label{z2action-bu1}
    z\cdot \left(h_{ijk}',b_i,a_{ij}\right) = \left(z\left(h_{ijk}'\right),z^*b_i,z^*a_{ij}\right),
\end{equation}
which is functorial.

Similarly, the action of $\textbf{B}{\cal Z}(\textbf{B}U(1))$, as principal $\textbf{B}U(1)$ bundles with flat connections, on bundle gerbes extends to an action on bundle gerbes with connection as
\begin{equation}
    \left(h_{ijk},B_i,A_{ij}\right)\cdot \left(h_{ijk}',b_i,a_{ij}\right) = \left(h_{ijk}h'_{ijk},B_i+b_i, A_{ij}+a_{ij}\right).
\end{equation}

\paragraph{Action on Wilson surfaces and holonomies.} The charged ``operators'' under the $U(1)$ $2$-form symmetry are direct generalization of the Wilson loops (\ref{eq:wilsonloopaction}). These operators are called \textit{Wilson surfaces} and similarly come from \textit{holonomies}, more precisely \textit{higher} holonomies \cite{Schreiber:2008kcv}. We denote these as
\begin{equation}
    W_n(S,b)= {\rm exp}\left(in \oint_S b \right),
\end{equation}
\begin{equation}
    B\cdot W_n(S,b)= W_n(S) \, {\rm exp}\left(in  \oint_S B \right),
\end{equation}
where, as before, if $S$ can be deformed to another closed 2d surface $S'$, then they bound a three-dimensional volume $V$, and the flatness of $B$ implies that
\begin{equation}
    {\rm exp}\left(in  \oint_S B \right) = {\rm exp}\left(in  \oint_{S'} B \right),
\end{equation}
so that this is ``topological''. This is a straightforward analog of the $U(1)$ gauge theory case, including that the charged objects are functions on equivalence classes of the objects of the field stack.

\paragraph{Gauging.} Let us now, as before, discuss gauging. It is straightforward to see that gauging the $U(1)$ $2$-form symmetry simply corresponds to the once-categorified gauging of the $U(1)$ $1$-form symmetry of $U(1)$ gauge theory described in Section~\ref{ssec:u1}. In other words, the sequence
\begin{equation}
    \textbf{B}^2U(1) \to \textbf{B}^2U(1) \to \textbf{B}\left({\rm INN}\left(\textbf{B}U(1)\right)\right)\to \textbf{B}^3 U(1)
\end{equation}
that exhibits $\textbf{B}^2U(1)$ as the total space of a $\textbf{B}^2U(1)$-principal bundle indicates that the fields of the $\textbf{B}^2U(1)$-gauged theory are 
\begin{equation}
    \textbf{Fields}_{\textbf{B}^2U(1)}=\textbf{H}\left(\Sigma,\textbf{B}_{\nabla}({\rm INN}(\textbf{B}U(1))\right),
\end{equation}
the principal ${\rm INN}(\textbf{B}U(1))$ bundles with connection over $\Sigma$. The previous remarks concerning the triviality of the theory at the level of topology, as well as the B fields getting ``eaten'' by the gauge transformations of a $3$-form field $C$ similarly apply here.

Now, we focus on the $\Z_2$ in the outer automorphisms $2$-group ${\rm OUT}(\textbf{B}U(1))$ (\ref{eq:2outbu1}). Its action is (\ref{z2action-bu1}). We have the sequence
\begin{equation}
     \textbf{B}^2U(1) \to \left(\textbf{B}^2U(1)\right)\dslash \Z_2 \to \textbf{B}\Z_2,
\end{equation}
where $\left(\textbf{B}^2U(1)\right)\dslash \Z_2\cong \textbf{B}G$ is the moduli stack of principal $G$ bundles, where $G$ is the \textit{split 2-group} arising from the \textit{2-group semidirect product} \cite{Elg14}
\begin{equation}
    1 \to \textbf{B}U(1) \to \textbf{B}U(1)\rtimes\Z_2 =G \to \Z_2 \to 1.
\end{equation}
We can present $G$ as a crossed module of Lie groups
\begin{equation}\label{eq:aut-u1}
    G = \left( U(1) \xrightarrow{0} \Z_2 \right),
\end{equation}
equipped with the nontrivial action $\Z_2 \to {\rm Out}(U(1))$. This is nothing but the automorphism $2$-group $\aut{U(1)}$ we encountered before (\ref{aut-u1}). This means that gauging the outer $\Z_2$ $0$-form symmetry changes the gauge fields from principal $\textbf{B}U(1)$ bundles with connection, to principal $\aut{U(1)}$ bundles with connection. Therefore, the fields of the $\Z_2$-gauge theory are
\begin{equation}
    \textbf{Fields}_{\Z_2} = \textbf{H}(\Sigma,\textbf{B}_{\nabla}\aut{\textbf{B}U(1)}).
\end{equation}
These are truly \textit{nonabelian} bundles known as \textit{$U(1)$-gerbes}, encoding a bundle gerbe along with the choice of a \textit{band}, in this case described by $\Z_2$ \cite[Section~1.2.6]{Schreiber:2013pra} cf. \cite{nikolaus2015principal}.

\begin{remark}There is yet another interpretation for these $\Z_2$-equivariant bundle gerbes. These are also called \textit{Jandl gerbes} \cite{Schreiber:2005mi,Gawedzki:2008um}, which were originally defined in the context of WZW models on orientifolds. They have also been called Real bundle gerbes \cite{Hekmati:2016ikh}, due to their equivariance under the $\Z_2$-involution.
\end{remark}

To understand the effect of gauging on the original fields, it is convenient to replace the crossed module (\ref{eq:aut-u1}) with the \textit{equivalent} crossed module
\begin{equation}
    \widetilde{G}= \left( O(2) \xrightarrow{\pi\times 1} \Z_2\times \Z_2\right),
\end{equation}
where $\pi: O(2) \to \Z_2$ is defined by the extension
\begin{equation}
    1\to U(1) \hookrightarrow O(2) \xrightarrow{\pi}\Z_2 \to 1,
\end{equation}
and the action of $\Z_2\times\Z_2$ is the trivial one for the left factor, and conjugation by $\begin{pmatrix}
    1 & 0 \\ 0 & -1
\end{pmatrix}$ for the right factor. 
The crossed module $\widetilde{G}$ is equivalent to $G$ since it fits in the exact sequence
\begin{equation}
    1 \to U(1) \hookrightarrow O(2) \xrightarrow{\pi\times 1} \Z_2\times \Z_2 \to \Z_2 \to 1,
\end{equation}
with the same action of $\Z_2$ on $U(1)$ and trivial Postnikov extension class.

With the crossed module $\widetilde{G}$ at hand, we can now observe the weak quotient on the original gauge fields. For a B field with Čech data $(h_{ijk},b_i,a_{ij})$, these include as $\widetilde{G}$ principal bundles with connection as
\begin{equation}
    (h_{ijk},b_i,a_{ij}) \mapsto \left(1,\begin{pmatrix}
        \cos\alpha_{ijk} & \sin\alpha_{ijk} \\
        -\sin\alpha_{ijk} & \cos\alpha_{ijk}
    \end{pmatrix}, \begin{pmatrix}
        0 & b_i
        \\
        -b_i & 0
    \end{pmatrix}, \begin{pmatrix}
        0 & a_{ij}
        \\
        -a_{ij} & 0
    \end{pmatrix}\right),
\end{equation}
where the additional parameter $1$ to the left refers to the \textit{twisted} $(\Z_2\times\Z_2)$-bundle that is part of the data describing a principal $\widetilde{G}$-bundle.

\begin{remark}  
It is straightforward to generalize this to field theories whose fields are connections on bundle $(n-1)$-gerbes:
\begin{equation}
    \textbf{Fields}= \textbf{H}\left(\Sigma, \textbf{B}^n_{\nabla}U(1)\right).
\end{equation}
The smooth automorphism $(n+1)$-group of $\textbf{B}^nU(1)$ is readily computed as the truncation
\begin{equation}
    \tau_{\leq n}{\rm Aut}\left(\textbf{B}^nU(1)\right) = \aut{\textbf{B}^{n-1}U(1)    } = \left( U(1)^{(n)} \xrightarrow{0} 0 \to \cdots\to 0 \to \Z_2^{(0)}      \right),
\end{equation}
fitting in the sequence
\begin{equation}
    \textbf{B}{\cal Z}\left(\textbf{B}^{n-1}U(1)\right)\to \aut{\textbf{B}^{n-1}U(1)}\to \out{\textbf{B}^{n-1}U(1)},
\end{equation}
where
\begin{equation}
   \textbf{B}{\cal Z}\left(\textbf{B}^{n-1}U(1)\right) = \left(U(1)^{(n)}\to 0^{(n-1)}\to\cdots\to 0^{(0)} \right),
\end{equation}
\begin{equation}
    \out{\textbf{B}^{n-1}U(1)    } = \left(0^{(n)} \to \cdots\to 0 \to \Z_2^{(0)}      \right).
\end{equation}

The action of these symmetries is a straightforward generalization of the previous section. In particular, the gauging of the $\Z_2$ $0$-form symmetry gives rise to $\Z_2$-\textit{equivariant bundle $(n-1)$-gerbes}, or \textit{higher Jandl gerbes}, as explored in \cite{Fiorenza:2012mr}.
\end{remark}

\section{String $2$-groups}\label{sec:string2groups}

In Section~\ref{sec:bfields}, we studied the higher-form symmetries that arise as automorphisms of smooth higher groups for the simplest higher smooth group that is not a Lie $1$-group, the first delooping of $U(1)$. In this section, we study higher-form symmetries of less trivial yet familiar higher groups, the so-called \textit{String} groups ${\rm String}(G_k)$, which are nontrivial examples of higher central extensions, as summarized in Appendix~\ref{sapp:stringextension}. Thus, we will be looking at field theories whose fields are smooth maps
\begin{equation}
    \textbf{Fields} = \textbf{H}\left(\Sigma, \textbf{B}_{\nabla}{\rm String}(G_k)\right).
\end{equation}

The observations of Section~\ref{sec:general} (cf. Eq.(\ref{eq:higheraut})), as well as the construction of Section~\ref{sec:2-action}, tell us that, amongst the higher-form symmetries of a field theory whose fields are principal ${\rm String}(G_k)$ bundles, there are those coming from the \textit{center} ${\cal Z}({\rm String}(G_k))$ of ${\rm String}(G_k)$, given the inclusion
\begin{equation}
    \textbf{B}{\cal Z}\left({\rm String}(G_k)\right) \hookrightarrow \aut{{\rm String}(G_k)}.
\end{equation}
These subsymmetries are parameterized by smooth flat maps $\Sigma\to \textbf{B}({\cal Z}\left({\rm String}(G_k))\right)$, \textit{flat} principal ${\cal Z}({\rm String}(G_k))$ bundles, and act on the fields, the principal $\String{G_k}$-bundles, by tensoring. 

In \cite{2022arXiv220201271W}, the Drinfeld center ${\cal Z}\left({\rm String}\left(G_k\right)\right)$ is computed in great generality, thereby computing the center higher-form symmetry for string $2$-group gauge theory. We take these results, summarized in Appendix~\ref{sapp:centerstring}, as a starting point for two cases, those of $G=SU(n)$ and $G=U(1)$.

In the more tractable case of $G=U(1)$, we apply the results of Section \ref{sec:2-action} to describe the tensoring action, and discuss some refinements to a connective symmetry, namely, an action on principal ${\rm String}\left(U(1)_k\right)$ bundles with adjusted connection. We explore in detail sufficient circumstances under which such a lift does take place.

\subsection{$G=SU(n)$}

In this section, we specialize to the String group $\String{G_k}$ for $G=SU(n)$ and $k=1$, so that the fields are
\begin{equation}\label{eq:stringsunfields}
    \textbf{Fields}= \textbf{H}\left(\Sigma, \textbf{B}_{\nabla}\String{SU(n)_1}\right).
\end{equation}
Here, by $\nabla$ we mean the adjusted connection \cite{Fiorenza:2012tb, Rist:2022hci}, but we do not aim to refine the bundle symmetries to connective symmetries.

\paragraph{Motivation.} We start by motivating considering fields described by connections on principal $\String{SU(n)_k}$ bundles. In \cite{Tanizaki:2019rbk}, the authors describe a method for restricting to $SU(n)$ gauge fields $A$ whose instanton number, namely the class $[\tr(F(A)\wedge F(A))]\in H^4(\Sigma,\Z)$ satisfies
\begin{equation}\label{eq:instanton-restriction1}
    [\tr\left( F(A)\wedge F(A) \right)] = p[\omega] \in H^4(\Sigma,\Z),
\end{equation}
for $[\omega]\in H^4(\Sigma,\Z)$ some other integral cohomology class. This is accomplished by introducing a (locally-defined) differential $3$-form potential $c_3$ along with a Lagrange multiplier that enforces the differential equation
\begin{equation}\label{eq:instanton-restriction2}
     \tr\left( F(A)\wedge F(A) \right) = p \,dc.
\end{equation}
In particular, for $p=0$, this means restricting to $SU(n)$ gauge fields whose instanton number vanishes. 

We can reinterpret this proposal as follows. First, a global, \textit{topological} formulation of this is the commutative diagram
\begin{equation}
        \begin{tikzcd}
\Sigma \arrow[rr, "P"] \arrow[dd] &  & BSU(n) \arrow[dd, "c_2"] \arrow[lldd, "\eta", Rightarrow] \\
                             &  &                                                                             \\
* \arrow[rr]                 &  & B^4\Z                                                             
\end{tikzcd},
\end{equation}
which says the principal $SU(n)$ bundle $P$ over $\Sigma$ has trivial second Chern class in $\Sigma$. The continuous map $c_2: BSU(n) \to B^4\Z$ is the generator of the cohomology group $H^4_{\rm sing}(BSU(n),\Z)\cong \Z$. This admits a smooth refinement \cite{Fiorenza:2012tb} to maps of moduli stacks of connections
\begin{equation}
            \begin{tikzcd}
\Sigma \arrow[rr, "a"] \arrow[dd] &  & \textbf{B}_{\nabla}SU(n) \arrow[dd, "\textbf{cs}"] \arrow[lldd, "\eta", Rightarrow] \\
                             &  &                                                                             \\
\textbf{B}_{\rm triv,\nabla}^3U(1) \arrow[rr]                 &  & \textbf{B}^3_{\nabla}U(1)                                                             
\end{tikzcd},\label{diag:cs-comm-diag-su(n)}
\end{equation}
where $\textbf{cs}: \textbf{B}_{\nabla}SU(n) \to \textbf{B}^3_{\nabla}U(1)$, the Chern-Simons 2-gerbe with connection, generates the cohomology group $H^3(SU(n),U(1))\cong \Z$ (cf. Eq.~(\ref{eq:smooth3cohomology})), and where $\textbf{B}_{\rm triv,\nabla}^3U(1)$ is the moduli space of bundle $2$-gerbes with connection with trivial group cocycle information\footnote{Equivalently, the Deligne moduli stack $\Omega^{1\leq \bullet\leq 3}$ which assigns differential forms at degrees $1\leq k\leq 3$.}. 

The central identity described by the diagram (\ref{diag:cs-comm-diag-su(n)}) is
\begin{equation}
    h = \textbf{cs}(a) + d\eta
\end{equation}
where $\eta, h$ are \textit{globally-defined} differential $2$ and $3$-form, respectively, even though the Chern-Simons $3$-form
\begin{equation}
    \textbf{cs}(a): = \tr\left(a\wedge da + \tfrac{2}{3}a\wedge a\wedge a \right)
\end{equation}
is usually only locally-defined. The differential of this is
\begin{equation}
    dh = d \textbf{cs}(a) = \tr\left(F(a)\wedge F(a) \right).
\end{equation}
Since the Chern-Weil representative of the second Chern class of the $SU(n)$ bundle is exact, the second Chern class (instanton number) is trivial.

The smooth maps of diagram (\ref{diag:cs-comm-diag-su(n)}) factor through the pullback of the cospan by its universal property, so that one obtains a principal $\String{SU(n)_1}$-bundle with connection
\begin{equation}
\begin{tikzcd}
\Sigma \arrow[rrrd, "a", bend left] \arrow[rddd, bend right] \arrow[rd, "x", dashed] &                                                                 &  &                                                                                     \\
                                                                                & \textbf{B}_{\nabla}\text{String}(SU(n)_1) \arrow[rr] \arrow[dd] &  & \textbf{B}_{\nabla}SU(n) \arrow[dd, "\textbf{cs}"] \arrow[lldd, "\eta", Rightarrow] \\
                                                                                &                                                                 &  &                                                                                     \\
                                                                                & {\textbf{B}_{\rm triv,\nabla}^3U(1)} \arrow[rr, hook]           &  & \textbf{B}^3_{\nabla}U(1)                                                          
\end{tikzcd}.
\end{equation}

In other words, restricting to $SU(n)$ bundles with connection with trivial instanton number means working with $SU(n)$ bundles with connection coming from principal $\String{SU(n)_1}$ bundles with connection. See \cite{Perez-Lona:2025add} for other instances of instanton restrictions and principal bundles of \textit{covering} groups.

\paragraph{Braided center.} Having motivated looking at fields which are connections of principal $\String{SU(n)_1}$ bundles as a way of restricting to $SU(n)$ connections with trivial instanton number, we now state its center higher-form symmetry. 

According to the result cited in Appendix~\ref{sapp:centerstring} from \cite{2022arXiv220201271W}, the center
\begin{equation}
    {\cal Z}(SU(n),1):={\cal Z}(\String{SU(n)_1})
\end{equation}
satisfies
\begin{eqnarray}
    \pi_0\left({\cal Z}(SU(n),1)\right) &=& \Z_n,
    \\
    \pi_1\left({\cal Z}(SU(n),1)\right) &=& U(1),
\end{eqnarray}
whose extension class $H^3(\Z_n,U(1))$ is the pullback of the extension class along the inclusion $\Z_n=Z(SU(n))\hookrightarrow SU(n)$, and whose braiding (Eqs. (\ref{eq:braiding1})-(\ref{eq:braiding2})) in this case is given by the form \cite{2022arXiv220201271W}
\begin{equation}\label{eq:sun-braiding}
    I(X,Y) = \tr(XY).
\end{equation}

That is to say, the center higher-form symmetry of the $\sigma$-model with fields principal $\String{SU(n)_1}$ bundles corresponds to a higher central extension ${\cal Z}(SU(n),1)$
\begin{equation}\label{eq:stringsuncentersequence}
  1\to   \textbf{B}U(1) \to {\cal Z}(SU(n),1)\to \Z_n\to 1,
\end{equation}
which can be understood as a $\Z_n$ group with a nontrivial associator given by the extension $3$-class valued in $U(1)$. This center higher-form symmetry acts by tensoring flat principal ${\cal Z}(SU(n),1)$ bundles with principal $\String{SU(n)_1}$ bundles, as the symmetry is parameterized by flat maps $\Sigma \to {\mathbf{B}{\cal Z}(SU(n),1)}$ as
\begin{equation}\label{eq:suncentersyminclusion}
\begin{tikzcd}
                           & \Sigma \arrow[d]                                  &                \\
\mathbf{B}^2U(1) \arrow[r] & {\mathbf{B}{\cal Z}(SU(n),1)} \arrow[r] \arrow[d] & \mathbf{B}\Z_n \\
                           & \text{AUT}(\String{SU(n)_1})                      &               
\end{tikzcd}.
\end{equation}

Notice that, because the group being extended is finite, the only connective data is the one proper to a $U(1)$ $2$-form symmetry. However, this does not seem to be sufficient to ensure that the center higher-form symmetry lifts to a connective symmetry. An in-depth analysis of this issue is beyond the scope of this paper. However, the case $G=U(1)$ discussed in Section~\ref{sssec:stringu1} is sufficiently tractable to allow for a concrete analysis in terms of Čech cochains and connections, exemplifying this problem of lifting to connective symmetries. 

\paragraph{Gaugings.} It suffices to mention the two possible gaugings proper to the center higher-form symmetry. First one can gauge the $U(1)$ $2$-form symmetry subgroup, coming from the inclusion of stacks
\begin{equation}
    \textbf{B}^2U(1) \hookrightarrow \textbf{B}{\cal Z}(SU(n),1),
\end{equation}
which just acts on the fields (\ref{eq:stringsunfields}) by tensoring with a bundle $2$-gerbe with a flat connection. The effect of gauging can be seen from the sequence
\begin{equation}
    \textbf{B}^2U(1) \to \textbf{B}\String{SU(n)_1)}\to \textbf{B}SU(n) \to \textbf{B}^3U(1),
\end{equation}
exhibiting a quotient stack
\begin{equation}
    \textbf{B}SU(n) \cong \left(\textbf{B}\String{SU(n)_1)}\right)\dslash \textbf{B}^2U(1).
\end{equation}
As expected, gauging the $U(1)$ $2$-form symmetry gives fields
\begin{equation}
    \textbf{Fields}_{\textbf{B}^2U(1)}= \textbf{H}\left(\Sigma, \textbf{B}_{\nabla}SU(n)\right).
\end{equation}
In other words, we recover the principal $SU(n)$ bundles with connection with \textit{no instanton number restriction}. Note that, here, we are talking about connections, not because we have shown that the symmetry is connective but because we already know there exists a sequence \cite{Fiorenza:2012tb}
\begin{equation}
    \textbf{B}^2U(1) \to \textbf{B}_{\nabla}\String{SU(n)_1}\to \textbf{B}_{\nabla}SU(n) \xrightarrow{\textbf{c}_2}\textbf{B}^3U(1),
\end{equation}
that exhibits $\textbf{B}_{\nabla}\String{SU(n)_1}$ as the total space of a principal $\textbf{B}^2U(1)$ bundle over $\textbf{B}_{\nabla}SU(n)$.

The other interesting gauging proper to the center higher-form symmetry is ${\cal Z}(SU(n),1)$ itself, which is simply
\begin{equation}
    \textbf{B}{\cal Z}(SU(n),1)\to \textbf{B}\String{SU(n),1} \to \textbf{B}PSU(n),
\end{equation}
meaning the fields of the gauged theory, \textit{at the level of smooth principal bundles}, are
\begin{equation}
    \textbf{Fields}_{ \textbf{B}{\cal Z}(SU(n),1)} = \textbf{H}\left(\Sigma, \textbf{B}PSU(n)\right),
\end{equation}
principal $PSU(n)=SU(n)/Z(SU(n))$ bundles. This can be understood as simultaneously gauging the $\Z_n$ $1$-form symmetry coming from $SU(n)$ along with the $U(1)$ $2$-form symmetry. We thus eliminate the nontrivial $2$-group structure, as above, and thus recovering nonzero second Chern classes, while also gauging the center $1$-form symmetry.

In general, the smooth $2$-group sequence (\ref{eq:stringsuncentersequence}) does not split, so there is no sensible way to regard $Z(SU(n))=\Z_n$ as a global ($1$-form) symmetry on its own. Intuitively, this also means that the $\Z_n$ acts nontrivially on both the $SU(n)$ ``part'' of the bundle as well as on the higher $U(1)$ ``part''. This indicates that the $\Z_n$ $1$-form symmetry is a \textit{projective} or \textit{anomalous} symmetry (e.g. \cite[Section 4.2]{Gwilliam:2025vdu}, \cite{VanDyke:2023waj}).

\begin{remark}
    It is instructive to compare our findings with the symmetries computed in the $SU(n)$ instanton restriction construction of \cite[Section 3]{Tanizaki:2019rbk}. The method presented therein, as previously mentioned (cf. Eq's.~(\ref{eq:instanton-restriction1}), (\ref{eq:instanton-restriction2})), forces the $SU(n)$ instanton number to be a multiple of $p\in \Z$ by introducing a local $3$-form field, along with an appropriate Lagrange multiplier. The authors, in this setting, then argue that the resulting theory has a $\Z_n^{(1)}$ $1$-form symmetry, and a $\Z_{p}^{(3)}$ $3$-form symmetry, which together form a nontrivial \textit{$4$-group}. Specializing to $p=0$, this seems to suggest that there should be a $(\Z/(0\Z)=\Z)$ $3$-form symmetry. By contrast, we found that the center higher-form symmetry of $\String{SU(n)_1}$ involves a $U(1)$ $2$-form symmetry. The discrepancy is due to the fact that, topologically, $B^2U(1)\cong B^3\Z$, so that these two topological higher groups are not distinguishable. However, this equivalence does not lift to an equivalence of smooth higher stacks, and so the $\Z$ $3$-form symmetry is actually smoothly refined to a $U(1)$ $2$-form symmetry. Likewise, the $4$-group found therein is refined to a \textit{smooth} $3$-group, in this case as the delooping of a \textit{smooth braided} $2$-group (cf. Eq's.~(\ref{eq:sun-braiding})-(\ref{eq:suncentersyminclusion})). This highlights the importance of performing the analysis at the level of smooth higher stacks, especially if one wants to do a dynamical gauging. 

We highlight, however, that the symmetries we described are indeed smooth but at the level of principal bundles, not yet with connection. While the arguments given in \cite{Tanizaki:2019rbk} might suggest the symmetry can indeed be made connective, we leave a precise verification of this statement for future work.

\end{remark}

\subsection{$G=U(1)$}\label{sssec:stringu1}

We now specialize to $G=U(1)$. This case is particularly interesting because $U(1)$ is its own center, and is not simply-connected, so that we expect a smooth (non-discrete) center higher-form symmetry involving connective data, as we already encountered for $U(1)$ gauge theory and its action by flat connection shifts.

The diagram defining the moduli stack $\textbf{B}_{\nabla}\String{U(1)_k}$ is
\begin{equation}
    \begin{tikzcd}
\Sigma \arrow[rrrd, "a", bend left] \arrow[rddd, bend right] \arrow[rd, "x", dashed] &                                                                &  &                                                                                    \\
                                                                                & \textbf{B}_{\nabla}\text{String}(U(1)_k) \arrow[rr] \arrow[dd] &  & \textbf{B}_{\nabla}U(1) \arrow[dd, "\textbf{k}\cdot\textbf{cs}"] \arrow[lldd, "\eta", Rightarrow] \\
                                                                                &                                                                &  &                                                                                    \\
                                                                                & {\textbf{B}_{\rm triv,\nabla}^3U(1)} \arrow[rr, hook]          &  & \textbf{B}^3_{\nabla}U(1)                                                         
\end{tikzcd},\label{diag:u1stringpullback}
\end{equation}
which on the connections says that
\begin{equation}\label{eq:u1stringconnection}
    h = k\,\textbf{cs}(a)+d\eta,
\end{equation}
for $k\in \Z$, and $a:\Sigma\to \textbf{B}_{\nabla}U(1)$ the connection on a principal $U(1)$ bundle. The identity (\ref{eq:u1stringconnection}), together with the fact that $a$ is connection, implies the $\Z$-quantized flux equations
\begin{eqnarray}
    dF(a) &=& 0, \label{eq:u1stringflux1}
    \\
    dh &=& k\left(F(a)\wedge F(a)\right). \label{eq:u1stringflux2}
\end{eqnarray}

\paragraph{Motivation.} These flux equations have appeared in at least two different places. This kind of differential form/Bianchi identity data happens to underlie a different topological structure, namely cocycles in \textit{differential $2$-Cohomotopy} \cite{Fiorenza:2015gla}, \cite[Example 9.3]{Fiorenza:2020dcp}, a nonabelian cohomology theory (and thus outside the scope of higher principal connections) where the fields are maps ${\rm Map}(\Sigma,S^2)$ valued in the $2$-sphere. This has recently been used, on the condensed matter theory side, in \cite{Sati:2025iiw} in the context of exotic quantization describing Fractional Quantum Hall Anyons. On the supergravity side, these identities arise from the ``democratic formulation'' of minimal supergravity in five dimensions (see e.g. \cite{Gauntlett:2002nw}), which in many respects is the lower-dimensional analog of eleven-dimensional supergravity C fields. Briefly, the bosonic Lagrangian of such theory is 
\begin{equation}
{\cal L} = -\tfrac{1}{4} R\star 1 - \tfrac{1}{2}F(a)\wedge\star F(a) - \tfrac{2}{3\sqrt{3}} F(a)\wedge F(a)\wedge a,
\end{equation}
which gives the equation of motion
\begin{equation}
    d\star F(a) = -\tfrac{2}{\sqrt{3}}F(a)\wedge F(a).
\end{equation}
Defining $h:= -\tfrac{2}{\sqrt{3}} \star F(a)$, along with the fact that $a$ is a $U(1)$ connection, gives the flux identities (\ref{eq:u1stringflux1})-(\ref{eq:u1stringflux2}) derived above. This means that the \textit{pre-metric} (cf. \cite{Sati:2024klw}) bosonic gauge fields can be\footnote{This is not to say that this is the unique option. Quantization laws differing only by torsion factors still give rise to the same flux identities, see \cite{Sati:2024klw} for more on this.} described as a connection on a principal $\String{U(1)_1}$ bundle. Therefore, connections on principal $\String{U(1)_1}$ bundles are indeed physically relevant, reason for which we are interested in knowing their symmetries and gaugings.

\paragraph{Crossed module presentation.} As explained in \cite[Section 4]{2022arXiv220201271W}, the smooth $2$-group $\String{U(1)_k)}$ is an example of the \textit{categorical tori} defined in \cite{Ganter:2014zoa}, which are $\textbf{B}U(1)$ central extensions of the form
\begin{equation}\label{eq:categorical-torus-rank-r}
    1\to \textbf{B}U(1) \to \Gamma \to T^r \to 1
\end{equation}
for $T^r=U(1)^r$ as a Lie group. These are also known as the \textit{T-duality} $2$\textit{-groups} \cite{Fiorenza:2012ec,Fiorenza:2016oki,Nikolaus:2018qop}. Here, we are concerned with the rank $r=1$ categorical torus, or \textit{categorical circle}. We use the notation $\Lambda = {\rm Hom}(U(1),U(1))=\Z$ for that characters, and $\Pi= \Lambda^{\vee}=\Z$ for the cocharacters, which fit in the short exact sequence
\begin{equation}
    \Pi\hookrightarrow \R \xrightarrow{\exp} U(1),
\end{equation}
The (topological) extension classes $H^4_{\rm sing}(BU(1),\Z)$ in this case are identified with
\begin{equation}
    H^4_{\rm sing} (BU(1),\Z) = {\rm Sym}^2(\Lambda),
\end{equation}
the $\Z$-bilinear forms 
\begin{equation}
    I: \R\times \R \to \R
\end{equation}
such that image of the domain restriction to $\Pi$ is contained in $2\Z\subset \Z$. Such forms can be decomposed as $I= -(J+J^t)$ for $J$ a $\Z$-bilinear form on $\R$ that restricts to an integer-valued form on $\Pi$. Furthermore, they induce maps $\Pi \to \Lambda$ as $\tau(\pi):= I(\pi,-)$.

More in detail, the symmetric $\Z$-bilinear forms $I(-,-)$ are constructed as
\begin{equation}
    I(x,y)= -2n\cdot xy,
\end{equation}
where $\Lambda=\Z=\ker(\exp)$ with $\exp:r\mapsto \exp(2\pi i r)$, and $n\in \Z$ is the extension class. For this form, a corresponding $\Z$-bilinear form $J$ is
\begin{equation}
    J(x,y) = n\cdot xy.
\end{equation}
The symmetric $\Z$-bilinear function $I$ induces linear forms
\begin{equation}
    \tau (x) := I(x,-) = -\left(J(x,-)+J(-,x)\right).
\end{equation}

The categorical torus (\ref{eq:categorical-torus-rank-r}) of rank $r=1$ (categorical circle) and extension class $k$ (described by the $\Z$-bilinear form $J$) is described as a crossed module of Lie groups
\begin{equation}\label{eq:suu-xmod}
 \String{U(1)_k}=  \left(\delta:= \imath_{\Z}: \Z\times U(1) \to  \R\right),
\end{equation}
where $\imath_{\Z}:(m,z)\mapsto m\in \R$, equipped with the action
\begin{eqnarray}
    \alpha&:& \R\times \Z\times U(1) \to  \Z\times U(1),
    \\
    && (x,m,z) \mapsto \left(m,z\cdot \exp\left( J(m,x)\right)\right),
\end{eqnarray}
fitting in the short exact sequence
\begin{equation}
    1\to U(1) \hookrightarrow  \left(\Z \times U(1)\right) \to \R \to  U(1)\to 1.
\end{equation}

\paragraph{Fields.} Let us now move on to the fields, the principal $\String{U(1)_k}$ bundles with connection. Following the general definition of a $2$-group bundle with connection in \cite{Rist:2022hci} (summarized in Appendix~\ref{app:adjbundles}), for a Čech cover on $\Sigma$, we describe a principal $\suu$-bundle with adjusted connection as consisting of: cochains $x_{ij}\in \R$, $m_{ijk}\in\Z$, $z_{ijk}\in U(1)$, differential $1$-forms $c_{ij},a_{i}$, and differential $2$-forms $b_i$, satisfying the identities
\begin{equation}\label{eq:twistedsuucocycle1}
   x_{ik} = m_{ijk}+x_{ij}+x_{jk},
\end{equation}
\begin{equation}\label{eq:twistedsuucocycle2}
    \left(m_{ikl}+m_{ijk},z_{ikl} z_{ijk}\right) =\left(m_{ijl}+m_{jkl}, z_{ijl} z_{jkl} \exp\left(J(m_{jkl},x_{ij})\right)\right),
\end{equation}
\begin{equation}\label{eq:cforms-identity}
    c_{ik} = c_{jk}+c_{ij} -z_{ijk}\nabla_{a_i}z^{-1}_{ijk}=c_{jk} + c_{ij} + d\log z_{ijk} +J\left(m_{ijk},a_i\right),
\end{equation}
\begin{equation}\label{eq:a1form-identity}
   a_j = a_i + dx_{ij},
\end{equation}
\begin{equation}\label{eq:bforms-identity}
    b_j = b_i + d\,c_{ij} - \kappa\left(-x_{ij},da_i\right),
\end{equation}
where $\nabla_{a_i}:=d+a_i\triangleright$. The connection has curvature forms
\begin{eqnarray}
    f_i &=& da_i,
    \\
    h_i &=& d b_i - \kappa\left(a_i,f_i\right). \label{eq:3curvaturesuu}
\end{eqnarray}
with adjustments
\begin{eqnarray}\label{eq:adj1}
    \kappa:& \R\times {\rm Lie}(\R) &\to \text{Lie}\left(\Z\times U(1)\right),
    \\
    & (x,dx') &\mapsto nxdx',
\end{eqnarray}
\begin{eqnarray}
    \kappa:& \rm{Lie}(\R)\times \rm{Lie}(\R)&\to {\rm Lie}(\Z\times U(1)),
    \\
    & (dx,dx') &\mapsto ndxdx',\label{eq:adj2}
\end{eqnarray}
where the $\Z$-bilinear form $J$ that determines the extension class is associatied with the integer $n$ as
\begin{eqnarray}
    J:& \R \times \R &\to \R,
    \\
    & (x,x') &\mapsto nxx'.\label{eq:jform-n}
\end{eqnarray}

In particular, the identity (\ref{eq:twistedsuucocycle1}) is represented by the simplex
\begin{equation}
    \begin{tikzcd}
                                                           &  &    &  & \bullet                                                                         \\
                                                           &  &    &  &                                                                                 \\
                                                           &  & {} &  &                                                                                 \\
                                                           &  &    &  &                                                                                 \\
\bullet \arrow[rrrr, "x_{ij}"'] \arrow[rrrruuuu, "x_{ik}"] &  &    &  & \bullet \arrow[uuuu, "x_{jk}"'] \arrow[lluu, "{(m_{ijk},z_{ijk})}", Rightarrow]
\end{tikzcd},\label{eq:2cocyclesimplex}
\end{equation}
\begin{equation*}
      x_{ik} = m_{ijk}+x_{ij}+x_{jk},
\end{equation*}
whereas the identity (\ref{eq:twistedsuucocycle2}) is the equality of simplices
\begin{equation}
    \begin{tikzcd}
\bullet \arrow[rrrr, "x_{jk}"] \arrow[rrdd, "{(m_{ijk},z_{ijk})}" description, Rightarrow]    &  &    &  & \bullet \arrow[dddd, "x_{kl}"] \arrow[lldddd, "{(m_{ikl},z_{ikl})}" description, Rightarrow] &   & \bullet \arrow[rrrr, "x_{jk}"] \arrow[rrrrdddd, "x_{jl}" description] \arrow[rrdddd, "{(m_{ijl},z_{ijl})}" description, Rightarrow] &  &    &  & \bullet \arrow[dddd, "x_{kl}"] \arrow[lldd, "{(m_{jkl},z_{jkl})}" description, Rightarrow] \\
                                                                                              &  &    &  &                                                                                              &   &                                                                                                                                     &  &    &  &                                                                                            \\
                                                                                              &  & {} &  &                                                                                              & = &                                                                                                                                     &  & {} &  &                                                                                            \\
                                                                                              &  &    &  &                                                                                              &   &                                                                                                                                     &  &    &  &                                                                                            \\
\bullet \arrow[uuuu, "x_{ij}"] \arrow[rrrr, "x_{il}"'] \arrow[rrrruuuu, "x_{ik}" description] &  & {} &  & \bullet                                                                                      &   & \bullet \arrow[rrrr, "x_{il}"'] \arrow[uuuu, "x_{ij}"]                                                                              &  & {} &  & \bullet                                                                                   
\end{tikzcd}\label{eq:3cocyclesimplex}
\end{equation}

\begin{equation*}
    (m_{ikl}+m_{ijk},z_{ikl} z_{ijk}) =\left(m_{ijl}+m_{jkl}, z_{ijl} z_{jkl} \exp\left(J\left(m_{jkl},x_{ij}\right)\right)\right),
\end{equation*}
where $x_{ij}\cdot z_{jkl}=z_{jkl} \exp\left(J\left(m_{jkl},x_{ij}\right)\right)$ is the action of $x_{ij}$ on $z_{jkl}$.

\paragraph{Braided center.}  The center higher-form symmetry is described by flat principal ${\cal Z}\left({\rm String}\left({U(1)_k}\right)\right)$  bundles. The Drinfeld center ${\cal Z}\left({\rm String}\left({U(1)_k}\right)\right)$, as computed in \cite{2022arXiv220201271W}), is described by the crossed module of Lie groups
\begin{equation}\label{eq:centerk}
    {\cal Z}_k := \left(\delta: \Z\times U(1) \to \R\oplus \Z \right)
\end{equation}
\begin{equation}
    \delta: (m,z) \mapsto (m,\tau(m)),
\end{equation}
where the additional factor of $\Z$ correspond to the characters $\Z={\rm Hom}(U(1),U(1))$ of $U(1)$. The action is
\begin{eqnarray}
    \alpha&:& \left(\R\oplus\Z\right)\times \Z\times U(1) \to  \Z\times U(1),
    \\
    && (x,\lambda,m,z) \mapsto \left(m,z\cdot \exp\left( J(m,x)\right)\right),
\end{eqnarray}
which can be checked to satisfy the crossed module identities.
The \textit{braiding} on this $2$-group is given by 
\begin{equation}
   \beta_{(x,\lambda),(x',\lambda')} = \left(0,\lambda\left(x'\right)\exp\left(J\left(x',x\right)\right)\right):\left(x+x',\lambda+\lambda'\right)\xrightarrow{\sim} \left(x'+x,\lambda'+\lambda\right).
\end{equation}
\begin{remark}
Due to a difference of conventions, the pair $\left(x,\lambda\right)$ here and in \cite{2022arXiv220201271W} corresponds to the $1$-morphism
\begin{equation}
    \left(x,\lambda^{-1}\right) : {\rm id}_{\textbf{B}{\rm String}\left(U(1)_k\right)} \to {\rm id}_{\textbf{B}{\rm String}\left(U(1)_k\right)}
\end{equation}
in ${\rm Aut}\left(\textbf{B}{\rm String}\left(U(1)_k\right)\right)$ of Section~\ref{sec:2-action}. Given that the half-braidings $\lambda$ always enter the formulae with a negative exponent (cf. Eq's. (\ref{eq:composition-sym-sigma}), (\ref{eq:action-sym-sigma})), it is equivalent but simpler to use this definition of $\lambda$ in this particular case so as to not introduce the negative exponent.
\end{remark}

\begin{remark}\footnote{We thank Konrad Waldorf for raising this question.}
It might seem puzzling at first that even though both $U(1)$ and $\textbf{B}U(1)$ are abelian, the center symmetry (\ref{eq:centerk}) is not the string $2$-group itself, nor does it even contain the string $2$-group. This is because ${\rm String}\left(U(1)_k\right)$ is not braided in general, as we now show.
\end{remark}

\begin{propo}[Nontrivial categorical circles are \textit{not} braided.]
    Let $\String{U(1)_k}$ be a categorical circle, a higher central extension
    \begin{equation}
        1 \to \textbf{B}U(1) \to \String{U(1)_k} \to U(1) \to 1,
    \end{equation}
    with extension class $k\in H^3(\textbf{B}U(1), U(1))$.

    The extension class $k$ obstructs the existence of a braiding on $\String{U(1)_k}$.
\end{propo}
\begin{proof}
   First, note that for an arbitrary categorical group $\cal G$, its Drinfeld center $\cal Z(G)$  has as its objects pairs $(X,b)$ where $X$ is an object in $\cal G$ and $b$ is a half-braiding. If $\cal G$ can be equipped with a braiding, then one can construct an inclusion braided monoidal functor $\cal G \to Z(G)$ which, in particular, specifies an injective group homomorphism
    \begin{equation}
        \phi: \pi_0({\cal G}) \to \pi_0 ({\cal Z(G)})
    \end{equation}
    from the group of equivalence classes of objects in $\cal G$ to that associated with $\cal Z(G)$.

    For the case at hand, this group is (cf. Eq.~\ref{eq:pi0zk})
    \begin{equation}
        \pi_0\left( {\cal Z}_k \right) = \frac{\R \oplus \Z}{\Z},
    \end{equation}
where the pairs $(x,\lambda)\in \R \oplus \Z$ are identified as
\begin{equation}
    (x,\lambda ) \sim \left(x+m, \lambda + \tau(m)\right).
\end{equation}
If the extension class $k$ is nonzero, then 
\begin{equation}
    \pi_0\left( {\cal Z}_{k\neq 0} \right) \cong \R\times \Z_m
\end{equation}
for some $m\in \Z$ that depends on $k$.

Now, suppose $\String{U(1)}_{k\neq 0 }$ admits a braiding. This would imply an \textit{injective} group homomorphism
\begin{equation}
    \phi: \pi_0\left(\String{U(1)_{k\neq 0}} \right) = U(1) \to \R\times \Z_m = \pi_0\left( {\cal Z}_{k\neq 0}\right),
\end{equation}
which does not exist, since the sequence
\begin{equation}
    1 \to \Z \to \R \to U(1) \to 1
\end{equation}
does not split.

On the other hand, if the extension class is zero, then $\tau(m)=0$ for all $m\in \Z$, and the group of classes of objects is instead
\begin{equation}
    \pi_0\left({\cal Z}_{k=0}\right) = U(1)\times \Z,
\end{equation}
and an injective group homomorphism
\begin{equation}
    \phi_0: \pi_0\left(\String{U(1)_{k=0}} \right) = U(1) \to U(1)\times \Z = \pi_0\left( {\cal Z}_{k=0}\right)
\end{equation}
amounts to a choice of character in ${\rm Hom}(U(1),U(1))=\Z$, which itself determines the braiding \cite{quinn1998group,braunling2021quinn}.   
\end{proof}

This explains why the center higher-form symmetry, controlled by ${\cal Z}_k$, does not have a subsymmetry controlled by flat principal $\String{U(1)_k}$  bundles.

As before, the parameters of the center higher-form symmetry correspond to principal (higher) bundles with flat connections. It is easy to see that a principal ${\cal Z}_k$ bundle with connection is determined by the same data of a $\String{U(1)_k}$ (cf. Eq.'s (\ref{eq:twistedsuucocycle1})-(\ref{eq:3curvaturesuu})), along with additional cochains $\lambda_{ij}\in \Z$ satisfying the twisted cocycle identity
\begin{equation}\label{eq:lambda-identity}
    \lambda_{ik} = \lambda_{ij}+\lambda_{jk}+\tau(m_{ijk}),
\end{equation}
which if we denote such characters as
\begin{equation}
    \lambda_{ij}(x'):= \exp\left(n_{ij}x' \right)
\end{equation}
for $n_{ij}\in \Z$, becomes
\begin{equation}\label{eq:nlambda}
    n_{ik} = n_{ij}+n_{jk} -2m_{ijk}.
\end{equation}
We take the same adjustments (\ref{eq:adj1})-(\ref{eq:adj2}).

\paragraph{Braided center action.} The ${\cal Z}_k$ bundles act by left-tensoring with $\suu$ bundles. Drawing from Section~\ref{sec:2-action}, let us construct the composition and action.

Let $P,P'$ be a pair of principal ${\cal Z}_k$ bundles specified by the data $(x_{ij},\lambda_{ij},m_{ijk},z_{ijk})$,  $(x'_{ij},\lambda'_{ij},m'_{ijk},z'_{ijk})$. Their image under the composition $2$-functor in Proposition~\ref{propo:gray-monoid} is
\begin{equation}\label{eq:tensoredcochains}
   \otimes\left(P,P'\right)= (X_{ij},\Lambda_{ij},M_{ijk},Z_{ijk}),
\end{equation}
\begin{equation}\label{eq:zk-composed-cochains1}
X_{ij} = x_{ij}+x'_{ij}, \: \:     \Lambda_{ij} = \lambda_{ij}+\lambda'_{ij},  \: \: M_{ijk} = m_{ijk}+m'_{ijk},
\end{equation}
\begin{eqnarray}\label{eq:zk-composed-cochains2}
     Z_{ijk} &=& z_{ijk}z'_{ijk}\lambda_{jk}(x'_{ij})\exp\left(J(x'_{ij},x_{jk})+J(m'_{ijk},x_{ik}) \right).
\end{eqnarray}
Moreover, for $P= (x_{ij},\lambda_{ij},m_{ijk},z_{ijk})$ a principal ${\cal Z}_k$ bundle, and $P'= (x'_{ij},m'_{ijk}, z'_{ijk})$ a principal $\String{U(1)_k}$ bundle, the image under the action $2$-functor from Proposition~\ref{propo:fields-mod} is
\begin{equation}\label{eq:tensoredcochainsaction}
  \cdot\left( P, \left(x'_{ij},m'_{ijk},z'_{ijk} \right) \right) =   (X_{ij},M_{ijk},Z_{ijk}),
\end{equation}
\begin{equation}\label{eq:zk-action-cochains1}
 X_{ij}=x_{ij}+x'_{ij},   \: \: M_{ijk}=m_{ijk}+m'_{ijk},
\end{equation}
 \begin{equation}\label{eq:zk-action-cochains2}
    Z_{ijk} = z_{ijk}z'_{ijk}\lambda_{jk}(x'_{ij})\exp\left(J(x'_{ij},x_{jk})+J(m'_{ijk},x_{ik}) \right).
\end{equation}

\paragraph{Connective center higher-form symmetry.} Promoting this symmetry to a \textit{connective} symmetry is subtler, due to the fact that connections require a \textit{choice} of adjustment, which is structure beyond the moduli stacks of bundles.

First, we show that principal ${\cal Z}_k$ bundles with \textit{fake-flat} connections can be tensored. This is expected because flat bundles are equivalent to bundles with flat connections, but it is instructive to see how this works out at the level of differential forms. We will work this out at the level of objects, but this can be shown to extend functorially as in Section~\ref{sec:2-action}. 
\begin{propo}[Tensor product of fake-flat ${\cal Z}_k$ connections] \label{propo:z-tensor-fake-flat}
Let 
\begin{eqnarray}
    \left(P,\nabla\right) &=& \left(x_{ij},\lambda_{ij},m_{ijk},z_{ijk},a_i,b_i,c_{ij}\right),
    \\
    \left(P',\nabla'\right) &=&\left(x'_{ij},\lambda'_{ij},m'_{ijk},z'_{ijk},a'_i,b'_i,c'_{ij}\right)
\end{eqnarray}
be a pair of fake-flat principal ${\cal Z}_k$ bundles with (adjusted) connection. Then the tuple $(X_{ij},\Lambda_{ij},M_{ijk},Z_{ijk}, A_i, B_i, C_{ij})$ where $(X_{ij},\Lambda_{ij},M_{ijk},Z_{ijk})$ are the cochains (\ref{eq:zk-composed-cochains1}), (\ref{eq:zk-composed-cochains2}) , and $(A_i,B_i,C_{ij})$ are differential forms defined as
\begin{equation}\label{eq:tensored-a-form}
    A_i : = a_i + a'_i,
\end{equation}
\begin{equation}\label{eq:tensored-b-form}
    B_{i} := b_i + b'_i + n a'_i \wedge a_i,
\end{equation}
\begin{equation}\label{eq:tensored-c-form}
    C_{ij}:= c_{ij} + c'_{ij} + \left(  nx'_{ij}a_j - n_{ij}a'_i - nx_{ij}a'_i \right)
\end{equation}
determine a principal ${\cal Z}_k$ bundle with (adjusted) fake-flat connection, denoted $(P\otimes P',\nabla\otimes \nabla')$. This describes a tensor product on principal ${\cal Z}_k$ bundles with fake-flat (adjusted) connections (at the level of objects).

The corresponding curvature $3$-form is
\begin{equation}
    H_i = h_i + h'_i .
\end{equation}

The forms (\ref{eq:tensored-a-form})-(\ref{eq:tensored-c-form}) make reference to integer-valued cochains $n_{ij}\in \Z$ defined in (\ref{eq:nlambda}) as
\begin{equation}
    \lambda_{ij}(x') = \exp\left(n_{ij}x \right),
\end{equation}
and $n\in \Z$ coming from the definition (\ref{eq:jform-n}) of the biadditive form $J$ as $J(x,x')= nxx'$.

\end{propo}
\begin{proof}
 It shown above, the tuple $(X_{ij},\Lambda_{ij},M_{ijk},Z_{ijk})$ defines a principal ${\cal Z}_k$ bundle. It remains to show the differential forms (\ref{eq:tensored-a-form})-(\ref{eq:tensored-c-form}) define a fake-flat adjusted connection on this bundle.

 It is immediate to see that the $A$-forms (\ref{eq:tensored-a-form}) satisfy the required identity (\ref{eq:a1form-identity}) with respect to the $X_{ij}$ cochains
 \begin{equation}
     A_j = A_i + d X_{ij}.
 \end{equation}

To deduce the new $C$-forms, the general ansatz is
\begin{equation}
    C_{ij} = c_{ij} + c'_{ij} +  (C-{\rm correction})_{ij},
\end{equation}
where the correction terms are necessarily defined on double intersections $U_{ij}$.

The $(M_{ijk},Z_{ijk})$ transition functions are already fixed, and give rise to the following terms in the sought identity (\ref{eq:cforms-identity})
\begin{gather}\nonumber
- Z_{ijk}\nabla_{A_i}Z^{-1}_{ijk}= -z_{ijk}\nabla_{a_i}z_{ijk}^{-1} - z'_{ijk}\nabla_{a'_i}z_{ijk}^{'-1}\\+    \left(n_{jk}dx'_{ij} + nx'_{ij}dx_{jk} + nx_{jk} dx'_{ij} + nm'_{ijk}dx_{ik} + nm'_{ijk}a_i + nm_{ijk}a_i'\right) \label{eq:additionaltensorc}
\end{gather}
for $\nabla_{a_i}:=d+a_i\triangleright$, and where the terms in brackets come from the failure of $(M_{ijk},Z_{ijk})$ matching the product of $(m_{ijk},z_{ijk})$ and $(m'_{ijk},z'_{ijk})$. The integers $n_{jk}$ arise from
\begin{equation}
    d\log \left(\lambda_{jk}\left(x'_{ij}\right)\right)= d\log \exp\left(n_{jk} x_{ij}'\right)= n_{jk}dx'_{ij},
\end{equation}
and we expanded the $\Z$-bilinear form $J(x,x')=nxx'$ as
\begin{equation}
    d\log\exp\left(J\left(x,x'\right)\right) = d\log\left(\exp\left(nxx'\right) \right) = nd(xx').
\end{equation}
Applying the identities (\ref{eq:twistedsuucocycle1}), (\ref{eq:a1form-identity}), and (\ref{eq:lambda-identity}) rewritten as
\begin{equation}\label{eq:n-cocycle}
    n_{ik} = -2nm_{ijk} + n_{ij} + n_{jk}
\end{equation}
since
\begin{equation}
    \tau(m_{ijk})(x) = \exp\left(-(J(m_{ijk},x')+J(x',m_{ijk}))\right) = \exp\left(-2nm_{ijk}x' \right),
\end{equation}
allows to rewrite the additional terms in brackets (\ref{eq:additionaltensorc}) as
\begin{equation}
 \left( nx'_{ik}a_k - n_{ik}a'_i - nx_{ik}a'_i \right) - \left(  nx'_{ij}a_j - n_{ij}a'_i - nx_{ij}a'_i \right) - \left( nx'_{jk}a_k - n_{jk}a'_j - nx_{jk}a'_j \right),
\end{equation}
which are the contributions from the correction terms of $C_{ik}-C_{ij}-C_{jk}$ provided these are defined as
\begin{equation}
    (C-{\rm correction})_{ij}:= \left(  nx'_{ij}a_j - n_{ij}a'_i - nx_{ij}a'_i \right).
\end{equation}

Moving on to the $B$-forms, a similar ansatz is
\begin{equation}\label{b-correction-ansatz}
    B_i = b_i + b'_i + (B-{\rm correction})_i
\end{equation}
where the corrections are defined on the individual patches $U_i$.

The terms dictating what the corrections have to be come from the identity (\ref{eq:bforms-identity})
\begin{gather}
    d C_{ij} - \kappa\left(-X_{ij},d A_i \right) = dc_{ij} + dc'_{ij} + \kappa\left(x_{ij},da_i\right) + \kappa\left(x'_{ij},da'_i\right) \nonumber
    \\
    + \left({\color{blue} \left(\left(na'_j\wedge a_j\right) - \left(na'_i\wedge a_i\right)\right)} + {\color{red} nx'_{ij}da_j - n_{ij}da'_i -nx_{ij}da'_i  + \kappa\left(x_{ij},da'_i \right) + \kappa\left(x'_{ij},da_i \right)} \right). \label{eq:b-modification-terms}
\end{gather}

The terms in (\ref{eq:b-modification-terms}) in blue are taken care of by introducing the correction
\begin{equation}
    (B-{\rm correction})_i:=  n a'_i\wedge a_i.
\end{equation}
The terms in red vanish individually by imposing the fake-flatness condition
\begin{equation}
    da_i = da'_i = 0,
\end{equation}
and thus establish that the forms (\ref{eq:tensored-a-form})-(\ref{eq:tensored-c-form}) define a fake-flat adjusted connection, denoted as $\nabla\otimes\nabla'$.

Finally, the curvature $3$-form is
\begin{equation}
    H_i = dB_i - \kappa\left(A_i,F_i \right) = d\left(b_i+b'_i \right) = h_i + h'_i.
\end{equation}
\end{proof}

Finally, let us describe how this result interacts with a connective refinement of the principal ${\cal Z}_{k}$ bundle action (\ref{eq:tensoredcochainsaction})-(\ref{eq:zk-action-cochains2}) on principal $\String{U(1)_k}$ bundles.

\begin{propo}[Subgroup action of ${\cal Z}_k$ on $\String{U(1)_k}$ bundles with connection]\label{propo:h-tensor-string-action}

Let
\begin{equation}
    {\cal H} = \left(U(1) \xrightarrow{0} \R \right),
\end{equation}
equipped with the trivial action of $\R$ on $U(1)$, be the smooth \textit{braided} sub-$2$-group of ${\cal Z}_k$ defined by the inclusions
\begin{equation}
\begin{tikzcd}
1 \arrow[r] & U(1) \arrow[r, hook] & U(1) \arrow[r, "0"] \arrow[d, "{\left(0,\text{id}_{U(1)}\right)}"', hook] & \R \arrow[r, two heads] \arrow[d, "{\left(\text{id}_{\R},0\right)}", hook] & \R \arrow[r]                     & 1 \\
1 \arrow[r] & U(1) \arrow[r, hook] & \Z \times U(1) \arrow[r, "\delta"']                                       & \R \oplus \Z \arrow[r, two heads]                                          & \frac{\R\oplus \Z}{\Z} \arrow[r] & 1
\end{tikzcd},
\end{equation}
which is equivalent to the direct product braided smooth $2$-group
\begin{equation}
    {\cal H} = \textbf{B}U(1) \times \R.
\end{equation}

Then principal $\String{U(1)_k}$ bundles with adjusted connection admit a tensor action by principal ${\cal H}$ bundles with fake-flat connection. 

 In Čech data, a principal $\cal H$ bundle with fake-flat connection given by cochains and differential forms $(x_{ij},z_{ijk},a_i,b_i,c_{ij})$ acts on a principal $\String{U(1)_k}$ bundle with adjusted connection $(x'_{ij},m'_{ijk},z'_{ijk},a'_i,b'_i, c'_{ij})$ as
\begin{equation}
 \left(x_{ij},z_{ijk},a_i,b_i,c_{ij}\right) \cdot \left(x'_{ij},m'_{ijk},z'_{ijk},a'_i,b'_i, c'_{ij}\right) = \left( X_{ij},M_{ijk},Z_{ijk}, A_i, B_i, C_{ij}\right)
\end{equation}
where
\begin{equation}
    X_{ij} = x_{ij} + x'_{ij},
\end{equation}
\begin{equation}
    M_{ijk} = m_{ijk},
\end{equation}
\begin{equation}
    Z_{ijk} = z_{ijk} z'_{ijk} \exp\left(J\left(x'_{ij},x_{jk} \right) + J\left(m'_{ijk},x_{ik} \right) \right),
\end{equation}
\begin{equation}
    A_i = a_i + a'_i,
\end{equation}
\begin{equation}\label{eq:h-action-string}
    B_i = b_i + b'_i + J\left(a'_i,a_i \right),
\end{equation}
\begin{equation}
    C_{ij} = c_{ij} + c'_{ij} + \left(J\left(x'_{ij},a_j \right)- J\left(x_{ij},a'_i \right) \right).
\end{equation}
\end{propo}
\begin{proof}
Principal ${\cal H}$ bundles are principal ${\cal Z}_k$ bundles whose $\lambda_{ij}$ cochains are all trivial. As a consequence of their cocycle identity written as (\ref{eq:n-cocycle}) it implies that the $m_{ijk}$ cochains are all zero, and thus the $x_{ij}$ satisfy the cocycle identity
\begin{equation}
    x_{ik} = x_{ij}+x_{jk},
\end{equation}
which describe the transition functions of a principal $\R$ bundle. The vanishing of $m_{ijk}$ furthermore imply that the action of $x_{ij}$ cochains on $z_{ijk}$ cochains is trivial, so these satisfy the cocycle identity
\begin{equation}
    z_{ikl} z_{ijk} = z_{ijl} z_{jkl},
\end{equation}
which describe a the transition functions of a bundle gerbe.

Hence, a principal $\cal H$ bundle with fake-flat connection is described by the cochains $x_{ij}$ of a principal $\R$ bundle, the cochains $z_{ijk}$ of a bundle gerbe, \textit{closed} differential forms $a_i$, and differential $1$- and $2$-forms $c_{ij}$, $b_i$ describing a connection on a bundle gerbe.

From above, we have the action at the level of bundles
\begin{equation}
  \left(x_{ij},0,0,z_{ijk}\right)\cdot \left(x'_{ij},m'_{ijk},z'_{ijk}\right) = \left(x_{ij}+x'_{ij},m'_{ijk}, z_{ijk}z'_{ijk}\exp\left(J\left(x'_{ij},x_{jk} \right) + J\left(m'_{ijk},x_{ik} \right) \right)\right),
\end{equation}
for $(x'_{ij},m'_{ijk},z'_{ijk})$ the bundle data of a principal $\String{U(1)_k}$ bundle with adjusted connection $(x'_{ij},m'_{ijk},z'_{ijk},a'_i,b'_i,c'_{ij})$.

As for the connective data, we follow the same process as in Proposition~\ref{propo:z-tensor-fake-flat}. We have the new $A$ forms (\ref{eq:tensored-a-form})
\begin{equation}
    A_i := a_i + a'_i,
\end{equation}
and the new $C$ forms (\ref{eq:tensored-c-form})
\begin{equation}
    C_{ij} := c_{ij}+c'_{ij} + \left(  nx'_{ij}a_j - nx_{ij}a'_i \right),
\end{equation}
taking into account that $n_{ij}=0$.

Now, the same analysis on the correction on $B$ forms gives (cf. Eq.~(\ref{eq:b-modification-terms}))
\begin{gather}
    d C_{ij} - \kappa\left(-X_{ij},d A_i \right) = dc_{ij} + dc'_{ij} + \kappa\left(x'_{ij},da'_i\right) \nonumber
    \\
    + \left({\color{blue} \left(\left(na'_j\wedge a_j\right) - \left(na'_i\wedge a_i\right)\right)} + {\color{red} \left(-nx_{ij}da'_i  + \kappa\left(x_{ij},da'_i \right)\right)} \right)\nonumber
    \\ 
    = dc_{ij} + dc'_{ij} + \kappa\left(x'_{ij},da'_i\right)
    + {\color{blue} \left(\left(na'_j\wedge a_j\right) - \left(na'_i\wedge a_i\right)\right)} 
\end{gather}
where we used the adjustment
\begin{equation}
    \kappa\left(x_{ij},da'_i \right) = nx_{ij}da'_i.
\end{equation}
Therefore, the resulting $B$ forms are
\begin{equation}
    B_i :=  b_i + b'_i + na'_i\wedge a_i = b_i + b'_i + J\left(a'_i,a_i\right) .
\end{equation}

The fake curvature is that of the adjusted connection on the principal $\String{U(1)_k}$ bundle
\begin{equation}
    F_i = dA_i = da'_i,
\end{equation}
and the curvature $3$-form is (cf. Eq.~(\ref{eq:2bundle-adj-3curv}))
\begin{gather}
    H_i = dB_i - \kappa(A_i,F_i) = db_i+db'_i -n a_i \wedge da'_i - n\left( a_i + a'_i\right)\wedge da'_i\nonumber,
    \\
   = db_i + db'_i - 2na_i\wedge da'_i - n a_i'\wedge da_i'\nonumber,
    \\
    = db_i + db'_i - 2J\left(a_i,da'_i\right) - J\left(a_i', da_i'\right)\nonumber,
    \\
     = h_i+ h'_i - 2J\left(a_i,da'_i\right) .
\end{gather}

\end{proof}

The action relevant for symmetries as per Section~\ref{sec:general}, namely that by principal bundles with \textit{flat} connections, is obtained from Proposition~\ref{propo:h-tensor-string-action} by further assuming that
\begin{equation}
    h_i = db_i = 0.
\end{equation}

\begin{remark}\label{rmk:mult-bundle-gerbes}
 Proposition~\ref{propo:h-tensor-string-action} only shows that flat principal bundles of the braided sub-$2$-group $\cal H$ of ${\cal Z}_k$ have a well-defined action on principal ${\rm String}\left(U(1)_k\right)$ bundles with adjusted connection. The same proof does not apply to principal bundles of the full ${\cal Z}_k$ $2$-group because, in the required identity (\ref{eq:bforms-identity}), one obtains the additional term $n_{ij}da_i'$ (cf. Equation~(\ref{eq:b-modification-terms})), which is not defined at a single patch $U_i$, and thus the ansatz (\ref{b-correction-ansatz}) does not work. Nevertheless, it is possible that there is a subtler way to define the action which does promote the $\textbf{B}{\cal Z}_k$ higher-form symmetry to the connective picture. The present situation of making a composition law compatible with connection data is reminiscent to the setting of \cite{waldorf2010multiplicative}. There, the focus are bundle gerbes with connection whose base space is a Lie group. The author observes that in order for the Lie group composition law to induce a composition of bundle gerbes with connection, it is necessary to require the composition to be defined not \textit{strictly} but only \textit{up to} trivial gerbes with (globally-defined) connection. We leave this interesting line of inquiry for future work.
\end{remark}

This action on connections directly specifies the action of the $\cal H$ center higher-form symmetry on holonomies, which as previously discussed are thought to become the charged operators in the fully quantum picture as Wilson lines/surfaces systems.

It is important to note that while $\cal H$ is a direct product group, the correspond actions are \textit{not} totally independent, as can be seen from Eq.~(\ref{eq:h-action-string}) where, even if the acting $\cal H$ has vanishing $(b_i, c_{ij})$ differential forms, the local $a_i$ nevertheless shift the local $b'_i$ forms of the $\suu$ bundle connection.

At this point, we can connect back to the context motivating the consideration of these fields, namely $5$d supergravity. From Proposition~\ref{propo:h-tensor-string-action} one sees that one has two infinitesimal transformations on the local gauge potentials $(a'_i,b'_i)$ that form part of the $\suu$ connection, parameterized by what we now know to be the potentials $(a_i,b_i)$ of a fake-flat and flat connection on a principal $\cal H$ bundle. These gauge variations are
\begin{equation}
    \delta_a a' =  a; \ \ \delta_b b' = b + J(a',a) = b -  n a\wedge a'.
\end{equation}
The transformations $\delta_a, \delta_b$ on their own right can be understood as the transformations induced by special cases of $\cal H$ bundles with fake-flat and flat connections, where the first one corresponds to principal $\R$ bundles with flat connection, and the later to bundle gerbes with flat connection.

Furthermore, it is straightforward to compute the commutators of these transformations, for $(a,b), (\tilde{a},\tilde{b})$ the gauge potentials of two principal $\cal H$ bundles with fake-flat and flat connections
\begin{eqnarray}
    [\delta_{a},\delta_{\tilde{b}}] = 0; & [\delta_b, \delta_{\tilde{b}}] = 0.
\end{eqnarray}
and a \textit{nontrivial} commutator
\begin{gather}
    a \cdot\left(\tilde{a}\cdot \left(a',b'\right)\right) = a \cdot \left(\tilde{a}+a',b'-n \tilde{a}\wedge a'\right) = \left(a+\tilde{a}+a',b'-n \tilde{a}\wedge a' - n a\wedge \left(\tilde{a}+a' \right)\right) ,
\end{gather}
\begin{gather}
    \tilde{a} \cdot\left(a\cdot \left(a',b'\right)\right) = \tilde{a} \cdot \left(a+a',b'-n a\wedge a'\right) = \left(\tilde{a}+a+a',b'-n a\wedge a' - n \tilde{a}\wedge \left(a+a' \right)\right)
\end{gather}
\begin{gather}
 a \cdot\left(\tilde{a}\cdot \left(a',b'\right)\right)  - \tilde{a} \cdot\left(a\cdot \left(a',b'\right)\right)  = \left(0,2n\tilde{a}\wedge a \right) = \left(0,2J\left(\tilde{a}, a\right) \right),
\end{gather}
which says that
\begin{eqnarray}
    [\delta_{a},\delta_{\tilde{a}}] = \delta_{b''}; & \ b'':= 2J\left( a,\tilde{a}\right) = 2n\, a\wedge \tilde{a}.
\end{eqnarray}

These are the five-dimensional analogs of what in eleven-dimensional supergravity has sometimes been referred to as the ``gauge algebra''\footnote{We thank Urs Schreiber for bringing this to our attention.} of such gauge fields \cite[Eq's 2.4-5]{Cremmer:1998px}, \cite{Lavrinenko:1999xi,Kalkkinen:2002tk,Bandos:2003et}, which as highlighted in \cite[Footnote 9]{Giotopoulos:2024xcg} are transformations more general than actual gauge transformations (whose parameters $a,b$ are required to be not just closed but exact differential forms) and therefore describe genuine global symmetries.

\paragraph{Gauging.}

Let us now discuss gauging. As in the $G=SU(n)$ case, the two natural cases to consider are the center $U(1)$ $2$-form \textit{sub}symmetry, and the center higher-form symmetry itself. However, as noted previously, the only part of the ${\cal Z}_k$ center higher-form symmetry that is \textit{strictly} compatible with nontrivial adjusted $\String{U(1)_k}$ connections is that described by principal ${\cal H}=\textbf{B}U(1)\times \R$ bundles with fake-flat and flat connections.

The former comes from the inclusion $\textbf{B}^2U(1) \hookrightarrow\textbf{B}{\cal Z}_k$ into the center higher-form symmetry. The flat bundles describing this symmetry are given, in the notation used above, by $(x_{ij},\lambda_{ij},m_{ijk},z_{ijk})$ with $x_{ij}=\lambda_{ij}=m_{ijk}=0$ and $z_{ijk}$ constant $U(1)$-valued functions on the triple intersections $U_{ijk}$ satisfying
\begin{equation}
    z_{ikl} z_{ijk} =z_{ijl} z_{jkl}.
\end{equation}
According to the action (\ref{eq:zk-action-cochains2}), this bundle acts on a principal $\String{U(1)_k}$ bundle $(x''_{ij},m''_{ijk},z''_{ijk})$ as
\begin{equation}
    z_{ijk}\cdot (x''_{ij},m''_{ijk},z''_{ijk}) = (x''_{ij},m''_{ijk},z_{ijk}z''_{ijk}),
\end{equation}
which is just the left-translation of bundle $2$-gerbes discussed in Section~\ref{sec:bfields}.

The result of gauging comes from the sequence that defines the string $2$-group itself as
\begin{equation}
    \textbf{B}^2 U(1) \to \textbf{B}\String{U(1)_k} \to \textbf{B}U(1) \xrightarrow{k\, \textbf{cs}}\textbf{B}^3U(1),
\end{equation}
that denotes $\textbf{B}U(1)$ as the weak quotient of $\textbf{B}\String{U(1)_k}$ by $\textbf{B}^2 U(1)$.

This is promoted to a connective symmetry according to the refined sequence in \cite{Fiorenza:2012tb}
\begin{equation}
    \textbf{B}^2U(1) \to \textbf{B}_{\nabla}\String{U(1)_k}\to \textbf{B}_{\nabla}U(1) \xrightarrow{k\, \textbf{cs}} \textbf{B}^3 U(1).
\end{equation}

If one wants to see explicitly how the bundle gerbe information gets ``absorbed'' by the gauging, it is necessary to construct the stacky quotient as in previous sections. This amounts to constructing \textit{principal smooth $3$-group bundles with adjusted connections}. However, the subject of adjusted connections for principal bundles of smooth $3$-groups is still a nascent one (see the recent work \cite{Gagliardo:2025oio}) and thus is out of the scope of the present paper.

Still, we can perform an analysis at the level of bundles/transition functions. In this case, the relevant smooth $3$-group $\cal G_0$ is a generalization of (\ref{eq:innu1xmod}), presented as a $2$-crossed complex of Lie groups \cite{baues1991combinatorial}
\begin{equation}
    {\cal G}_0 = \left(U(1) \xrightarrow{(0,\text{id}_{U(1)})} \Z\times U(1) \xrightarrow{\delta} \R\right),
\end{equation}
where $\delta:(m,z)\mapsto m$, actions
\begin{eqnarray}
    \alpha_1:& \R \times \left(\Z \times U(1)\right) &\to \Z \times U(1)
    \\ & (x,m,z) &\mapsto \left(m, z\exp\left( J\left(m,x\right)\right)\right),
\end{eqnarray}
\begin{eqnarray}
    \alpha_2:& \R \times U(1) &\to U(1),
    \\& (x,z) &\mapsto z,
\end{eqnarray}
and bracket
\begin{eqnarray}
    \{-,-\}:&\left(\Z\times U(1)\right)\times \left(\Z\times U(1)\right) &\to U(1),
    \\
    & \{\left(m_1,z_1\right),\left(m_2,z_2\right)\} & \mapsto 1.
\end{eqnarray}

This, however, is a specialization of the other gauging we consider immediately below, so we derive the Čech data for that more general case instead.

The other gauging to consider is that of the $\cal H$ center higher-form subsymmetry. We have the $2$-crossed module
\begin{equation}
    {\cal G} = \left(U(1) \xrightarrow{(0,0,\text{id}_{U(1)})} \R\times \left(\Z\times U(1)\right) \xrightarrow{\text{id}_{\R}\times \delta} \R\right), \label{eq:g-gauged-string-2xmod}
\end{equation}
$\delta:(m,z)\mapsto m$, actions
\begin{eqnarray}
    \alpha_1:& \R \times \left(\R\times \left( \Z \times U(1)\right)\right) &\to \R\times \left(\Z \times U(1)\right)
    \\ & (x,y,m,z) &\mapsto \left(y,m, z\exp\left( J\left(m,x\right)\right)\right),
\end{eqnarray}
\begin{eqnarray}
    \alpha_2:& \R \times U(1) &\to U(1),
    \\& (x,z) &\mapsto z,
\end{eqnarray}
and bracket
\begin{eqnarray}
    \{-,-\}:&\left(\R\times \left(\Z\times U(1)\right)\right)\times \left(\R\times \left(\Z\times U(1)\right)\right) &\to U(1), \label{eq:gauged-string-bracket}
    \\
    & \{\left(y_1,m_1,z_1\right),\left(y_2,m_2,z_2\right)\} & \mapsto \exp\left(-J\left(m_1,y_2\right) \right),
\end{eqnarray}
which is the unique bracket determined by the required identity (\ref{eq:2xmod-bracket-identity}), provided $\R\times (\Z\times U(1))$ has the direct product multiplication.

Notice, in particular, that the induced action (\ref{eq:3gp-haction}) of $\R \times (\Z\times U(1))$ on $U(1)$ is the trivial one.

A principal $\cal G$ bundle with respect to a good cover $\{U_i\}_{i\in {\cal I}}$ of $\Sigma$ consists of (see Appendix~\ref{app:principal-3-bundles}) cochains
\begin{eqnarray}
    x_{ij} &:& U_{ij}\to \R; \\ \left(y_{ijk},m_{ijk},z_{ijk}\right)&:& U_{ijk}\to \R\times\Z\times U(1); \\ h_{ijkl}&:& U_{ijkl} \to U(1)
\end{eqnarray}
satisfying the identities 
\begin{equation}
     x_{ik} = y_{ijk} + m_{ijk} + x_{ij} + x_{jk},
\end{equation}
\begin{equation}
    y_{ikl} + y_{ijk} = y_{ijl} + y_{jkl},
\end{equation}
\begin{equation}
  m_{ikl} + m_{ijk} = m_{ijl} + m_{jkl}  ,
\end{equation}
\begin{equation}
    z_{ikl}z_{ijk} h_{ijkl}  = z_{ijl} z_{jkl} \exp\left(J(m_{jkl},x_{ij}) \right),
\end{equation}
\begin{equation}
    h_{ijkl} h_{ijlm}h_{jklm} = h_{iklm} \exp\left(-J\left(m_{klm},y_{ijk} \right) \right) h_{ijkm}.
\end{equation}

We can therefore regard any principal $\suu$ bundle $P'$ as a principal $\cal G$ bundle $\imath (P')$ via the inclusion
\begin{equation}
    \left(x'_{ij},m'_{ijk},z'_{ijk} \right) \mapsto \left(x'_{ij},y'_{ijk}=0,m'_{ijk},z'_{ijk},h'_{ijkl}=1 \right),
\end{equation}
and likewise its image $\imath(P\cdot P')$ under the action of a $\cal H$ bundle $(x_{ij},\lambda_{ij}=0,m_{ijk}=0,z_{ijk})$ as
\begin{gather}
\left(x_{ij}+x'_{ij},m'_{ijk},z_{ijk}z'_{ijk}\exp\left(J\left( x'_{ij},x_{jk}\right)+J\left(m'_{ijk},x_{ik}\right) \right) \right)\\  \mapsto \left(x_{ij}+x'_{ij},y'_{ijk}=0,m'_{ijk},z_{ijk}z'_{ijk}\exp\left(J\left( x'_{ij},x_{jk}\right)+J\left(m'_{ijk},x_{ik}\right) \right),h'_{ijkl}=1 \right).   
\end{gather}
As is by now expected, $P'$ and $P\cdot P'$ are gauge equivalent in the gauged theory, that is, as $\cal G$ bundles via the gauge transformation
\begin{equation}
    \left(x_i,y_{ij},m_{ij},z_{ij},h_{ijk} \right): \imath\left(P'\right) \xrightarrow{\sim} \imath\left(P\cdot P'\right) \label{eq:gauged-string-gauge-transformation},
\end{equation}
\begin{eqnarray}
    x_i = 0; \ \ y_{ij} = x_{ij}; \ \ m_{ij} = 0; \ \ z_{ij} = 1, \label{eq:gauged-string-gauge-transformation-cochains1}
\end{eqnarray}
\begin{equation}
    h_{ijk} = z_{ijk}\exp\left(J\left( x'_{ij},x_{jk}\right)+J\left(m'_{ijk},x_{ik}\right)\right). \label{eq:gauged-string-gauge-transformation-cochains2}
\end{equation}

Most of the relevant identities are straightforward to verify, here we present only the preservation of the triviality of the $h_{ijkl}$ cochains as it is the most involved one which makes use of the new structure when compared to the crossed module of Lie groups scenario, namely the bracket (\ref{eq:gauged-string-bracket}).

The cochains (\ref{eq:gauged-string-gauge-transformation-cochains1})-(\ref{eq:gauged-string-gauge-transformation-cochains2}) define, according to Eq.~(\ref{eq:ell-gauge-transformation}), the target cochain $\widetilde{h}_{ijkl}$ as
\begin{gather}
    \widetilde{h}_{ijkl} = h_{ijk}^{-1} h_{ikl}^{-1} h_{ijl} \Big\{\left(x_{ij},0,1\right), x'_{ij}\triangleright_1\left(0,m'_{jkl},z'_{jkl} \right) \Big\}^{-1} h_{jkl} 
    \\
    = \left(h_{ijk}^{-1} h_{ikl}^{-1} h_{ijl}h_{jkl}\right)  \left(\exp\left(-J\left(m'_{jkl},x_{ij} \right) \right) \right)^{-1} \nonumber
    \\
    = \left(\exp\left(J\left(m'_{ijk},x_{ij}\right) \right)\right)^{-1}  \exp\left(J\left(m'_{jkl},x_{ij} \right) \right) = 1 \nonumber.
\end{gather}

This shows that, at least at the level of bundles, the fields of the $\textbf{B}\cal H$ gauged theory are
\begin{equation}
    \textbf{Fields}_{\textbf{B}{\cal H}} = \textbf{H}\left(\Sigma, \textbf{B}{\cal G} \right)
\end{equation}
for $\cal G$ the smooth $3$-group (\ref{eq:g-gauged-string-2xmod}). 

What does this physically mean? First, we look at the topological side. We can gain some insight by computing the topological realization of $\textbf{B}{\cal G}$. We get a classifying space with only three potentially nontrivial higher homotopy groups computed as
\begin{eqnarray}
    \pi_3 \left(\vert \textbf{B}{\cal G}\vert \right) &=& \ker\left(\left(0,0,\text{id}_{U(1)} \right)\right) = 0,
    \\
    \pi_2 \left(\vert \textbf{B}{\cal G}\vert \right) &=& \frac{\ker\left( \text{id}_{\R}\times \delta\right)}{\text{Im}\left(\left(0,0,\text{id}_{U(1)} \right)\right)} = \Z,
    \\
    \pi_1 \left(\vert \textbf{B}{\cal G}\vert \right) &=& \text{coker}\left( \text{id}_{\R}\times \delta\right) = 1.
\end{eqnarray}
Furthermore, notice that the higher center symmetry group we gauged is an extension, albeit trivial, of $\R$ by $\textbf{B}U(1)$ (cf. Proposition~\ref{propo:h-tensor-string-action}). This extended $\R$ is precisely the non-faithfully acting $\R$ $1$-form symmetry of $U(1)$ gauge theory we encountered in Section~\ref{ssec:u1}. Therefore, we can interpret the gauging of $\cal H$ as the \textit{simultaneous gauging of the center higher-form symmetry} $\textbf{B}^2U(1)$ of the bundle gerbe part of $\suu$ bundles \textit{and of the non-faithfully acting $\R$ center $1$-form symmetry} of the $U(1)$ part, properly lifted to a symmetry of $\suu$ bundles.

\section{Future directions}\label{sec:conclusion}

In this paper, we have shown how higher smooth geometry, via automorphism stacks, naturally encodes a wide variety of higher-form symmetries, namely outer automorphism and center symmetires, without the need of precise knowledge of the fully quantum theory and its generalized operators. We have rigorously constructed these symmetries for pure gauge theory of strict smooth $2$-groups. We then applied this to a variety of scenarios of physical interest. We have rederived known results about pure gauge theory. Then, we analyzed the case of B fields, realized as bundle gerbes with connection. Finally, we constructed novel examples of higher-form symmetries, in particular center higher-form symmetries, for pure higher gauge theory of specific string $2$-groups.

There are many directions of further research. First, let us mention those related to the general formulation of Section~\ref{sec:general}. As mentioned there, especially in gauge theory, often the \textit{choice} of target as the moduli space of a specific (higher) group already encodes fields satisfying some physically-relevant constraints such as certain equations of motion. Nevertheless, especially geometric constraints, like those involving a metric via the Hodge star operator, cannot be encoded in such a way on the whole of spacetime\footnote{This is more subtle when working with subspaces, such as Cauchy surfaces \cite{Sati:2023mta}.}. Therefore, the higher-form symmetries derived via the higher automorphism group need to be furthermore compatible with such constraints. Since the Lagrangian principle is one of the most widespread source of physical constraints, it would be interesting to investigate further in explicit examples the interplay between these higher-form symmetries and Lagrangians. The relation between diffeomorphism and Lagrangians as higher bundle gerbes is explored in \cite{nuiten2013cohomological}, there known under the name of higher \textit{quantomorphisms}. Further, recently the authors of \cite{Giotopoulos:2023pvl, Giotopoulos:2025obs} have begun to formalize Lagrangian field theory in terms of modern higher smooth geometry. We expect this to be relevant for this research direction. In particular, this work should also allow to investigate the action of higher smooth \textit{super}groups and therefore of \textit{fermionic symmetries}.

Another line of inquiry related to the general formulation is the issue encountered in Section~\ref{sec:string2groups} (cf. Remark~\ref{rmk:mult-bundle-gerbes}) of lifting bundle symmetries to symmetries of bundles with adjusted connection. Having constructed in Section~\ref{sec:2-action}, at the general level of $2$-groupoids, a composition law of ${\rm Aut}\left( \textbf{B}{\cal G}\right)$-valued maps, and an action of these on principal $\cal G$ bundles for any strict Lie $2$-group $\cal G$, it is natural to consider the proper extension of these to bundles with adjusted connections, perhaps in the spirit of \cite{waldorf2010multiplicative}.

There are also concrete examples to explore. One of such is related to the center higher-form symmetry of $\suu$ presented in Section~\ref{sec:string2groups}. There, the $5$d analog of the ``gauge algebra'' of $11$d supergravity gauge potential was obtained. The direct generalization to $11$d of our analysis would suggest to choose as a target stack the moduli stack of principal $\mathscr{G}$ bundles, where $\mathscr{G}$ is a higher smooth central extension
\begin{equation}
    \textbf{ B}^6U(1)\to \textbf{B}{\mathscr{G}}\to \textbf{B}^3U(1) \xrightarrow{\textbf{cs}_7} \textbf{B}^7U(1),
    \end{equation}
whose extension class is the $7$d Chern-Simons form \cite{Fiorenza:2012tb}
\begin{equation}
    \textbf{cs}_7:\textbf{B}^3U(1) \to \textbf{B}^7U(1).
\end{equation}
However, there already exists a different proposal, known as \textit{Hypothesis H} \cite{Sati:2024klw}, for the quantization law/global gauge group quantizing the $11$d fluxes, the $4$-sphere $S^4$. It would be interesting to apply our analysis to this proposal and obtain the gauge algebra as coming from the automorphism higher-form symmetries.

Another example to consider is that of target stacks that are not moduli stacks, namely smooth higher group\textit{oids}. These are known to play a role in describing broken symmetries, Higgs fields, and curved Yang-Mills-Higgs theories \cite{Kotov:2015nuz,Fischer:2021glc,groupoid-conn}.

Lastly, on a wider scope, there is also the question about the relation between \textit{any} fully quantum higher-form symmetry, and the pre-quantum higher automorphisms described here. In \cite{Heidenreich:2020pkc}, the authors argue that characteristic classes of principal bundles describe symmetries in the quantum theory, with the Chern-Weil representatives providing the Noether currents. Now, as a consequence of interpreting gaugings as stacky quotients, all the target spaces for the gauged theories presented in this paper come equipped with a distinguished characteristic class, as can be seen from the sequence in Eq.~(\ref{eq:gauging-fibration}). This indicates that Chern-Weil symmetries may be better understood as quantum/dual symmetries of the higher-form symmetries described here. The subtlety is that, at the level of higher smooth geometry, a space equipped with a characteristic/extension class by itself does not describe a group action on the space, since by definition such a space would be the base space of a principal bundle, which is invariant under the group action. An actual action should be visible only at the quantum picture, specifically when \textit{stabilizing} or \textit{linearizing} (e.g. in the sense of \cite{Schreiber:2013pra,Schreiber:2014xva}). We hope to address this in future work.

\section{Acknowledgments}

We would like to thank Hisham Sati for suggesting the project, and Urs Schreiber for fruitful discussions especially about the content of Section~\ref{sec:2-action}. We would like to thank them and Eric Sharpe for their comments on the draft.

\appendix

\section{Adjusted principal $2$-bundles with connection}\label{app:adjbundles}

In this appendix, we summarize the specification of adjusted principal bundles of Lie $2$-groups with connection.

\subsection{Lie $2$-groups and algebras}\label{sapp:lie2gps}
The presentation that we will mostly use in this paper for smooth (or Lie) $2$-groups is that of crossed modules \cite{mac1998categories}, where the groups and homomorphisms involved are smooth.

A Lie $2$-group ${\cal G} = (G,H,\delta,\triangleright)$ consists of $G,H$ a pair of Lie groups, a Lie group homomorphism
\begin{equation}
    \delta: H \to G,
\end{equation}
 called the boundary map, and
\begin{equation}\label{eq:xmodact}
    \triangleright: G\times H \to H,
\end{equation}
 a smooth action of $G$ on $H$, such that the following two identities hold
\begin{equation}\label{eq:xmod1}
    \triangleright\circ \left( \delta\times\text{id}_H\right) = \text{Ad}_H: H\times H \to H,
\end{equation}
\begin{equation}\label{eq:xmod2}
    \delta\circ\triangleright = \text{Ad}_G \circ \left(\text{Id}_G\times \delta \right),
\end{equation}
for ${\rm Ad}_H, {\rm Ad}_G$ the adjoint actions of $H$ and $G$ on themselves.

Since its defining data is smooth, a crossed module of Lie groups gives, upon differentiation, a crossed module of Lie algebras $\mathscr{g}=({\frak g}, {\frak h}, \delta, \triangleright)$ where ${\frak g} = {\rm Lie}(G)$, ${\frak h} = {\rm Lie}(H)$, and by abuse of notation we denote the Lie algebra homomorphisms
\begin{equation}
    \delta: {\frak h} \to {\frak g},
\end{equation}
\begin{equation}
    \triangleright: {\frak g \otimes \frak h \to \frak h},
\end{equation}
where the latter is an action by derivations of $\frak g$ on $\frak h$, satisfying the (differential version of the) identities (\ref{eq:xmod1})-(\ref{eq:xmod2}).

\subsection{Adjustments}

A principal bundle of a smooth $2$-group ${\cal G}=(G,H,\delta, \triangleright)$ can be schematically (though not literally) thought of as a system of mutually twisting $BH$ and $G$ bundles. When it comes to connections as originally defined (e.g. \cite{Schreiber:2008kcv}) in these bundles, a condition known as \textit{fake-flatness} is imposed for consistency, where the curvature of the would-be principal $G$ bundle is required to match the connection form of the would-be $BH$ bundle. However, this condition is known \cite{Gastel:2018joi, samann2020towards} to locally reduce the $\cal G$ bundle with connection to a $BK$ bundle with connection, where $BK$ is the delooping of an abelian Lie group $K$, which for many purposes is too restrictive. 

To ameliorate this issue, the notion of \textit{adjustment} \cite{samann2020towards,Rist:2022hci} must be introduced.

An adjusted crossed module of Lie groups $({\cal G}, \kappa)$ is a crossed module of Lie groups ${\cal G} = (G,H,\delta,\triangleright)$ equipped with a smooth map
\begin{equation}
    \kappa: G \times {\frak g} \to {\frak h},
\end{equation}
linear in $\frak g$, satisfying the following conditions
\begin{equation}
 \kappa\left( \delta(h), V\right) = h\left(V \triangleright h^{-1} \right)   ,
\end{equation}
\begin{equation}
    \kappa\left(g_2g_1,V \right) = g_2\triangleright \kappa\left( g_1,V\right) + \kappa\left(g_2, g_1 V g_1^{-1}- \delta\left( \kappa\left(g_1,V\right)\right) \right),
\end{equation}
for $g_1,g_2 \in G$, $V\in {\frak g}$, $h\in H$, and
\begin{equation}
    \triangleright: G \times {\frak h} \to {\frak h}
\end{equation}
the infinitesimal action of $G$ on $\frak h$ obtained from the action (\ref{eq:xmodact}).

Upon differentiation, the adjustment $\kappa$ of ${\cal G} = (G,H, \delta,\triangleright)$ induces an adjustment on the corresponding crossed module of Lie algebras $\mathscr{g}=({\frak g}, {\frak h}, \delta, \triangleright)$, a map
\begin{equation}
    \kappa: {\frak g}\times {\frak g} \to {\frak h}
\end{equation}
which satisfies
\begin{equation}
    \kappa\left(\delta(W),V \right) = -V\triangleright W,
\end{equation}
\begin{gather}
 \nonumber   \kappa\left([V_1,V_2],V_3 \right) = V_1\triangleright\kappa\left(V_2,V_3 \right)  - V_2 \triangleright \kappa\left(V_1,V_3\right) + \kappa\left(V_1,[V_2,V_3] \right) - \kappa\left(V_2, [V_1,V_3] \right)\\ - \kappa\left(V_1, \delta\left(\kappa\left(V_2,V_3\right)\right) \right) + \kappa\left(V_2, \delta\left( \kappa\left(V_1,V_3\right)\right) \right) ,
\end{gather}
for $W,V_1,V_2,V_3\in {\frak g}$.

\subsection{Čech data of principal $2$-bundles with connection}

Given a crossed module of Lie groups with a \textit{choice} of adjustment, one can describe adjusted principal bundles with adjusted connections, gauge transformations, and gauge-of-gauge transformations. The contents of this section are taken from \cite[Section 2]{Rist:2022hci}.

For $\Sigma$ a smooth manifold with $\{U_i\}_{i \in {\cal I}}$ a good cover, and $({\cal G}, \kappa)$ an adjusted crossed module of Lie groups, the Čech data associated with a principal ${\cal G}=(G,H,\delta,\triangleright)$ bundle with adjusted connection over $\Sigma$ consists of 
\begin{enumerate}
    \item $G$-valued cochains $g_{ij}: U_{ij} \to G$,
    \item $H$-valued cochains $h_{ijk}: U_{ijk} \to  H$,
    \item $\frak g$-valued differential $1$-forms $a_i$,
    \item $\frak h$-valued differential $1$- and $2$- forms $\Lambda_{ij}, b_i$, respectively
\end{enumerate}
satisfying the following list of identities
\begin{equation}
    g_{ik} = \delta\left( h_{ijk}\right) g_{ij}g_{jk},
\end{equation}
\begin{equation}
    h_{ikl}h_{ijk} = h_{ijl}\,\left(g_{ij}\triangleright h_{jkl}\right),
\end{equation}
\begin{equation}
    \Lambda_{ik} = \Lambda_{jk} + g_{jk}^{-1}\triangleright \Lambda_{ij} - g_{ik}^{-1} \triangleright \left(h_{ijk} \left(d+A_i\triangleright \, \right)h^{-1}_{ijk} \right),
\end{equation}
\begin{equation}
    A_j = g_{ij}^{-1} A_i g_{ij} + g_{ij}^{-1} dg_{ij} - \delta\left(\Lambda_{ij}\right),
\end{equation}
\begin{equation}
    B_j = g_{ij}^{-1} \triangleright B_i + d\Lambda_{ij} + A_j\triangleright \Lambda_{ij} + \tfrac{1}{2}[\Lambda_{ij},\Lambda_{ij}] - \kappa\left(g_{ij}^{-1}, F_i \right),
\end{equation}
for $F_i$ the \textit{adjusted fake curvature}, a $\frak g$-valued differential $2$-form
\begin{equation}
    F_i:= d A_i + \tfrac{1}{2}[A_i,A_i] + \delta\left( B_i \right).
\end{equation}

The \textit{adjusted curvature} of this bundle is a $\frak h$-valued differential $3$-form $H_i$ defined as
\begin{equation}\label{eq:2bundle-adj-3curv}
    H_i := dB_i + A_i\triangleright B_i - \kappa\left(A_i,F_i \right).
\end{equation}

Given a pair $(g_{ij},h_{ijk}, A_i, B_i, \Lambda_{ij})$, $(g'_{ij},h'_{ijk}, A'_i, B'_i, \Lambda'_{ij})$ of adjusted principal $\cal G$ bundles with connection, an \textit{adjusted gauge transformation }
\begin{equation}
    (b,a,\lambda): (g_{ij},h_{ijk}, A_i, B_i, \Lambda_{ij}) \xrightarrow{\sim} (g'_{ij},h'_{ijk}, A'_i, B'_i, \Lambda'_{ij})
\end{equation}
consists of
\begin{enumerate}
    \item $H$-valued cochains $b_{ij}:U_{ij}\to H$,
    \item $G$-valued cochains $a_i: U_i\to G$, 
    \item $\frak h$-valued differential $1$-forms $\lambda_i$,
\end{enumerate}
satisfying the identities
\begin{equation}
    h'_{ijk} = a_i^{-1} \triangleright \left( b_{ik} h_{ijk} \left(g_{ij} \triangleright b_{jk}^{-1} \right) b_{ij}^{-1}\right),
\end{equation}
\begin{equation}
    g'_{ij} = a_i^{-1} \delta\left( b_{ij} \right) g_{ij} a_j,
\end{equation}
\begin{equation}
    \Lambda'_{ij} = a_j^{-1} \triangleright \Lambda_{ij} + \lambda_j - \left(g'_{ij}\right)^{-1}\triangleright\lambda_i + \left(a_j^{-1} g_{ij}^{-1} \right)\triangleright\left( b_{ij}^{-1} \left(d+ A_i \, \triangleright \right) ,b_{ij}\right)
\end{equation}
\begin{equation}
    A'_i = a_i^{-1} A_i a_i + a_i^{-1} d a_i - \delta\left(\lambda_i \right),
\end{equation}
\begin{equation}
    B'_i = a_i^{-1} \triangleright B_i + d\lambda_i + A'_i \triangleright\lambda_i + \tfrac{1}{2}[\lambda_i,\lambda_i] - \kappa\left(a_i^{-1},F_i \right).
\end{equation}

Lastly, we have adjusted equivalences, or \textit{adjusted gauge-of-gauge transformations}, for a pair $(a,b,\lambda)$, $(a',b',\lambda')$
\begin{equation}
    m : (a,b,\lambda) \xrightarrow{\sim} (a,b,\lambda),
\end{equation}
consisting of cochains $m_i: U_i \to H$ satisfying the identities
\begin{equation}
    b'_{ij} = m_{i} b_{ij} \left( g_{ij} \triangleright m_j^{-1}\right),
\end{equation}
\begin{equation}
    a'_i = \delta\left( m_i \right) a_i,
\end{equation}
\begin{equation}
    \lambda'_i = \lambda_i + a_i^{-1}\triangleright\left( m_{i} \left( d+ A_i\,\triangleright\right)m_i\right).
\end{equation}

\section{Background to strict $2$-groupoids}\label{app:2grpds}

In this Appendix, we compile relevant definitions and results from the theory of strict $2$-groupoids.

\begin{definition}[{cf. \cite[\S4.1]{JohnsonYau2021}}]\label{LaxFunctor}

Given ${\cal X,Y}\in \mathrm{Grpd}_{{}_2}$, a \textit{lax functor} consists of a function on the objects $F: {\rm ob}({\cal X}) \to {\rm ob}({\cal Y})$, and local functors
\begin{equation}\label{LaxFunctor-local}
    F: {\cal X}(x,y) \to {\cal Y}(Fx,Fy)
\end{equation}
with component $2$-cells of a natural transformation called the lax functoriality constraint
\begin{equation}\label{LaxFunctor-component2cell}
    F(f,g):Fg\circ Ff \xrightarrow{} F(g\circ f)
\end{equation}
\begin{equation} 
F_e(x):{\rm id}_{F(x)} \xrightarrow{} F({\rm id}_x)
\end{equation}
satisfying associativity
\begin{equation}\label{LaxFunctor-assoc}
    \begin{tikzcd}
\left(Fh \circ Fg\right) \circ Ff \arrow[dd, "{F(h,g)}"'] \arrow[rr, "{{\rm id}_{Fh}\star F(g,f)}"] &  & Fh\circ F(gf) \arrow[dd, "{F(h,gf)}"] \\
                                                                                                    &  &                                       \\
F(hg)\circ Ff \arrow[rr, "{F(hg,f)}"']                                                              &  & F\left((hg)f \right)                 
\end{tikzcd},
\end{equation}
for $\star$ the horizontal composition of $2$-morphisms, and left, right unity
\begin{equation}
    \begin{tikzcd}
{\rm id}_{F(x)}\circ Ff \arrow[dd, "F_e(x)\star {\rm id}_{Ff}"'] \arrow[rr, "\ell_{\cal Y}"] &  & Ff                                                &  & Ff \circ {\rm id}_{Fx} \arrow[rr, "r_{\cal Y}"]                                              &  & Ff                                                         \\
                                                                                             &  &                                                   &  &                                                                                              &  &                                                            \\
F({\rm id}_x)\circ Ff \arrow[rr, "{F({\rm id}_x,f)}"']                                       &  & F({\rm id}_x \circ f) \arrow[uu, "\ell_{\cal X}"] &  & Ff\circ F{\rm id}_x \arrow[rr, "{F(f,{\rm id}_x)}"'] \arrow[uu, "{\rm id}_{Ff}\star F_e(x)"] &  & F\left(f\circ{\rm id}_x \right) \arrow[uu, "Fr_{\cal X}"']
\end{tikzcd},
\end{equation}
where for strict $2$-groupoids the left and right unitors $\ell, r$ are trivial.

Further, a \textit{strict $2$-functor} is a lax functor for which the natural transformation $F(g,h)$ is the identity natural isomorphism.
\end{definition}

\begin{definition}[{cf. \cite[\S4.2]{JohnsonYau2021}}]
\label{Transformations}
Given $\mathcal{X}, \mathcal{Y} \in \mathrm{Grpd}_{{}_2}$ and $F,G : \mathcal{X} \longrightarrow \mathcal{Y}$ a pair of parallel 2-functors, a \emph{transformation} between them 
\begin{equation}
  \begin{tikzcd}
    \mathcal{X}
    \ar[
      rr,
      bend left=40,
      "{ F }"{description, name=s}
    ]
    \ar[
      rr,
      bend right=40,
      "{ G }"{description, name=t}
    ]
    \ar[
      from=s, to=t,
      Rightarrow, 
      "{ \eta }"
    ]
    &&
    \mathcal{Y}
  \end{tikzcd}
\end{equation}
is a natural transformation
\begin{equation}
    \eta: \eta(x_1)^*G \to \left(\eta(x_2)\right)_*F: {\cal X}\left(x_1,x_2\right) \to {\cal Y}\left( Fx_1, Gx_2\right),
\end{equation}
of the form
\begin{equation}
  \begin{tikzcd}
    x_1
    \ar[
      dd,
      "{ f }"{description}
    ]
    &
    F(x_1)
    \ar[
      dd,
      "F(f)"{description}
    ]
    \ar[
      rr,
      "{ \eta(x_1) }"{description}
    ]
    &&
    G(x_1)
    \ar[
      dd,
      "{ G(f) }"{description}
    ]
    \ar[
      ddll,
      shorten=15pt,
      Rightarrow,
      "{ \eta(f) }"{description}
    ]
    \\
    {}
    \ar[
      r,
      |->,
      shorten=15pt
    ]
    & {}
    \\
    x_2
    &
    F(x_2)
    \ar[
      rr,
      "{ \eta(x_2) }"{description}
    ]
    &&
    G(x_2)
  \end{tikzcd}  
\end{equation}
such that the following equations hold among 2-morphisms in $\mathcal{Y}$:
\begin{itemize}
\item[(i)] {\bf (naturality)}
\begin{equation}
  \begin{tikzcd}
    x_1
    \ar[
      dd,
      bend left=50,
      "{ f }"{description, pos=.4, name=s}
    ]
    \ar[
      dd,
      bend right=50,
      "{ g }"{description, pos=.6, name=t}
    ]
    \ar[
      from=s, to=t,
      Rightarrow,
      "{ \alpha }"{description}
    ]
    \\
    \\
    x_2
  \end{tikzcd}
  \mapsto
  \left(
  \begin{aligned}
  &
  \begin{tikzcd}
    F(x_1)
    \ar[
      dd,
      bend left=50,
      "{ F(f) }"{description, pos=.4, name=s}
    ]
    \ar[
      dd,
      bend right=50,
      "{ F(g) }"{description, pos=.6, name=t}
    ]
    \ar[
      from=s, to=t,
      Rightarrow,
      "{ F(\alpha) }"{pos=1}
    ]
    \ar[
      rr,
      "{ \eta(x_1) }"{description}
    ]
    &&
    G(x_1)
    \ar[
      dd,
      bend left=50,
      "{ G(f) }"{description}
    ]
    \ar[
      ddll,
      Rightarrow,
      shorten=15pt,
      shift right=5pt,
      bend left=20,
      "{ \eta(f) }"{description}
    ]
    \\
    \\
    F(x_2)
    \ar[
      rr,
      "{ \eta(x_2) }"{description}
    ]
    &&
    G(x_2)
  \end{tikzcd}
  \\
  =
  \\
  &
  \begin{tikzcd}
    F(x_1)
    \ar[
      dd,
      bend right=50,
      "{ F(g) }"{description, pos=.6, name=t}
    ]
    \ar[
      rr,
      "{ \eta(x_1) }"{description}
    ]
    &&
    G(x_1)
    \ar[
      dd,
      bend left=50,
      "{ G(f) }"{description, pos=.4, name=s}
    ]
    \ar[
      dd,
      bend right=50,
      "{ G(g) }"{description, pos=.6, name=t}
    ]
    \ar[
      from=s, to=t,
      Rightarrow,
      "{ G(\alpha) }"{pos=1}
    ]
    \ar[
      ddll,
      Rightarrow,
      shorten=15pt,
      shift left=5pt,
      bend right=20,
      "{ \eta(g) }"{description}
    ]
    \\
    \\
    F(x_2)
    \ar[
      rr,
      "{ \eta(x_2) }"{description}
    ]
    &&
    G(x_2)
  \end{tikzcd}
  \end{aligned}
  \right)
\end{equation}
\item[(ii)] {\bf (unitality)}
\begin{equation}
  x
  \;\mapsto\;
  \left(
  \begin{aligned}
  &
  \begin{tikzcd}
    F(x)
    \ar[
      dd,
      bend left=50,
      "{ \mathrm{id}_{F(x)} }"{description, pos=.4, name=s}
    ]
    \ar[
      dd,
      bend right=50,
      "{ F(\mathrm{id}_x) }"{description, pos=.6, name=t}
    ]
    \ar[
      from=s, to=t,
      Rightarrow,
      "{ F_{\!\mathrm{e}}\!(x) }"{pos=1.2}
    ]
    \ar[
      rr,
      "{ \eta(x) }"{description}
    ]
    \ar[
      ddrr,
      bend left=20,
      "{ \eta(x) }"{description, name=mid}
    ]
    &&
    G(x)
    \ar[
      dd,
      bend left=50,
      "{ \mathrm{id}_{G(x)} }"{description}
    ]
    \ar[
      from=mid,
      to=3-1,
      shorten=10pt,
      Rightarrow,
      bend left=20,
      shift right=3pt,
      "{
        \ell^{-1}
      }"
    ]
    \ar[
      from=1-3,
      to=mid,
      shorten=0pt,
      Rightarrow,
      bend left=20,
      shift right=-3pt,
      "{
        r
      }"
    ]
    \\
    & 
    \\
    F(x)
    \ar[
      rr,
      "{ \eta(x) }"{description}
    ]
    &&
    G(x)
  \end{tikzcd}
  \\
  =
  \\
  &
  \begin{tikzcd}
    F(x)
    \ar[
      dd,
      bend right=50,
      "{ F(\mathrm{id}_x) }"{description, pos=.6, name=t}
    ]
    \ar[
      rr,
      "{ \eta(x_1) }"{description}
    ]
    &&
    G(x)
    \ar[
      dd,
      bend left=50,
      "{ \mathrm{id}_{G(x)} }"{description, pos=.4, name=s}
    ]
    \ar[
      dd,
      bend right=50,
      "{ G(\mathrm{id}_x) }"{description, pos=.6, name=t}
    ]
    \ar[
      from=s, to=t,
      Rightarrow,
      "{ G_{\!\mathrm{e}}\!(x) }"{pos=1.3}
    ]
    \ar[
      ddll,
      Rightarrow,
      shorten=15pt,
      shift left=5pt,
      bend right=20,
      "{ \eta(\mathrm{id}_x) }"{description}
    ]
    \\
    \\
    F(x)
    \ar[
      rr,
      "{ \eta(x_2) }"{description}
    ]
    &&
    G(x)
  \end{tikzcd}
  \end{aligned}
  \right)
\end{equation}
\item[(iii)] {\bf (associativity)}
\begin{equation}\label{Transformation-Associativity}
  \hspace{-.6cm}
  \begin{tikzcd}[
    column sep=15pt
  ]
    x_1
    \ar[
      dr,
      "{ f }"{description}
    ]
    \ar[
      dd,
      "{ g \circ f }"{description}
    ]
    \\
    & 
    x_2
    \ar[
      dl,
      "{ g }"{description}
    ]
    \\
    x_3
  \end{tikzcd}
  \mapsto
  \left(
  \begin{aligned}
  &
  \begin{tikzcd}[
    column sep=30pt
  ]
    F(x_1)
    \ar[
      dr,
      "{ F(f) }"{description}
    ]
    \ar[
      dd,
      "{ F(g \circ f) }"{sloped, swap}
    ]
    \ar[
      rr,
      "{ \eta(x_1) }"{description}
    ]
    &&[-23pt]
    G(x_1)
    \ar[
      dl,
      Rightarrow,
      shorten=8,
      "{ \eta(f) }"
    ]
    \ar[
      dr,
      "{ G(f) }"{description}
    ]
    \\
    {}
    & 
    F(x_2)
    \ar[
      dl,
      "{ F(g) }"{description}
    ]
    \ar[
      l,
      Rightarrow,
      "{ F(f,g) }"{description, pos=.45}
    ]
    \ar[
      rr,
      "{ \eta(x_2) }"{description}
    ]
    &&
    G(x_2)
    \ar[
      dlll,
      Rightarrow,
      shorten=65,
      "{ \eta(g) }"{pos=.52}
    ]
    \ar[
      dl,
      "{ G(g) }"{description}
    ]
    \\
    F(x_3)
    \ar[
      rr,
      "{ \eta(x_3) }"{description}
    ]
    &&
    G(x_3)
  \end{tikzcd}
  \\
  & =
  \\
  &
  \begin{tikzcd}[
    column sep=30pt
  ]
    F(x_1)
    \ar[
      dd,
      "{ F(g \circ f) }"{sloped, swap}
    ]
    \ar[
      rr,
      "{ \eta(x_1) }"{description}
    ]
    &[+10pt]&
    G(x_1)
    \ar[
      ddll,
      Rightarrow,
      shorten=34pt,
      "{ \eta(g \circ f) }"
    ]
    \ar[
      dr,
      "{ G(f) }"{description}
    ]
    \ar[
      dd,
      "{ G(g \circ f) }"{description,sloped}
    ]
    \\
    & 
    &
    {}
    &
    G(x_2)
    \ar[
      dl,
      "{ G(g) }"{description}
    ]
    \ar[
      l,
      Rightarrow,
      "{ G(f,g) }"{description}
    ]
    \\
    F(x_3)
    \ar[
      rr,
      "{ \eta(x_3) }"{description}
    ]
    &&
    G(x_3)
  \end{tikzcd}
  \end{aligned}
  \right)
\end{equation}
\end{itemize}
\end{definition}
\begin{definition}[{cf. \cite[\S12.2]{JohnsonYau2021}}]\label{GrayTensorProduct}
    Given $\mathcal{X}, \mathcal{Y} \in \mathrm{Grpd}_{{}_2}$, their \textit{Gray tensor product} ${\cal X}\otimes_G {\cal Y}$ is defined by
    \begin{itemize}
        \item objects ${\rm ob}({\cal X}\otimes_{G} {\cal Y}) = {\rm ob}({\cal X})\times {\rm ob}({\cal Y})$,
        \item $1$-morphisms generated by the composition of
        \begin{eqnarray}
            f\otimes_ y:& x\otimes y \to x'\otimes y
            \\
            x \otimes g:& x\otimes y \to x \otimes y',
        \end{eqnarray}
        for $f:x\to x'$, $g:y \to y'$, satisfying
        \begin{equation}
            {\rm id}_x \otimes y = {\rm id}_{x\otimes y} = x\otimes {\rm id}_y,
        \end{equation}
        \begin{equation}
\begin{tikzcd}
x\otimes y \arrow[dd, "f \otimes y"'] \arrow[rrdd, "(f'f )\otimes y"] &  &              &  & x\otimes y \arrow[dd, "x\otimes g"'] \arrow[rrdd, "x\otimes (g'g)"] &  &              \\
                                                                      &  &              &  &                                                                     &  &              \\
x'\otimes y \arrow[rr, "f'\otimes y"']                                &  & x''\otimes y &  & x\otimes y' \arrow[rr, "x\otimes g'"']                              &  & x\otimes y''
\end{tikzcd}
        \end{equation}
\item $2$-cells are generated through vertical and horizontal composition of equivalence classes of generating $2$-cells
\begin{eqnarray}
    \alpha \otimes y &:& f_1 \otimes y \to f_2 \otimes y,
    \\
    x \otimes \beta &:& x \otimes g_1 \to x\otimes g_2,
    \\
    \Sigma_{f,g} &:& (f\otimes y')(x\otimes g) \to (x' \otimes g)(f \otimes y), \label{proto-2-cells1}
    \\
    \Sigma_{f,g}^{-1} &:& (x' \otimes g)(f \otimes y) \to (f\otimes y')(x\otimes g),  \label{proto-2-cells2}
\end{eqnarray}
satisfying the conditions
\begin{equation}
    {\rm id}_f \otimes y = {\rm id}_{f\otimes y} \ \ \ {\rm and} \ \ \ x\otimes {\rm id}_g = {\rm id}_{x\otimes g},
\end{equation}
\begin{equation}
  \begin{tabular}{ccc}
     \begin{tikzcd}
    x \otimes y
    \ar[
      rr,
      bend left=40,
      "{f_1 \otimes y}"{description, name=s}
    ]
    \ar[
      rr,
      bend right=40,
      "{f_2 \otimes y}"{description, name=t}
    ]
    \ar[
      from=s, to=t,
      Rightarrow,
      "{ \alpha \otimes y }"
    ]
    &&
    x' \otimes y
    \ar[
      rr,
      bend left=40,
      "{f_1' \otimes y}"{description, name=s}
    ]
    \ar[
      rr,
      bend right=40,
      "{f_2' \otimes y}"{description, name=t}
    ]
    \ar[
      from=s, to=t,
      Rightarrow,
      "{ \alpha' \otimes y }"
    ]
    &&
    x'' \otimes y
  \end{tikzcd}    &  $=$ & \begin{tikzcd}
       x \otimes y
    \ar[
      rr,
      bend left=40,
      "{(f_1'f_1) \otimes y}"{description, name=s}
    ]
    \ar[
      rr,
      bend right=40,
      "{(f_2'f_2) \otimes y}"{description, name=t}
    ]
    \ar[
      from=s, to=t,
      Rightarrow,
      "{(\alpha' \star \alpha) \otimes y }"
    ]
    &&
  \ \ \ \ \ \  x'' \otimes y
  \end{tikzcd} 
  \end{tabular}
,\end{equation} 
\begin{equation}
  \begin{tabular}{ccc}
     \begin{tikzcd}
    x \otimes y
    \ar[
      rr,
      bend left=40,
      "{x \otimes g_1}"{description, name=s}
    ]
    \ar[
      rr,
      bend right=40,
      "{x \otimes g_2}"{description, name=t}
    ]
    \ar[
      from=s, to=t,
      Rightarrow,
      "{ x \otimes \beta }"
    ]
    &&
    x \otimes y'
    \ar[
      rr,
      bend left=40,
      "{x \otimes g_1'}"{description, name=s}
    ]
    \ar[
      rr,
      bend right=40,
      "{x \otimes g_2'}"{description, name=t}
    ]
    \ar[
      from=s, to=t,
      Rightarrow,
      "{ x \otimes \beta' }"
    ]
    &&
    x \otimes y''
  \end{tikzcd}    &  $=$ & \begin{tikzcd}
       x \otimes y
    \ar[
      rr,
      bend left=40,
      "{x \otimes (g_1'g_1)}"{description, name=s}
    ]
    \ar[
      rr,
      bend right=40,
      "{x \otimes (g_2'g_2}"{description, name=t}
    ]
    \ar[
      from=s, to=t,
      Rightarrow,
      "{x \otimes (\beta'\star \beta) }"
    ]
    &&
  \ \ \ \ \ \  x \otimes y''
  \end{tikzcd} 
  \end{tabular}
,\end{equation}
subject to the following equivalence relations, closed under vertical composition
\begin{equation}
    \Sigma_{f,g}\circ \Sigma^{-1}_{f,g} \sim {\rm id}_{(x'\otimes g)(f\otimes y)}, \: \: \: \: \: \Sigma^{-1}_{f,g} \circ \Sigma_{f,g} \sim {\rm id}_{(f\otimes y')(x\otimes g)},
\end{equation}
\begin{equation}
    \begin{tabular}{ccc}
     \begin{tikzcd}
    x \otimes y
    \ar[
      rr,
      bend left=40,
      "{f_1 \otimes y}"{description, name=s}
    ]
    \ar[
      rr,
      bend right=40,
      "{f_3 \otimes y}"{description, name=t2}
    ]
    \ar[
    rr,
    "{f_2 \otimes y}"
    {description, name = t1}
    ]
    \ar[
      from=s, to=t1,
      Rightarrow,
      "{ \alpha \otimes y }"
    ]
     \ar[
      from=t1, to=t2,
      Rightarrow,
      "{ \alpha' \otimes y }"
    ]
    &&
    x' \otimes y
  \end{tikzcd}       & $\sim$ &  \begin{tikzcd}
    x \otimes y
    \ar[
      rr,
      bend left=40,
      "{f_1 \otimes y}"{description, name=s}
    ]
    \ar[
      rr,
      bend right=40,
      "{f_3 \otimes y}"{description, name=t2}
    ]
    \ar[
      from=s, to=t2,
      Rightarrow,
      "{ \alpha'\alpha \otimes y }"
    ]
    &&
    \ x' \otimes y
    \end{tikzcd}
    \end{tabular}
,\end{equation}
\begin{equation}
    \begin{tabular}{ccc}
     \begin{tikzcd}
    x \otimes y
    \ar[
      rr,
      bend left=40,
      "{x \otimes g_1}"{description, name=s}
    ]
    \ar[
      rr,
      bend right=40,
      "{x \otimes g_3}"{description, name=t2}
    ]
    \ar[
    rr,
    "{x \otimes g_2}"
    {description, name = t1}
    ]
    \ar[
      from=s, to=t1,
      Rightarrow,
      "{ x \otimes \beta }"
    ]
     \ar[
      from=t1, to=t2,
      Rightarrow,
      "{ x \otimes \beta' }"
    ]
    &&
    x \otimes y'
  \end{tikzcd}       
  &  $\sim$ &  
  \begin{tikzcd}
    x \otimes y
    \ar[
      rr,
      bend left=40,
      "{x \otimes g_1}"{description, name=s}
    ]
    \ar[
      rr,
      bend right=40,
      "{x \otimes g_3}"{description, name=t2}
    ]
    \ar[
      from=s, to=t2,
      Rightarrow,
      "{ x \otimes \beta'\beta }"
    ]
    &&
    \ x \otimes y'
    \end{tikzcd}
    \end{tabular}
,\end{equation}
\begin{equation}
    \left( \Sigma_{f',g} \star \left( {\rm id}_f \otimes y \right)\right)\left(\left( {\rm id}_{f'} \otimes y'\right) \star \Sigma_{f,g} \right) \sim \Sigma_{f'f,g},
\end{equation}
\begin{equation}
    \left(\left( x' \otimes {\rm id}_{g'}\right) \star \Sigma_{f,g}\right) \sim \left( \Sigma_{f,g'} \star \left( x \otimes {\rm id}_g \right) \right) \sim \Sigma_{f, g'g},
\end{equation}
\begin{gather}
    \left( \left( x'' \otimes {\rm id}_{g'} \right) \star \left( {\rm id}_{f'} \otimes y' \right) \star \Sigma_{f,g} \right) \left( \Sigma_{f',g'} \star \left( {\rm id}_{f} \otimes y' \right) \star \left( x \otimes  {\rm id}_g \right) \right) \sim
    \\
    \left( \Sigma_{f',g'} \star \left( x' \otimes {\rm id}_g \right) \star \left( {\rm id}_f \otimes y \right) \right) \left( \left( {\rm id}_{f'} \otimes y'' \right) \star \left( x' \otimes {\rm id}_{g'}\right) \star \Sigma_{f,g} \right),\nonumber
\end{gather}
\begin{equation}
    \left(\left( x' \otimes \beta\right) \star \left( \alpha \otimes y \right) \right) \Sigma_{f,g} \sim \Sigma_{f',g'} \left( \left( \alpha \otimes y' \right) \star \left( x \otimes \beta \right) \right),
\end{equation}
\begin{equation}
    \left( {\rm id} \star \Sigma \right)\left(\Sigma' \star {\rm id}\right) \sim \left( \Sigma' \star \Sigma\right) \sim \left(\Sigma' \star {\rm id}\right)\left({\rm id}\star \Sigma \right),
\end{equation}
for $\Sigma, \Sigma'$ composable generating $2$-cells of the kind (\ref{proto-2-cells1}, \ref{proto-2-cells2}).
    \end{itemize}

\end{definition}

\section{Background to String $2$-groups}\label{app:string2gps}

In this appendix, we collect known mathematical results on String $2$-groups used in Section~\ref{sec:string2groups}.

\subsection{$\String{G}$ as a higher central extension}\label{sapp:stringextension}
Here, we recall how String groups are understood as smooth higher central extensions. 

Briefly, group extensions (see e.g. \cite{Weibel_1994}) are exact sequences of groups
\begin{equation}
    1 \to K \to \Gamma \to G \to 1.
\end{equation}
When $K\leq Z(\Gamma)$  is central in $\Gamma$, such extensions are called central and are classified by the second group cohomology group $H^2_{\rm gp}(G,K)$ for $K$ equipped with the trivial $G$-module structure. Topologically, this cohomology group can be computed as
\begin{equation}
    H^2_{\rm gp}(G,K) \cong \pi_0 \left(\text{Map}\left(BG,B^2K\right)\right) \cong H^2_{\rm sing}(BG,K).
\end{equation}
Similarly, elements of the third cohomology group $H^3(G,K)$, which are likewise computed as 
\begin{equation}
    H^3_{\rm gp}(G,K) \cong H^3_{\rm sing}(BG,K),
\end{equation}
classify sequences of the form
\begin{equation}\label{eq:xmod}
    1\to K \to \Gamma_1 \to \Gamma_0\to G\to 1.
\end{equation}
Since $\Gamma:= (\Gamma_1\to \Gamma_0)$ forms a crossed module, the sequence (\ref{eq:xmod}) is equivalently understood as a central extension of 2-groups
\begin{equation}
    1 \to BK \to \Gamma \to G \to 1,
\end{equation}
where $BK$ is the delooping of $K$ regarded as a $2$-group, and $G$ is a discrete $2$-group, meaning it does not have nontrivial morphisms.

It is known (e.g. \cite{mimura1991topology})  that for $G$ a compact connected Lie group, it satisfies
\begin{equation}
    H^3_{\rm sing}(BG,U(1))\cong H^4_{\rm sing}(BG,\Z),
\end{equation}
so that the fourth integral cohomology group classifies the $2$-group extensions of $G$ as a topological group by $\textbf{B}U(1)$. If $G$ is, furthermore, simple, then
\begin{equation}
    H^3_{\rm sing}(BG,U(1))\cong H^4_{\rm sing}(BG,\Z)\cong \Z.
\end{equation}

The $2$-group extension admits a smooth refinement \cite{brylinski2000differentiablecohomologygaugegroups,segalcohomology,schommer2011central,FSS12} (review in \cite{Kang:2023uvm}) to a \textit{smooth 2-group central extension}
\begin{equation}
    1\to \textbf{B}U(1) \to \String{G_k}\to G_k \to 1 \ \  {\rm (smooth)}
\end{equation}
where the notation $G_k$ means the Lie group $G$ along with the extension class 
\begin{equation}\label{eq:smooth3cohomology}
k\in H^3_{\rm sm}(G,U(1))\cong \pi_0\left(\textbf{H}\left(\textbf{B}G,\textbf{B}^3U(1)\right)\right)
\end{equation}
in smooth group cohomology. The resulting extension $\String{G_k}$ is called the \textit{string $2$-group of} $G_k$. 

In particular, if $G$ is simple, $H^3(G,U(1))\cong \Z$ and thus the extension class $k$ is called the \textit{level}.

\subsection{The center of String groups}\label{sapp:centerstring}

The center of $\String{G_k}$, a nontrivial smooth $2$-group, is denoted as ${\cal Z}({\rm String}(G_k))$ and is generally a \textit{braided} $2$-group. Regarding a smooth $2$-group $\cal G$ as a smooth groupoid $\underline{\cal G}$ equipped with a smooth composition law $\otimes: \underline{\cal G}\times \underline{\cal G}\to \underline{\cal G}$ under which all objects $X\in {\rm ob}(\underline{\cal G})$ are invertible, a braided $2$-group is a $2$-group $\underline{\cal G}$ equipped with a family of natural isomorphisms 
\begin{equation}
    b_{-,X}: -\otimes X\xrightarrow{\sim} X \otimes -,
\end{equation}
satisfying the hexagon condition (see e.g. \cite{street2006characterization} for the relevant identities).

In \cite{2022arXiv220201271W}, the center of string groups is computed in great generality, which we use as a starting point for two concrete examples.

\subsubsection{$G$ compact simple simply-connected}

\begin{fact*}[Braided center of $\String{G_k}$ for $G$ compact simple simply-connected) \cite{2022arXiv220201271W}]\label{fact:stringcentersimple}

Let $G$ be a compact simple simply-connected Lie group, and $k\in \Z \cong H^3_{\rm sm}(G,U(1))$ a smooth cohomology class which defines the higher central extension
\begin{equation}
    1 \to \textbf{B}U(1) \to \String{G_k}\to G_k \to 1.
\end{equation}
The \textit{(Drinfeld) center} $Z(\String{G_k})$ of the string smooth 2-group $\String{G_k}$ is a smooth groupoid ${\cal Z}(G,k):=\underline{Z(\String{G_k})}$ with
\begin{eqnarray}
    \pi_0\left({\cal Z}(G,k)\right) &=& Z(G),
    \\
    \pi_1\left({\cal Z}(G,k)\right) &=& U(1),
\end{eqnarray}
whose braiding is determined by the quadratic form
\begin{eqnarray}
    q&:& Z(G) \to U(1),\label{eq:braiding1}
    \\
    && z \mapsto \exp \tfrac{k}{2}I(\bar{z},\bar{z}), \label{eq:braiding2}
\end{eqnarray}
where $I(-,-):{\frak g}_{\C}\otimes {\frak g}_{\C}\to \C$ is the
unique smallest $\text{Ad}_G$-invariant positive-definite form satisfying $I(z, z) \in 2\Z$ for all coroots of ${\frak g}_{\C}$ the complexification of $\frak g$ the Lie algebra of $G$, and $\bar{z}$ is a lift of $Z(G)$ to ${\frak g}_{\C}$.
\end{fact*}

\subsubsection{$G=T$}

\begin{fact*}[Braided center of $\String{T_k}$ \cite{2022arXiv220201271W}]

Let $G= T=U(1)^r$ be the torus of rank $r$, $\Lambda={\rm Hom}(T,U(1))$ its group of characters, and $\Pi=\Lambda^*={\rm Hom}(\Lambda,U(1))$ its group of cocharacters. Let $J:{\frak t}\times {\frak t}\to \R$ the $\Z$-bilinear form satisfying the integrality condition that determines the extension class $k\in H^4(T,\Z)$ of its string $2$-group $\String{T_k}$. Its center ${\cal Z}_k:={\cal Z}(\String{T_k})$ is the braided smooth $2$-group presented as the crossed module of Lie groups (cf. Appendix~\ref{sapp:lie2gps})
\begin{eqnarray}
    \delta&:& \Pi\times U(1) \to {\frak t}\oplus \Lambda,
    \\
    && (m,z) \mapsto (m,\tau(m))
\end{eqnarray}
for
\begin{equation}
    \tau(m) := -\left(J(m,-)+J(-,m)\right),
\end{equation}
 action
\begin{eqnarray}
    \alpha:& \left({\frak t}\oplus \Lambda \right)\times \left(\Pi\times U(1) \right) &\to \left(\Pi\times U(1) \right),
    \\
    & \left( (x,\lambda),(m,z)\right) &\mapsto  \left(m,z\exp\left(J(m,x)\right)\right),
\end{eqnarray}
and braiding
\begin{equation}
    \beta_{[x,\lambda],[x',\lambda']}:= \lambda(x')\exp\left(J(x',x)\right).
\end{equation}

This center fits in the exact sequence of smooth $2$-groups
\begin{equation}\label{eq:zkextension}
    1\to \textbf{B}U(1) \to {\cal Z}_k \to \frac{{\frak t}\oplus \Lambda}{\Pi} \to 1,
\end{equation}
where the rightmost group is
\begin{equation}\label{eq:pi0zk}
    \frac{{\frak t}\oplus \Lambda}{\Pi} = {\rm coker}\,\delta = \pi_0\left({\cal Z}_k\right).
\end{equation}

\end{fact*}

\section{Principal $3$-group bundles}\label{app:principal-3-bundles}

In this appendix, we collect the definitions necessary to define principal bundles of smooth $3$-groups presented as $2$-crossed modules of Lie groups.

\subsection{Crossed $2$-modules of Lie groups}

Certain smooth $3$-groups allow a presentation in terms of crossed $2$-modules of Lie groups. In the following, all groups involved are Lie groups, and the group homomorphisms are Lie group homomorphims.

A crossed $2$-module
\begin{equation}
    {\cal G} = \left(G,H,L,\partial_1,\partial_2, \triangleright_1, \triangleright_2,\{-,-\} \right)
\end{equation}
consists of \cite{baues1991combinatorial}, \cite[Appendix B]{Gagliardo:2025oio}: three Lie groups $G,H,L$, homomorphisms
\begin{equation}
    \partial_1: H\to G,
\end{equation}
\begin{equation}
    \partial_2: L \to H,
\end{equation}
such that the composition $\partial_1\circ \partial_2 = 0$ is the trivial homomorphism, actions
\begin{equation}
  \triangleright_1: G \times H \to H,  
\end{equation}
\begin{equation}
    \triangleright_2: G \times L \to L,
\end{equation}
and a homomorphism
\begin{equation}
    \{-,-\}: H\times H \to L,
\end{equation}
satisfying the following identities
\begin{eqnarray}
    \partial_1\left(g\triangleright_1 h \right) &=& g \partial_1 \left( h\right) g^{-1} ,
\\
   \partial_2\left(g\triangleright_2 \ell \right) &=& g\triangleright_1 \partial_2 \left( \ell\right),
   \\
    g\triangleright_2 \{ h_1,h_2\} &=& \{ g\triangleright_1 h_1, g\triangleright_1 h_2 \},
\\
    \partial_2\left( \{h_1,h_2\} \right) &=&  h_1h_2 h_1^{-1} \left(\partial_1(h_1)\triangleright_1 h_2^{-1} \right),
\\
  \{\partial_2\left( \ell_1\right),\partial_2\left(\ell_2 \right) \}  &=& \ell_1\ell_2\ell_1^{-1}\ell_2^{-1}, \label{eq:2xmod-bracket-identity}
    \\
    \{h_1h_2,h_3\} &=& \{ h_1,h_2 h_3 h_2^{-1}\} \left(\partial_1\left( h_1\right)\triangleright_2\{h_2,h_3\} \right),
\\
  \{h_1,h_2h_3\}  &=& \{h_1,h_2\} \{h_1,h_3\} \{ \left(  h_1h_3 h_1^{-1} \left(\partial_1(h_1)\triangleright_1 h_3^{-1}\right)\right)^{-1}, \partial_1\left(h_1\right)\triangleright_1 h_2\}
\\
  \{ \partial_2\left(\ell\right),h\}\{h,\partial_2\left(\ell\right)\}  &=& \ell \left( \partial_1\left(h\right) \triangleright_2 \ell^{-1}\right).
\end{eqnarray}

There is, furthermore, an induced $H$-action on $L$ defined as
\begin{eqnarray}\label{eq:3gp-haction}
    \triangleright:& H\times L & \to L,
    \\
    & (h,\ell) & \mapsto \ell \{\left(\partial_2\left(\ell\right)\right)^{-1},h \}.
\end{eqnarray}

\subsection{Čech data for principal $3$-bundle}

Let ${\cal G}$ be a smooth $3$-group presented as a $2$-crossed module of Lie groups $\left(G,H,L,\partial_1,\partial_2, \triangleright_1, \triangleright_2,\{-,-\} \right)$. For $\{U_i\}_{i \in {\cal I}}$ a good cover of a smooth manifold $\Sigma$, a principal $\cal G$ bundle consists of \cite[Section 3]{Gagliardo:2025oio} cochains
\begin{eqnarray}
    g_{ij}: U_{ij} \to G,
    \ \
    h_{ijk} : U_{ijk} \to H,
    \ \
    \ell_{ijkl} : U_{ijkl} \to L,
\end{eqnarray}
satisfying the identities
\begin{eqnarray}
    g_{ik} &=& \partial_1\left(h_{ijk} \right)g_{ij} g_{jk},
    \\
    h_{ikl} h_{ijk} \partial_2\left(\ell_{ijkl}\right) &=& h_{ijl}\left(g_{ij}\triangleright_1 h_{jkl}\right),
\end{eqnarray}
\begin{gather}
    \ell_{ijkl} \left(\left(g_{ij}\triangleright_1 h_{jkl}^{-1}\right)\triangleright \ell_{ijlm}\right)\left(g_{ij}\triangleright_2\ell_{jklm} \right) 
    \\= \left(h^{-1}_{ijk}\triangleright \ell_{iklm} \right)\{h^{-1}_{ijk},g_{ik}\triangleright_1 h^{-1}_{klm} \} \left(\left(g_{ij}g_{jk}\triangleright_1h^{-1}_{klm}\right) \triangleright \ell_{ijkm}\right) \nonumber.
\end{gather}

Given a pair of $\cal G$ bundles $(g_{ij},h_{ijk},\ell_{ijkl})$, $(g'_{ij},h'_{ijk},\ell'_{ijkl})$, a \textit{gauge transformation}
\begin{equation}
\left(g_i,h_{ij},\ell_{ijk}\right): \left(g_{ij},h_{ijk},\ell_{ijkl}\right) \xrightarrow{\sim} \left(g'_{ij},h'_{ijk},\ell'_{ijkl}\right)
\end{equation}
consists of cochains
\begin{eqnarray}
    g_{i}: U_{i} \to G,
    \ \
    h_{ij} : U_{ij} \to H,
    \ \
    \ell_{ijk} : U_{ijk} \to L,
\end{eqnarray}
satisfying the identities
\begin{equation}
    g'_{ij} = g_i^{-1} \partial_1\left( h_{ij}\right)g_{ij}g_j,
\end{equation}
\begin{equation}
    h'_{ijk} = g_i^{-1} \triangleright_1 \left(h_{ik}h_{ijk}\partial_2\left(\ell_{ijk}\right)\left(g_{ij}\triangleright_1 h_{jk}^{-1} \right) h_{ij}^{-1}\right),
\end{equation}
\begin{gather}
 \ell'_{ijkl} = g_i^{-1} \triangleright_2 \Big(\left( h_{ij}g_{ij}\triangleright_1 h_{jk}\right)\triangleright \big( \ell_{ijk}^{-1} h_{ijk}^{-1}\triangleright\{h_{ijk},\left(g_{ij}g_{jk}\right)\triangleright_1 h_{kl}\}^{-1} \label{eq:ell-gauge-transformation}
 \\
 \times \left(\left(g_{ij}g_{jk} \right)\triangleright_1 h_{kl} \right)\triangleright \Big( \left(h_{ijk}^{-1}\triangleright \ell_{ikl}^{-1} \right)\ell_{ijkl}\nonumber
 \\ 
 \times \left( g_{ij}\triangleright_1 h_{jkl}^{-1}\right)\triangleright \Big(\ell_{ijl}\left( \left( g_{ij}\triangleright_1 h_{jl}^{-1}\right)h_{ij}^{-1}\right)\triangleright \{h_{ij},g_{ij}\triangleright h_{jl} \}^{-1} \nonumber
 \\
 \times h_{ij}^{-1}\triangleright \{ h_{ij},g_{ij}\triangleright_1 h_{jkl}\}^{-1}\Big)\nonumber
 \\
 \times \left(h_{ij}^{-1}\triangleright\{h_{ij},\partial_2\left(g_{ij}\triangleright_2 \ell_{jkl}\right) \}^{-1} \right) \left( g_{ij}\triangleright_2 \ell_{jkl}\right)\left(h_{ij}^{-1} \triangleright \{h_{ij},\left(g_{ij}g_{jk}\right)\triangleright_1 h_{kl}^{-1}\}^{-1} \right)\nonumber
 \\
 \times h_{ij}^{-1} \triangleright \{ h_{ij}, g_{ij}\triangleright_1 h^{-1}_{jl}\}^{-1}\Big)
 \Big)    . \nonumber
\end{gather}

\bibliographystyle{utphys}
\bibliography{draft}

\end{document}